\documentclass[11pt]{article}

\usepackage[usenames,dvipsnames,svgnames,table]{xcolor}
\definecolor{myOrange}{RGB}{252,141,89}
\definecolor{myBlue}{RGB}{145,191,219}
\definecolor{myGreen}{RGB}{145,207,96}

\definecolor{myLightGray}{RGB}{210,210,210}
\definecolor{myGray}{RGB}{150,150,150}

\definecolor{li}{HTML}{00677C}
\definecolor{wi}{HTML}{8E217D}
\definecolor{ti}{HTML}{F29400}
\definecolor{si}{HTML}{E43117}
\definecolor{vi}{HTML}{39842E}
\definecolor{hi}{HTML}{9B0a7d}

\usepackage{amssymb}
\usepackage{amsthm}
\usepackage{amsmath}
\usepackage{stmaryrd}
\usepackage{mathtools}
\usepackage{graphicx}
\usepackage[utf8]{inputenc}
\usepackage{caption}
\usepackage{subcaption}
\usepackage{acro}
\usepackage{hyphenat}
\usepackage{dsfont}
\usepackage{nicefrac}
\newcolumntype{L}{>{$}l<{$}} 

\usepackage{subcaption}
\usepackage{todonotes}
\usepackage{ifthen}
\usepackage{enumitem}
\usepackage{xspace}

\usepackage{geometry}
\usepackage[colorlinks, citecolor=OliveGreen, urlcolor=black, linkcolor=black]{hyperref}

\usepackage{tikz}
\usetikzlibrary{patterns}
\usetikzlibrary{decorations,arrows}
\usetikzlibrary{decorations.pathmorphing}
\usepgflibrary{decorations.pathreplacing} 
\usetikzlibrary{positioning}
\newlength{\hatchspread}
\newlength{\hatchthickness}
\newlength{\hatchshift}
\newcommand{\hatchcolor}{}
\tikzset{hatchspread/.code={\setlength{\hatchspread}{#1}},
	hatchthickness/.code={\setlength{\hatchthickness}{#1}},
	hatchshift/.code={\setlength{\hatchshift}{#1}},
	hatchcolor/.code={\renewcommand{\hatchcolor}{#1}}}
\tikzset{hatchspread=3pt,
	hatchthickness=0.4pt,
	hatchshift=0pt,
	hatchcolor=black}
\pgfdeclarepatternformonly[\hatchspread,\hatchthickness,\hatchshift,\hatchcolor]
{custom north west lines}
{\pgfqpoint{\dimexpr-2\hatchthickness}{\dimexpr-2\hatchthickness}}
{\pgfqpoint{\dimexpr\hatchspread+2\hatchthickness}{\dimexpr\hatchspread+2\hatchthickness}}
{\pgfqpoint{\dimexpr\hatchspread}{\dimexpr\hatchspread}}
{
	\pgfsetlinewidth{\hatchthickness}
	\pgfpathmoveto{\pgfqpoint{0pt}{\dimexpr\hatchspread+\hatchshift}}
	\pgfpathlineto{\pgfqpoint{\dimexpr\hatchspread+0.15pt+\hatchshift}{-0.15pt}}
	\ifdim \hatchshift > 0pt
	\pgfpathmoveto{\pgfqpoint{0pt}{\hatchshift}}
	\pgfpathlineto{\pgfqpoint{\dimexpr0.15pt+\hatchshift}{-0.15pt}}
	\fi
	\pgfsetstrokecolor{\hatchcolor}
	\pgfusepath{stroke}
}

\pgfdeclarepatternformonly[\hatchspread,\hatchthickness,\hatchshift,\hatchcolor]
{custom north east lines}
{\pgfqpoint{\dimexpr-2\hatchthickness}{\dimexpr-2\hatchthickness}}
{\pgfqpoint{\dimexpr\hatchspread+2\hatchthickness}{\dimexpr\hatchspread+2\hatchthickness}}
{\pgfqpoint{\dimexpr\hatchspread}{\dimexpr\hatchspread}}
{
	\pgfsetlinewidth{\hatchthickness}
	\pgfpathmoveto{\pgfqpoint{\dimexpr\hatchshift-0.15pt}{-0.15pt}}
	\pgfpathlineto{\pgfqpoint{\dimexpr\hatchspread+0.15pt}{\dimexpr\hatchspread-\hatchshift+0.15pt}}
	\ifdim \hatchshift > 0pt
	\pgfpathmoveto{\pgfqpoint{-0.15pt}{\dimexpr\hatchspread-\hatchshift-0.15pt}}
	\pgfpathlineto{\pgfqpoint{\dimexpr\hatchshift+0.15pt}{\dimexpr\hatchspread+0.15pt}}
	\fi
	\pgfsetstrokecolor{\hatchcolor}
	\pgfusepath{stroke}
}

\newcommand\xqed[1]{%
	\leavevmode\unskip\penalty9999 \hbox{}\nobreak\hfill
	\quad\hbox{#1}}
\newcommand\demo{\xqed{$\lhd$}}

\newenvironment{proofClaim}{\noindent\textit{Proof.~}}{\demo \par%
	\addvspace{\baselineskip}
}

\newcommand{\sset}[2]{\{#1\,|\,#2\}}

\newcommand{\eps}{\varepsilon}
\newcommand{\OPT}{\mathrm{OPT}}
\newcommand{\jobs}{\mathcal{J}}

\newcommand{\conf}{\mathcal{C}}

\newcommand{\NumGridpoints}{N}

\newcommand{\items}{\mathcal{I}}
\newcommand{\itemsL}{\mathcal{L}}
\newcommand{\itemsH}{\mathcal{H}}
\newcommand{\itemsT}{\mathcal{T}}
\newcommand{\itemsV}{\mathcal{V}}
\newcommand{\itemsS}{\mathcal{S}}
\newcommand{\itemsM}{\mathcal{M}}
\newcommand{\itemsMV}{\mathcal{M}_V}

\newcommand{\iBox}{\mathcal{B}}
\newcommand{\iBoxH}{\iBox_{\itemsH}}
\newcommand{\iBoxTV}{\iBox_{\itemsT\cup\itemsV}}
\newcommand{\iBoxL}{\iBox_{\itemsL}}
\newcommand{\iBoxV}{\iBox_{\itemsV}}
\newcommand{\iBoxT}{\iBox_{\itemsT}}
\newcommand{\iBoxP}{\iBox_{\mathcal{P}}}
\newcommand{\iBoxS}{\iBox_{\mathcal{S}}}

\newcommand{\inn}{\mathrm{inn}}

\newcommand{\numBox}{N_B}
\newcommand{\numFrac}{N_F}
\newcommand{\numStrips}{N_S}

\newcommand{\NN}{\mathds{N}}

\newcommand{\RR}{\mathds{R}}
\newcommand{\Oh}{\mathcal{O}}

\newcommand{\NP}{\mathrm{NP}}
\newcommand{\Poly}{\mathrm{P}}

\newcommand{\iW}[1]{w(#1)}
\newcommand{\iH}[1]{h(#1)}
\newcommand{\areaI}[1]{a(#1)}

\newcommand{\subbox}{sub\hyp{}box\xspace}
\newcommand{\subboxes}{sub\hyp{}boxes\xspace}

\newboolean{BlackAndWhite}
\setboolean{BlackAndWhite}{true}
\newcommand{\drawVerticalItem}[5][$ $]{%
	\ifthenelse{\boolean{BlackAndWhite}}{%
		\draw[fill = white!85!black, fill opacity = 0.7] (#2,#3) rectangle node[midway, opacity = 1]{#1} (#4,#5)}{%
		\draw[fill = white!50!vi, fill opacity = 0.7] (#2,#3) rectangle node[midway, opacity = 1]{#1} (#4,#5)}}
\newcommand{\drawTallItem}[5][$ $]{%
	\ifthenelse{\boolean{BlackAndWhite}}{%
		\draw[fill = white!65!black, fill opacity = 0.7] (#2,#3) rectangle node[midway, opacity = 1]{#1} (#4,#5)}{%
		\draw[fill = white!50!ti, fill opacity = 0.7] (#2,#3) rectangle node[midway, opacity = 1]{#1} (#4,#5)}}
\newcommand{\drawLargeItem}[5][$ $]{	\ifthenelse{\boolean{BlackAndWhite}}{%
		\draw[fill = white!55!black, fill opacity = 0.7] (#2,#3) rectangle node[midway, opacity = 1]{#1} (#4,#5)}{%
		\draw[fill = white!50!li, fill opacity = 0.7] (#2,#3) rectangle node[midway, opacity = 1]{#1} (#4,#5)}}
\newcommand{\drawSmallItem}[5][$ $]{	\ifthenelse{\boolean{BlackAndWhite}}{%
		\draw[fill = white!95!black, fill opacity = 0.7] (#2,#3) rectangle node[midway, opacity = 1]{#1} (#4,#5)}{%
		\draw[fill = white!50!si, fill opacity = 0.7] (#2,#3) rectangle node[midway, opacity = 1]{#1} (#4,#5)}}
\newcommand{\drawHorizontalItem}[5][$ $]{	\ifthenelse{\boolean{BlackAndWhite}}{%
		\draw[fill = white!75!black, fill opacity = 0.7] (#2,#3) rectangle node[midway, opacity = 1]{#1} (#4,#5)}{%
		\draw[fill = white!50!hi, fill opacity = 0.7] (#2,#3) rectangle node[midway, opacity = 1]{#1} (#4,#5)}}

\def\DEBUG{true} 

\ifdefined\DEBUG

\def\rem#1{{\marginpar{\raggedright\scriptsize #1}}}
\newcommand{\malr}[1]{\rem{\textcolor{myGreen}{$\bullet$ #1}}}
\else

\newcommand{\malr}[1]{}
\fi

\setdescription{leftmargin=0em,labelindent=0\parindent,labelsep = 1em, itemsep=0.1\baselineskip,topsep=0.1\baselineskip,itemsep =\topsep,  listparindent=\parindent, font=\normalfont\normalsize\itshape}

\newcounter{steplistcount}
\newcounter{steplistcounti}
\newcounter{caselistcount}
\newcounter{caselistcounti}
\renewcommand*\thesteplistcount{\arabic{steplistcount}}
\renewcommand*\thesteplistcounti{\arabic{steplistcounti}}
\renewcommand*\thecaselistcount{\arabic{caselistcount}}
\renewcommand*\thecaselistcounti{\arabic{caselistcounti}}
\newlist{stepList}{description}{2}
\setlist[stepList,1]{%
	before={\setcounter{steplistcount}{0}},
	leftmargin=0em,labelindent=0\parindent,labelsep = \parindent, itemsep=0.1\baselineskip,topsep=0.1\baselineskip,itemsep =\topsep,  listparindent=\parindent,style = sameline
	,font=\normalfont\normalsize\itshape Step~\stepcounter{steplistcount}\thesteplistcount:~ 
}
\setlist[stepList,2]{%
	before={\setcounter{steplistcounti}{0}},%
	leftmargin=0em,labelindent=0\parindent,labelsep = \parindent, itemsep=0.1\baselineskip,topsep=0.1\baselineskip,itemsep =\topsep,  listparindent=\parindent,style = sameline%
	,font=\normalfont\normalsize\itshape Step~ \stepcounter{steplistcounti}\thesteplistcount.\thesteplistcounti:~ 
}

\newlist{caseList}{description}{2}
\setlist[caseList,1]{%
	before={\setcounter{caselistcount}{0}},%
	leftmargin=0em,labelindent=0\parindent,labelsep = \parindent, itemsep=0.1\baselineskip,topsep=0.1\baselineskip,itemsep =\topsep, listparindent=\parindent,style = sameline%
	,font=\normalfont\normalsize\itshape Case~\stepcounter{caselistcount}\thecaselistcount:~
}

\setlist[caseList,2]{%
	before={\setcounter{caselistcounti}{0}},%
	leftmargin=0em,labelindent=0\parindent,labelsep = \parindent, itemsep=0.1\baselineskip,topsep=0.1\baselineskip,itemsep =\topsep,  listparindent=\parindent,style = sameline%
	,font=\normalfont\normalsize\itshape Case~ \stepcounter{caselistcounti}\thecaselistcount.\thecaselistcounti:~ %
}

\theoremstyle{plain}
\newtheorem{lemma}{Lemma}
\newtheorem{theorem}{Theorem}
\newtheorem{corollary}{Corollary}
\theoremstyle{remark}
\newtheorem*{claim}{Claim}
\newtheorem{observation}{Observation}
\newtheorem*{remark*}{Remark}

\title{Closing the gap for pseudo-polynomial strip packing
	\thanks{Research was supported by German Research Foundation (DFG) project JA 612 /20-1}
}
\author{Klaus Jansen, Malin Rau\\
	Institut für Informatik, Christian-Albrechts-Universität zu Kiel, Germany\\
	\texttt{$\{$kj,mra$\}$@informatik.uni-kiel.de}
}
\date{}

\begin{document}

\maketitle

\begin{abstract}
	The set of 2-dimensional packing problems builds an important class of optimization problems and Strip Packing together with 2-dimensional Bin Packing and 2-dimensional Knapsack is one of the most famous of these problems.
	Given a set of rectangular axis parallel items and a strip with bounded width and infinite height the objective is to find a packing of the items into the strip which minimizes the packing height. 
	We speak of pseudo-polynomial Strip Packing if we consider algorithms with pseudo-polynomial running time with respect to the width of the strip. 
	
	It is known that there is no pseudo-polynomial algorithm for Strip Packing with a ratio better than $5/4$ unless $\Poly=\NP$. 
	The best algorithm so far has a ratio of $4/3 +\varepsilon$. 
	In this paper, we close this gap between inapproximability result and best known algorithm  by presenting an algorithm with approximation ratio $5/4 + \varepsilon$ and thus categorize the problem accurately. 
	The algorithm uses a structural result which states that each optimal solution can be transformed such that it has one of a polynomial number of different forms. 
	The strength of this structural result is that it applies to other problem settings as well for example to Strip Packing with rotations (90 degrees) and Contiguous Moldable Task Scheduling.
	This fact enabled us to present algorithms with approximation ratio $5/4 + \varepsilon$ for these problems as well.
\end{abstract}

\section{Introduction}
In the Strip Packing problem, we have to pack a set $\items$ of rectangular items into a given strip with width $W \in \NN$ and infinite height. 
Each item $i\in \items$ has a width $\iW{i} \in \NN_{\leq W}$ and a height $\iH{i} \in \NN$. 
The area of an item $i \in \items$ is defined as $\areaI{i} :=\iH{i}\cdot \iW{i}$ and the area of a set of items $\items' \subseteq \items$ is defined as $\areaI{\items'} := \sum_{i \in \items'}\iH{i}\cdot \iW{i}$.

A packing of the items is given by a mapping $\rho: I \rightarrow \NN_{\leq W} \times \NN, i \mapsto (x_i,y_i)$ which assigns the lower left corner of an item $i \in I $ to a position $\rho(i) = (x_i,y_i)$ in the strip. 
An inner point of $i \in \items$ (with respect to a packing $\rho$) is a point from the set $\inn (i) := \sset{(x,y) \in \RR\times\RR}{x_i < x < x_i + \iW{i}, y_i < y < y_i + \iH{i}}$. 
We say two items $i,j \in \items$ overlap if they share an inner point (i.e., $\inn(i) \cap \inn(j) \not = \emptyset$). 
A packing is feasible if no two items overlap and if $x_i + w(i) \leq W$ for all $i \in I$. 
The objective of the Strip Packing problem is to find a feasible packing $\rho$ with minimal height $\iH{\rho} := \max\sset{y_i + \iH{i}}{i \in \items, \rho(i) = (x_i,y_i)}$.
In the following given an instance $I$ of the Strip Packing problem, we will denote this minimal packing height with $\OPT(I)$ and dismiss the $I$ if the instance is clear from the context.

In this paper, we study pseudo-polynomial approximation algorithms with respect to the width of the strip, i.e., we consider algorithms where the width of the strip $W$ is allowed to appear polynomially in the running time.
Recently, we were able to show that we cannot find an algorithm with approximation ratio strictly better than $5/4$ unless $\Poly = \NP$~\cite{HenningJRS18}. 
On the other hand, the algorithm with the best ratio so far computes a $4/3 +\eps$ approximation~\cite{GalvezGIK16, JansenR16}. 
This yields a large gap and it was unknown whether the ratio of the algorithm or the lower bound was tight.
We manage to close this gap and thus categorize the complexity of the problem correctly.
We accomplish this by proving a strong result about the structure of optimal solutions, which enables us to close the gap to the lower bound, except for a negligibly small $\eps$.

\begin{theorem}
	\label{thm:StripPacking}
	There is a pseudo-polynomial algorithm for Strip Packing which finds a $(5/4 +\eps)$-approximation in $\mathcal{O}(n\log(n))\cdot W^{\mathcal{O}_{\eps}(1)}$ operations. 
\end{theorem}

The structural result applies also to other problem settings and, therefore, the algorithmic result can be extended to them. One example is the setting of Strip Packing where we are allowed to rotate the items by 90 degrees. In this setting, the items still have to be placed axis-aligned, but we can decide if the longer or shorter side defines the height of the item.

\begin{theorem}
	\label{thm:StripPackingRotation}
	There is a pseudo-polynomial algorithm for Strip Packing with rotations which finds a $(5/4 +\eps)$-approximation in $(n W)^{\mathcal{O}_{\eps}(1)}$ operations. 
\end{theorem}

A generalization of Strip Packing is the Contiguous Moldable Task Scheduling problem. 
In this setting, we are given a set of jobs $\jobs$ and a set of $m$ machines. 
Each job $j \in \jobs$ can be scheduled on different numbers of machines given by $M_j \subseteq \{1,\dots,m\}$. 
Depending on the number of machines $i \in M_j$, each job $j \in \jobs$ has a specific processing time $p_j(i) \in \NN$.  

A schedule $S$ is given by three functions: $\sigma: \jobs \rightarrow \NN$ which maps each job $j\in \jobs$ to a starting time $\sigma(j)$; $\rho : \jobs \rightarrow \{1,\dots,m\}$ which maps each job $j\in \jobs$ to the number of processors $\rho(j) \in M_j$ it is processed on; and $\varphi: \jobs \rightarrow \{1,\dots,m\}$ which maps each job $j\in \jobs$ to the first machine it is processed on. 
The job $j \in \jobs$ will use the machines $\varphi(j)$ to $\varphi(j) + \rho(j) -1$ contiguously. 
A schedule $S = (\sigma,\rho,\varphi)$ is feasible if each machine processes at most one job at a time 
and its makespan is defined by $\max_{j \in \jobs}\sigma(j) + p_j(\rho(j))$.
The objective is to find a feasible schedule, which minimizes the makespan.

This problem is a true generalization of Strip Packing as it contains this problem (and Strip Packing with rotations) as a special case: 
We define the number of machines $m$ as the width of the strip $W$ and
for each item $i \in \items$ we introduce one job $i$ with $M_i := \{\iW{i}\}$ and processing time $p_i(\iW{i}) = \iH{i}$ (or introduce one job $i$ with $M_i := \{\iW{i},\iH{i}\}$ and processing times $p_i(\iW{i}) = \iH{i}$ and $p_i(\iH{i}) = \iW{i}$ respectively). 
Therefore, we cannot hope for a pseudo-polynomial algorithm with a ratio better than $5/4$ unless $\Poly = \NP$. 
We managed to adapt the algorithmic result to find an algorithm with an approximation ratio, which almost matches this bound.

\begin{theorem}
	\label{thm:Scheduling}
	There is a pseudo-polynomial algorithm for the Contiguous Moldable Parallel Tasks Scheduling Problem which finds a $(5/4 +\eps)$-approximation in $(nm)^{\mathcal{O}_{\eps}(1)}$ operations. 
\end{theorem}

We say a job $j \in \jobs$ is monotone if the work of the job $w(p_j(i)) :=p_j(i) \cdot i$ does not increase if we decrease the number of machines. 
There is an FPTAS by Jansen and Land~\cite{JansenL18} for the case that all jobs are monotonic and $m \geq 8n/\eps$. 
For the case that $m < 8n/\eps$ we greatly improve the previous algorithms by carefully applying our structural result, yielding a polynomial algorithm for the case of monotonic jobs.

\begin{corollary}
	There is a polynomial algorithm for Scheduling Monotonic Moldable Parallel Tasks on Contiguous Machines which finds a $(5/4 +\eps)$-approximation in $n^{\mathcal{O}_{\eps}(1)}$ operations.
\end{corollary}

\subsection*{Related Work}
In this paper, we consider approximation algorithms. 
We say an approximation algorithm $A$ has an (absolute) approximation ratio $\alpha$, if for each instance $I$ of the problem it holds that $A(I) \leq \alpha\OPT(I)$.
if an algorithm $A$ has an approximation ratio of $\alpha$, we say it is result is an $\alpha$-approximation.
Furthermore, a family of algorithms consisting of algorithms with approximation ratio $(1+\eps)$ is called polynomial time approximation scheme (PTAS), and a PTAS whose running time is polynomial in both the input length and $1/\eps$ is called fully polynomial (FPTAS).
If the running time of the approximation scheme is not polynomial but pseudo-polynomial, we denote it as pseudo-PTAS or PPTAS.
An algorithm $A$ has an asymptotic approximation ratio $\alpha$ if there is a constant $c$ such that $A(I) \leq \alpha\OPT(I) + c$
and we denote a polynomial time approximation scheme with respect to the asymptotic approximation ratio as an A(F)PTAS.

\paragraph*{Strip Packing} 
Strip Packing is an important $\NP$-hard problem which has been studied since 1980 (Baker et al.~\cite{BakerCR80}). 
It arises naturally in many settings as scheduling or cutting stock problems in industrial manufacturing (e.g, cutting rectangular pieces out of a sheet of material as cloth or wood).
Recently, it also has been applied to practical problems as electricity allocation and peak demand reductions in smart-grids~\cite{KarbasiounSKL13, RanjanKS17, TangHLW13}. 

In a series of papers~\cite{BakerBK81, BakerCR80, BougeretDJRT2011, CoffmanGJT80, Golan, HarrenS09, JansenS09, Kenyon00, Schiermeyer94, Sleator80, Steinberg97, Sviridenko12} algorithms with improved approximation ratios have been presented and 
$5/3+\eps$ is the best absolute approximation ratio achieved so far by an algorithm presented by Harren, Jansen, Prädel, and van Stee~\cite{HarrenJPS14}. 
On the other hand, by a reduction from the Partition Problem, one can see that it is not possible to find an algorithm with approximation ratio better than $3/2$ unless $\Poly = \NP$.
Therefore, a still open question is whether there is an algorithm with approximation ratio $3/2$ or if the lower bound is larger than $3/2$.

We will use two of these algorithms as subroutines to schedule certain sets of jobs. 
The first is the famous NFDH-Algorithm by Coffman et al.~\cite{CoffmanGJT80}, which finds a packing with the properties from the following (slightly adapted) lemma.
\begin{lemma}[See~\cite{CoffmanGJT80}]
	\label{thm:NFDH}
	For any list $L$ ordered by nonincreasing height it holds that 
	\[\mathrm{NFDH}(L) \leq 2W(L)/m + p_{max} \leq 2\cdot \OPT(L) + p_{max}\text{.}\]
\end{lemma}
The other algorithm is Steinbergs-Algorithm, which we will use to bound the height of an optimal packing from above.
Steinbergs-Algorithm has the following properties:
\begin{lemma}[See~\cite{Steinberg97}]
	\label{lma:Steinberg}
	Let $w_{\max} := \max_{i \in \items} \iW{i}$ and $h_{\max}:= \max_{i \in \items} \iH{i}$.
	If the following inequalities hold, 
	\begin{align*}
	w_{\max} \leq W, \hspace*{2em} h_{\max} \leq H, \hspace*{2em} 2\areaI{\items} \leq WH - (2w_{\max} -W)_+(2h_{\max} - H)_+,
	\end{align*}
	(where $x_+ := \max\{x,0\}$)
	then the items $\items$ can be placed inside a rectangle $Q$ with width $W$ and height $H$. 
\end{lemma}

In contrast to absolute approximation ratios, asymptotic approximation ratios can get better than $3/2$ 
and they  have been improved in a series of papers~\cite{BakerBK81, CoffmanGJT80, Golan}. 
The first asymptotic fully polynomial approximation scheme (in short AFPTAS) was presented by Kenyon and R\'emila~\cite{Kenyon00}. 
It has an additive term of $\mathcal{O}(h_{\max}/\eps^2)$, where $h_{\max}$ is the largest occurring item height. 
The additive term was improved by Sviridenko~\cite{Sviridenko12} and Bougeret et al.~\cite{BougeretDJRT2011} to $\mathcal{O}((\log(1/\eps)/\eps)h_{\max})$ simultaneously. 
Furthermore, Jansen and Solis-Oba~\cite{JansenS09} presented an asymptotic PTAS with an additive term $h_{\max}$ at the expense of the running time. 
Asymptotic algorithms are useful when the maximal occurring item height is small compared to the optimum. 
However, if the maximal occurring height equals the optimum, these algorithms have an approximation ratio of $2$ or even worse. 
This motivates the search for algorithms with better approximation ratios in expense of the processing time.

Strip Packing can be seen as a Scheduling Problem and is denoted by $P|line_j|C_{\max}$ and is sometimes also called scheduling on non-fragmentable
multiprocessor systems~\cite{TurekWY92}. 
We have given $m := W$ machines and have to schedule parallel jobs on the machines contiguously.
In realistic instances, we can hope that the number of machines is moderate, (e.g., bounded by a polynomial in the number of jobs or to be even a small constant).
Therefore, it is reasonable to consider pseudo-polynomial algorithms, where we allow $W$ to appear polynomially in the running time.
The Partition Problem is solvable in pseudo-polynomial time. 
Therefore, the lower bound of $3/2$ for absolute approximation ratios does not hold for pseudo-polynomial algorithms.   
The best approximation ratio has been improved step by step~\cite{JansenT10, NadiradzeW16, GalvezGIK16, JansenR16} and 
$4/3 +\eps$ is the best absolute approximation ratio achieved so far~\cite{GalvezGIK16, JansenR16}. 
On the other hand, we cannot approximate arbitrary in this scenario. 
Adamaszek, Kociumaka, Pilipczuk, and Pilipczuk~\cite{AdamaszekKPP17} proved a lower bound of $12/11$ if $P \not = NP$. 
This lower bound was improved to $5/4$ by Henning, Jansen, Rau, and Schmarje~\cite{HenningJRS18} if $\Poly \not = \NP$.
There are differences in the size of the optimal solutions of the same instances for contiguous task scheduling and the closely related non-contiguous task scheduling $P|size_j|C_{\max}$. 
These differences were noted by~\cite{TurekWY92} and intensively studied by~\cite{BladekDGS15}.
Furthermore notable are the differences in the pseudo-polynomial absolute approximation ratio. 
While for the contiguous case we have a lower bound of $5/4$ if $\Poly \not = \NP$, in the non contiguous case there is a pseudo-PTAS~\cite{JansenT10}.

\begin{figure}
	\resizebox{\textwidth}{!}{
		\begin{tikzpicture}
		\pgfmathsetmacro{\w}{29}
		\draw[-](-0.5 +\w,0)--(3/2 * \w + 0.5,0);
		\draw(\w,0)--(\w,-0.5) node [below] {$1+\varepsilon$};
		\draw(12/11 * \w,0.5) node[above]{\cite{AdamaszekKPP17}}--(12/11 * \w,-0.5) node [below] {$12/11$};
		\draw[double,thick] (5/4 * \w,0.5) node[above]{\cite{HenningJRS18}, $5/4 +\eps$}--(5/4 * \w,-0.5) node [below,xshift=0.01*\w cm] {$5/4$, \small{Theorem~\ref{thm:StripPacking}}};
		\draw(4/3 * \w,- 0.5)node [below]{{\cite{GalvezGIK16,JansenR16}}} --(4/3 * \w,0.5) node [above] {$4/3+\varepsilon$};
		\draw(7/5 * \w,-0.5)node [below]{{\cite{NadiradzeW16}}} --(7/5 * \w,0.5) node [above] {$7/5+\varepsilon$};
		\draw(3/2 * \w,-0.5) node [below]{\cite{JansenT10}} --(3/2 * \w,0.5) node [above] {$3/2+\varepsilon$};
		\draw [<-, very thick] (5/4 * \w,0.25) -- (4/3 * \w,0.25);
		\end{tikzpicture}
	}
	\caption{The upper and lower bounds for pseudo-polynomial approximations achieved so far.}
	\label{fig:approximationratios}
\end{figure}
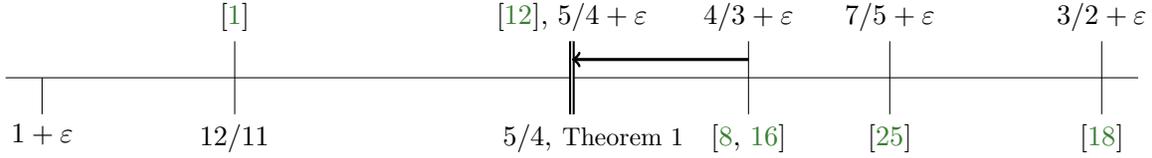

Strip Packing with rotations has been explicitly studied in the following papers~\cite{EpsteinS06a, MiyazawaW04, JansenS05, JansenS09}. 
Algorithms for Strip packing without rotations using the area of the items to prove their ratio, e.g., NFDH, FFDH~\cite{BakerCR80} or Steinberg's algorithm~\cite{Steinberg97}, work for Strip Packing with rotations as well. 
Furthermore, algorithms using 2D Knapsack with area maximization as a subroutine can also be extended to Strip Packing with rotations. 
On the other hand, the lower bounds of $3/2$ for polynomial and $5/4$ for pseudo-polynomial approximation ratios hold for Strip Packing with rotations as well unless $\Poly = \NP$.

\paragraph*{Contiguous Moldable Task Scheduling} 

Moldable jobs are studied for two kinds of jobs, those that need to be scheduled on contiguous machines, and those that do not. 
Note that due to a reduction from the Partition Problem algorithms for non-monotonic jobs cannot have a ratio better than $3/2$ for both cases unless $\Poly = \NP$. 
Furthermore for non-monotonic contiguous jobs, we cannot find a pseudo-polynomial algorithm with ratio better than $5/4$ unless $\Poly = \NP$ since this problem contains Strip Packing as a special case.
Turek, Wolf and Yu~\cite{TurekWY92} presented an algorithm that assigns jobs to numbers of processors and then schedules the fix instance with known algorithms for these scenarios. 
This algorithms achieves a 2-approximation for the non-contiguous case and a 2.5-approximation for the contiguous case, using Sleator's algorithm~\cite{Sleator80} as a subroutine. 
Furthermore, they pointed out that for improved approximations for the fixed processor instances the algorithm achieves better approximation ratios. 
More precisely, if the algorithm uses Steinberg's Algorithm~\cite{Steinberg97} as a subroutine instead (which was not known when the paper was published), it has an approximation ratio of $2$ for the contiguous case. 
The running time of these algorithms was improved by Ludwig and Tiwari~\cite{LudwigT94} from $\mathcal{O}(mn\cdot L(m,n))$ to $\mathcal{O}(mn+ L(m,n))$, where $\mathcal{O}(L(m,n))$ is the running time of the used subroutine for the fixed machine instance. 
There is a pseudo-polynomial algorithm with ratio $(3/2 +\eps)$ by Mouni\'e et al.~\cite{MounieRT07} for monotonic non-contiguous moldable jobs. 
Jansen and Th\"ole~\cite{JansenT10} extended the $(3/2 +\eps)$ ratio to non-monotonic contiguous moldable jobs. 
Furthermore, they presented a pseudo\hyp{}PTAS for non-monotonic non-contiguous moldable jobs. 
Together with the FPTAS from~\cite{JansenL18} for the case $m \geq 8n/\eps$ this delivers a PTAS for the case of monotonic non-contiguous jobs. 
Additionally, the running time of the algorithm by Mouni\'e et al.~\cite{MounieRT07} is improved to be nearly linear by Jansen and Land~\cite{JansenL18}.
A polynomial $(3/2+\eps)$ approximation algorithm for non-monotonic non-contiguous jobs was presented by Jansen~\cite{Jansen12}, which is arbitrary close to the best possible algorithm for this case unless $\Poly=\NP$.

\subsection*{Methodology and Organization of this Paper} 

\begin{figure}
	\centering
	\begin{minipage}[t]{.38\linewidth}
		\resizebox{\textwidth}{!}{%
			\begin{tikzpicture}
			\pgfmathsetmacro{\w}{2.3}
			\pgfmathsetmacro{\h}{3.1}
			
			\draw[white] (0,0) rectangle (\w,5*\h/4);
			\draw[fill = black!15!white] (0,0) rectangle (\w,\h);
			
			\draw[fill = black!30!white] (0.0*\w,0.6*\h) rectangle (0.1*\w,1.0*\h);
			\draw[fill = black!30!white] (0.1*\w,0.1*\h) rectangle (0.15*\w,0.7*\h);
			\draw[fill = black!30!white] (0.15*\w,0.15*\h) rectangle (0.25*\w,0.55*\h);
			\draw[fill = black!30!white] (0.2*\w,0.61*\h) rectangle (0.3*\w,0.98*\h);
			\draw[fill = black!30!white] (0.3*\w,0.3*\h) rectangle (0.35*\w,0.73*\h);
			\draw[fill = black!30!white] (0.4*\w,0.1*\h) rectangle (0.45*\w,0.54*\h);
			\draw[fill = black!30!white] (0.4*\w,0.59*\h) rectangle (0.5*\w,0.95*\h);
			\draw[fill = black!30!white] (0.45*\w,0.05*\h) rectangle (0.5*\w,0.42*\h);
			\draw[fill = black!30!white] (0.5*\w,0.48*\h) rectangle (0.65*\w,0.9*\h);
			\draw[fill = black!30!white] (0.5*\w,0.04*\h) rectangle (0.55*\w,0.4*\h);
			\draw[fill = black!30!white] (0.55*\w,0.06*\h) rectangle (0.65*\w,0.43*\h);
			\draw[fill = black!30!white] (0.65*\w,0.0*\h) rectangle (0.75*\w,0.82*\h);
			\draw[fill = black!30!white] (0.75*\w,0.15*\h) rectangle (0.85*\w,0.5*\h);
			\draw[fill = black!30!white] (0.8*\w,0.55*\h) rectangle (0.9*\w,0.9*\h);
			\draw[fill = black!30!white] (0.85*\w,0.2*\h) rectangle (0.95*\w,0.55*\h);
			
			\draw[dashed] (-0.1*\w,\h/3) -- (1.1*\w,\h/3);
			\draw[dashed] (-0.1*\w,2*\h/3) -- (1.1*\w,2*\h/3);
			
			\node at (\w/2, -0.5) {\captionsize optimal};
			
			\begin{scope} [xshift = 1.2 *\w cm]
			\draw[] (0,4*\h/3) rectangle (\w,\h);
			
			\draw[fill = black!15!white] (0.0*\w,0.35*\h) rectangle (0.05*\w,1*\h);
			\draw[fill = black!15!white] (0.05*\w,0.36*\h) rectangle (0.1*\w,1*\h);
			\draw[fill = black!15!white] (0.1*\w,0.4*\h) rectangle (0.15*\w,1*\h);
			\draw[fill = black!30!white] (0.15*\w,0.57*\h) rectangle (0.3*\w,1*\h);
			\draw[fill = black!30!white] (0.3*\w,0.6*\h) rectangle (0.4*\w,1*\h);
			\draw[fill = black!15!white] (0.4*\w,0.6*\h) rectangle (0.45*\w,1*\h);
			\draw[fill = black!30!white] (0.45*\w,0.63*\h) rectangle (0.55*\w,1*\h);
			\draw[fill = black!30!white] (0.55*\w,0.64*\h) rectangle (0.65*\w,1*\h);
			\draw[fill = black!30!white] (0.65*\w,0.65*\h) rectangle (0.75*\w,1*\h);
			\draw[fill = black!15!white] (0.75*\w,0.7*\h) rectangle (0.8*\w,1*\h);
			\draw[fill = black!15!white] (0.8*\w,0.82*\h) rectangle (0.9*\w,1*\h);
			
			\draw[fill = black!30!white] (0.0*\w,0.0*\h) rectangle (0.1*\w,0.35*\h);
			\draw[fill = black!30!white] (0.1*\w,0.0*\h) rectangle (0.2*\w,0.36*\h);
			\draw[fill = black!30!white] (0.2*\w,0.0*\h) rectangle (0.25*\w,0.36*\h);
			\draw[fill = black!30!white] (0.25*\w,0.0*\h) rectangle (0.35*\w,0.37*\h);
			\draw[fill = black!30!white] (0.35*\w,0.0*\h) rectangle (0.4*\w,0.37*\h);
			\draw[fill = black!30!white] (0.4*\w,0.0*\h) rectangle (0.5*\w,0.4*\h);
			\draw[fill = black!30!white] (0.5*\w,0.0*\h) rectangle (0.55*\w,0.53*\h);
			\draw[fill = black!15!white] (0.55*\w,0.0*\h) rectangle (0.65*\w,0.6*\h);
			\draw[fill = black!30!white] (0.65*\w,0.0*\h) rectangle (0.7*\w,0.6*\h);
			\draw[fill = black!15!white] (0.7*\w,0.0*\h) rectangle (0.75*\w,0.63*\h);
			\draw[fill = black!30!white] (0.75*\w,0.0*\h) rectangle (0.8*\w,0.70*\h);
			\draw[fill = black!30!white] (0.8*\w,0.0*\h) rectangle (0.9*\w,0.82*\h);
			\draw[fill = black!15!white] (0.9*\w,0.0*\h) rectangle (0.95*\w,1*\h);
			\draw[fill = black!15!white] (0.95*\w,0.0*\h) rectangle (1.0*\w,1*\h);
			
			\draw[fill = black!15!white] (0.2*\w,0.4*\h +2*\h/3) rectangle (0.25*\w,0.63*\h+2*\h/3);
			\draw[fill = black!15!white] (0.4*\w,0.53*\h+2*\h/3) rectangle (0.45*\w,0.64*\h+2*\h/3);
			\draw[fill = black!15!white] (0.45*\w,0.37*\h+2*\h/3) rectangle (0.5*\w,0.64*\h+2*\h/3);
			\draw[fill = black!15!white] (0.5*\w,0.36*\h+2*\h/3) rectangle (0.55*\w,0.58*\h+2*\h/3);
			\draw[fill = black!15!white] (0.55*\w,0.37*\h+2*\h/3) rectangle (0.65*\w,0.58*\h+2*\h/3);
			\draw[fill = black!15!white] (0.8*\w,0.36*\h+2*\h/3) rectangle (0.85*\w,0.65*\h+2*\h/3);
			\draw[fill = black!15!white] (0.85*\w,0.35*\h+2*\h/3) rectangle (0.9*\w,0.65*\h+2*\h/3);
			
			\draw[dashed] (-0.1*\w,\h/3) -- (1.1*\w,\h/3);
			\draw[dashed] (-0.1*\w,2*\h/3) -- (1.1*\w,2*\h/3);
			\draw[red, pattern = north west lines] (0*\w,\h) rectangle (0.1*\w, 4*\h/3);
			
			\node at (\w/2, -0.5) {\captionsize reordered};
			
			\end{scope}
			\end{tikzpicture}
		}
		\subcaption{Previous reordering technique}
		\label{fig:comparisonA}
	\end{minipage}%
	\hfill%
	\begin{minipage}[t]{.57\linewidth}
		\resizebox{\textwidth}{!}{%
			\begin{tikzpicture}
			\pgfmathsetmacro{\w}{2.3}
			\pgfmathsetmacro{\h}{3.1}
			\draw[white] (0,0) rectangle (\w,5*\h/4);
			\draw[fill = black!15!white] (0,0) rectangle (\w,\h);
			
			\draw[fill = black!30!white] (0.1*\w,0.1*\h) rectangle (0.15*\w,0.7*\h);
			\draw[fill = black!30!white] (0.8*\w,0.3*\h) rectangle (0.85*\w,0.9*\h);
			
			\draw[fill = black!30!white] (0.5*\w,0.6*\h) rectangle (0.6*\w,0.9*\h);
			\draw[fill = black!30!white] (0.55*\w,0.05*\h) rectangle (0.65*\w,0.3*\h);
			\draw[fill = black!30!white] (0.45*\w,0.32*\h) rectangle (0.575*\w,0.58*\h);
			
			\draw[fill = black!30!white] (0.4*\w,0.55*\h) rectangle (0.45*\w,0.95*\h);
			\draw[fill = black!30!white] (0.35*\w,0.03*\h) rectangle (0.45*\w,0.4*\h);
			\draw[fill = black!30!white] (0.3*\w,0.42*\h) rectangle (0.4*\w,0.69*\h);
			\draw[fill = black!30!white] (0.25*\w,0.71*\h) rectangle (0.35*\w,0.98*\h);
			
			\draw[fill = black!30!white] (0.15*\w,0.62*\h) rectangle (0.25*\w,0.91*\h);
			\draw[fill = black!30!white] (0.0*\w,0.55*\h) rectangle (0.1*\w,0.8*\h);
			\draw[fill = black!30!white] (0.65*\w,0.62*\h) rectangle (0.7*\w,0.91*\h);
			\draw[fill = black!30!white] (0.7*\w,0.54*\h) rectangle (0.8*\w,0.97*\h);
			\draw[fill = black!30!white] (0.9*\w,0.7*\h) rectangle (0.95*\w,0.97*\h);
			
			\draw[fill = black!30!white] (0.8*\w,0.01*\h) rectangle (0.95*\w,0.28*\h);
			
			\draw[fill = black!30!white] (0.675*\w,0.01*\h) rectangle (0.775*\w,0.4*\h);
			\draw[fill = black!30!white] (0.2*\w,0.01*\h) rectangle (0.3*\w,0.3*\h);
			\draw[fill = black!30!white] (0.05*\w,0.02*\h) rectangle (0.1*\w,0.28*\h);
			
			\draw[fill = black!30!white] (0.025*\w,0.28*\h) rectangle (0.075*\w,0.53*\h);
			\draw[fill = black!30!white] (0.175*\w,0.34*\h) rectangle (0.275*\w,0.59*\h);
			\draw[fill = black!30!white] (0.625*\w,0.32*\h) rectangle (0.675*\w,0.6*\h);
			\draw[fill = black!30!white] (0.925*\w,0.33*\h) rectangle (0.975*\w,0.66*\h);
			
			\draw[dashed] (-0.1*\w,\h/4) -- (1.1*\w,\h/4);
			\draw[dashed] (-0.1*\w,\h/2) -- (1.1*\w,\h/2);
			\draw[dashed] (-0.1*\w,3*\h/4) -- (1.1*\w,3*\h/4);
			
			\node at (\w/2, -0.5) {\captionsize optimal};
			
			\begin{scope} [xshift = 1.2 *\w cm]
			\draw[] (0,0) rectangle (\w,5*\h/4);
			
			\draw[fill = black!15!white] (0.0*  \w,0.5 *\h) rectangle (0.025*\w,1.25*\h);
			\draw[fill = black!30!white] (0.025*\w,0.65*\h) rectangle (0.075*\w,1.25*\h);
			\draw[fill = black!15!white] (0.025*\w,0.5*\h) rectangle (0.075*\w,0.63*\h);
			\draw[fill = black!15!white] (0.075*\w,0.52*\h) rectangle (0.125*\w,1.25*\h);
			\draw[fill = black!15!white] (0.125*\w,0.78*\h) rectangle (0.15*\w,1.25*\h);
			\draw[fill = black!30!white] (0.15*\w,0.82*\h) rectangle (0.25*\w,1.25*\h);
			\draw[fill = black!15!white] (0.25*\w,0.83*\h) rectangle (0.3*\w,1.25*\h);
			\draw[fill = black!15!white] (0.3*\w,0.85*\h) rectangle (0.35*\w,1.25*\h);
			\draw[fill = black!30!white] (0.35*\w,0.85*\h) rectangle (0.4*\w,1.25*\h);
			\draw[fill = black!15!white] (0.4*\w,0.89*\h) rectangle (0.45*\w,1.25*\h);
			\draw[fill = black!15!white] (0.45*\w,0.91*\h) rectangle (0.475*\w,1.25*\h);
			\draw[fill = black!30!white] (0.475*\w,0.95*\h) rectangle (0.575*\w,1.25*\h);
			\draw[fill = black!30!white] (0.575*\w,0.96*\h) rectangle (0.675*\w,1.25*\h);
			\draw[fill = black!30!white] (0.675*\w,0.96*\h) rectangle (0.725*\w,1.25*\h);
			\draw[fill = black!30!white] (0.725*\w,0.98*\h) rectangle (0.775*\w,1.25*\h);
			\draw[fill = black!30!white] (0.775*\w,0.98*\h) rectangle (0.875*\w,1.25*\h);
			\draw[fill = black!30!white] (0.875*\w,1*\h) rectangle (0.975*\w,1.25*\h);
			
			\draw[fill = black!15!white] (0.0*\w,0.0*\h) rectangle (0.025*\w,0.5*\h);
			\draw[fill = black!15!white] (0.025*  \w,0.0*\h) rectangle (0.075*\w,0.46*\h);
			\draw[fill = black!15!white] (0.075* \w,0.0*\h) rectangle (0.1*\w,0.46*\h);
			\draw[fill = black!15!white] (0.1* \w,0.0*\h) rectangle (0.15*\w,0.44*\h);
			\draw[fill = black!15!white] (0.15*  \w,0.0*\h) rectangle (0.175*\w,0.43*\h);
			\draw[fill = black!30!white] (0.175* \w,0.0*\h) rectangle (0.275*\w,0.39*\h);
			\draw[fill = black!30!white] (0.275* \w,0.0*\h) rectangle (0.375*\w,0.37*\h);
			\draw[fill = black!15!white] (0.375*\w,0.0*\h) rectangle (0.4*\w,0.33*\h);
			\draw[fill = black!15!white] (0.4*  \w,0.0*\h) rectangle (0.45*\w,0.32*\h);
			\draw[fill = black!30!white] (0.45*\w,0.0*\h) rectangle (0.55*\w,0.29*\h);
			\draw[fill = black!30!white] (0.55* \w,0.0*\h) rectangle (0.7*\w,0.27*\h);
			\draw[fill = black!30!white] (0.7*  \w,0.0*\h) rectangle (0.75*\w,0.26*\h);
			\draw[fill = black!30!white] (0.75 *\w,0.0*\h) rectangle (0.85*\w,0.25*\h);
			
			\draw[fill = black!15!white] (0.85* \w,0.0*\h) rectangle (0.875*\w,0.57*\h);
			\draw[fill = black!30!white] (0.875*\w,0.0*\h) rectangle (0.925*\w,0.6*\h);
			\draw[fill = black!15!white] (0.925*\w,0.0*\h) rectangle (0.95*\w,0.71*\h);
			\draw[fill = black!15!white] (0.95* \w,0.0*\h) rectangle (0.975*\w,0.75*\h);
			\draw[fill = black!15!white] (0.975*\w,0.0*\h) rectangle (1*\w,1*\h);

			\draw[fill = black!15!white] (0.125*\w,0.57*\h) rectangle (0.2*\w,0.75*\h);
			\draw[fill = black!15!white] (0.2  *\w,0.52*\h) rectangle (0.25*\w,0.75*\h);
			\draw[fill = black!30!white] (0.25 *\w,0.5 *\h) rectangle (0.3*\w,0.75*\h);
			\draw[fill = black!30!white] (0.3  *\w,0.5 *\h) rectangle (0.4*\w,0.75*\h);
			\draw[fill = black!30!white] (0.4  *\w,0.49*\h) rectangle (0.525*\w,0.75*\h);
			\draw[fill = black!30!white] (0.525*\w,0.48*\h) rectangle (0.625*\w,0.75*\h);
			\draw[fill = black!30!white] (0.625*\w,0.47*\h) rectangle (0.675*\w,0.75*\h);
			\draw[fill = black!15!white] (0.675*\w,0.44*\h) rectangle (0.7*\w,0.75*\h);
			\draw[fill = black!30!white] (0.7  *\w,0.42*\h) rectangle (0.75*\w,0.75*\h);
			\draw[fill = black!15!white] (0.75 *\w,0.31*\h) rectangle (0.775*\w,0.75*\h);
			\draw[fill = black!15!white] (0.775*\w,0.3 *\h) rectangle (0.8*\w,0.75*\h);
			\draw[fill = black!15!white] (0.8  *\w,0.29*\h) rectangle (0.825*\w,0.75*\h);
			\draw[fill = black!15!white] (0.825*\w,0.26*\h) rectangle (0.85*\w,0.75*\h);
			
			\draw[fill = black!15!white] (0.825 *\w,0.75*\h) rectangle (0.85 *\w,0.88*\h);
			\draw[fill = black!15!white] (0.85*\w,0.75*\h) rectangle (0.9 *\w,0.92*\h);
			\draw[fill = black!15!white] (0.9*\w,0.75*\h) rectangle (0.925*\w,0.94*\h);
			\draw[fill = black!15!white] (0.925  *\w,0.75*\h) rectangle (0.95*\w,0.94*\h);
			\draw[fill = black!15!white] (0.95 *\w,0.75*\h) rectangle (0.975*\w,0.99*\h);

			\draw[dashed] (-0.1*\w,\h/4) -- (1.1*\w,\h/4);
			\draw[dashed] (-0.1*\w,\h/2) -- (1.1*\w,\h/2);
			\draw[dashed] (-0.1*\w,3*\h/4) -- (1.1*\w,3*\h/4);
			\draw[dashed] (-0.1*\w,\h) -- (1.1*\w,\h);
			
			\node at (\w/2, -0.5) {\captionsize reordered};
			
			\begin{scope} [xshift = 1.2 *\w cm]
			\draw[] (0,0) rectangle (\w,5*\h/4);
			
			\draw[fill = black!15!white] (0.0*  \w,0.5 *\h) rectangle (0.025*\w,1.25*\h);
			\draw[fill = black!30!white] (0.025*\w,0.65*\h) rectangle (0.075*\w,1.25*\h);
			\draw[fill = black!15!white] (0.025*\w,0.5*\h) rectangle (0.075*\w,0.63*\h);
			\draw[fill = black!15!white] (0.075*\w,0.52*\h) rectangle (0.125*\w,1.25*\h);
			\draw[fill = black!15!white] (0.125*\w,0.78*\h) rectangle (0.15*\w,1.25*\h);
			\draw[fill = black!30!white] (0.15*\w,0.82*\h) rectangle (0.25*\w,1.25*\h);
			\draw[fill = black!15!white] (0.25*\w,0.83*\h) rectangle (0.3*\w,1.25*\h);
			\draw[fill = black!15!white] (0.3*\w,0.85*\h) rectangle (0.35*\w,1.25*\h);
			\draw[fill = black!30!white] (0.35*\w,0.85*\h) rectangle (0.4*\w,1.25*\h);
			\draw[fill = black!15!white] (0.4*\w,0.89*\h) rectangle (0.45*\w,1.25*\h);
			\draw[fill = black!15!white] (0.45*\w,0.91*\h) rectangle (0.475*\w,1.25*\h);
			\draw[fill = black!30!white] (0.475*\w,0.95*\h) rectangle (0.575*\w,1.25*\h);
			\draw[fill = black!30!white] (0.575*\w,0.96*\h) rectangle (0.675*\w,1.25*\h);
			\draw[fill = black!30!white] (0.675*\w,0.96*\h) rectangle (0.725*\w,1.25*\h);
			\draw[fill = black!30!white] (0.725*\w,0.98*\h) rectangle (0.775*\w,1.25*\h);
			\draw[fill = black!30!white] (0.775*\w,0.98*\h) rectangle (0.875*\w,1.25*\h);
			\draw[fill = black!30!white] (0.875*\w,1*\h) rectangle (0.975*\w,1.25*\h);
			
			\draw[fill = black!15!white] (0.0*\w,0.0*\h) rectangle (0.025*\w,0.5*\h);
			\draw[fill = black!15!white] (0.025*  \w,0.0*\h) rectangle (0.075*\w,0.46*\h);
			\draw[fill = black!15!white] (0.075* \w,0.0*\h) rectangle (0.1*\w,0.46*\h);
			\draw[fill = black!15!white] (0.1* \w,0.0*\h) rectangle (0.15*\w,0.44*\h);
			\draw[fill = black!15!white] (0.15*  \w,0.0*\h) rectangle (0.175*\w,0.43*\h);
			\draw[fill = black!30!white] (0.175* \w,0.0*\h) rectangle (0.275*\w,0.39*\h);
			\draw[fill = black!30!white] (0.275* \w,0.0*\h) rectangle (0.375*\w,0.37*\h);
			\draw[fill = black!15!white] (0.375*\w,0.0*\h) rectangle (0.4*\w,0.33*\h);
			\draw[fill = black!15!white] (0.4*  \w,0.0*\h) rectangle (0.45*\w,0.32*\h);
			\draw[fill = black!30!white] (0.45*\w,0.0*\h) rectangle (0.55*\w,0.29*\h);
			\draw[fill = black!30!white] (0.55* \w,0.0*\h) rectangle (0.7*\w,0.27*\h);
			\draw[fill = black!30!white] (0.7*  \w,0.0*\h) rectangle (0.75*\w,0.26*\h);
			\draw[fill = black!30!white] (0.75 *\w,0.0*\h) rectangle (0.85*\w,0.25*\h);
			
			\draw[fill = black!15!white] (0.85* \w,0.0*\h) rectangle (0.875*\w,0.57*\h);
			\draw[fill = black!30!white] (0.875*\w,0.0*\h) rectangle (0.925*\w,0.6*\h);
			\draw[fill = black!15!white] (0.925*\w,0.0*\h) rectangle (0.95*\w,0.71*\h);
			\draw[fill = black!15!white] (0.95* \w,0.0*\h) rectangle (0.975*\w,0.75*\h);
			\draw[fill = black!15!white] (0.975*\w,0.0*\h) rectangle (1*\w,1*\h);

			\draw[fill = black!15!white] (0.125*\w,0.57*\h -0.13*\h) rectangle (0.15*\w,0.75*\h-0.13*\h);
			\draw[fill = black!15!white] (0.15*\w,0.57*\h -0.14*\h) rectangle (0.175*\w,0.75*\h-0.14*\h);
			\draw[fill = black!15!white] (0.175*\w,0.57*\h -0.18*\h) rectangle (0.2*\w,0.75*\h-0.18*\h);
			\draw[fill = black!15!white] (0.2  *\w,0.52*\h-0.13*\h) rectangle (0.25*\w,0.75*\h-0.13*\h);
			\draw[fill = black!30!white] (0.25 *\w,0.5 *\h-0.11*\h) rectangle (0.3*\w,0.75*\h-0.11*\h);
			\draw[fill = black!30!white] (0.3  *\w,0.5 *\h-0.11*\h) rectangle (0.4*\w,0.75*\h-0.11*\h);
			\draw[fill = black!30!white] (0.4  *\w,0.49*\h-0.17*\h) rectangle (0.525*\w,0.75*\h-0.17*\h);
			\draw[fill = black!30!white] (0.525*\w,0.48*\h-0.19*\h) rectangle (0.625*\w,0.75*\h-0.19*\h);
			\draw[fill = black!30!white] (0.625*\w,0.47*\h-0.20*\h) rectangle (0.675*\w,0.75*\h-0.20*\h);
			\draw[fill = black!15!white] (0.675*\w,0.44*\h-0.17*\h) rectangle (0.7*\w,0.75*\h-0.17*\h);
			\draw[fill = black!30!white] (0.7  *\w,0.42*\h-0.16*\h) rectangle (0.75*\w,0.75*\h-0.16*\h);
			\draw[fill = black!15!white] (0.75 *\w,0.31*\h-0.06*\h) rectangle (0.775*\w,0.75*\h-0.06*\h);
			\draw[fill = black!15!white] (0.775*\w,0.3 *\h-0.05*\h) rectangle (0.8*\w,0.75*\h-0.05*\h);
			\draw[fill = black!15!white] (0.8  *\w,0.29*\h-0.04*\h) rectangle (0.825*\w,0.75*\h-0.04*\h);
			\draw[fill = black!15!white] (0.825*\w,0.26*\h-0.01*\h) rectangle (0.85*\w,0.75*\h-0.01*\h);
			
			\draw[fill = black!15!white] (0.825 *\w,0.75*\h) rectangle (0.85 *\w,0.88*\h);
			\draw[fill = black!15!white] (0.85*\w,0.75*\h) rectangle (0.9 *\w,0.92*\h);
			\draw[fill = black!15!white] (0.9*\w,0.75*\h) rectangle (0.925*\w,0.94*\h);
			\draw[fill = black!15!white] (0.925  *\w,0.75*\h) rectangle (0.95*\w,0.94*\h);
			\draw[fill = black!15!white] (0.95 *\w,0.75*\h) rectangle (0.975*\w,0.99*\h);
			
			\draw[red, pattern = north west lines] (0.43 *\w,0.6*\h) rectangle (0.53*\w,0.6*\h + 0.25*\h);
			
			\draw[dashed] (-0.1*\w,\h/4) -- (1.1*\w,\h/4);
			\draw[dashed] (-0.1*\w,\h/2) -- (1.1*\w,\h/2);
			\draw[dashed] (-0.1*\w,3*\h/4) -- (1.1*\w,3*\h/4);
			\draw[dashed] (-0.1*\w,\h) -- (1.1*\w,\h);
			
			\node at (\w/2, -0.5) {\captionsize gap for disc. items};
			\end{scope}
			\end{scope}
			\end{tikzpicture}
		}%
		\subcaption{The new shifting and reordering technique}
		\label{fig:comparisonB}
	\end{minipage}
	\caption{Comparison of old and new strategies in the simplified case}
	\label{fig:comparison}
\end{figure}
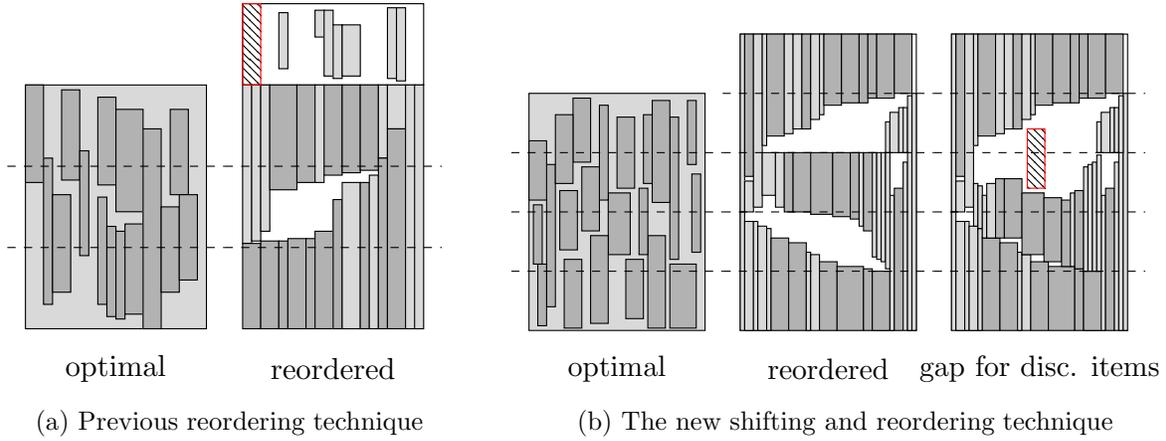

In the approaches seen before, i.e., in~\cite{NadiradzeW16},~\cite{GalvezGIK16} and~\cite{JansenR16}, there arises a natural set of critical items, e.g., all items with height larger than $\nicefrac{1}{3}~\OPT$ in~\cite{GalvezGIK16} and~\cite{JansenR16}. 
The characteristic of this set is that the aspired approximation ratio is exceeded if we place one of these items on top of the packing.

The technique used in these previous approaches is heavily dependent on the fact that there can be at most two critical items on top of each other. 
This allows to place all critical items in the optimal packing area while discarding some noncritical items, which are placed on top of the optimal packing later (see Figure~\ref{fig:comparisonA}). 
If three critical items can be put on each other, this technique will not work. 
To find an algorithm with ratio $4/3 - \eps$, we need to overcome this major obstacle.

To construct a $(5/4+\eps)$-approximation, we introduce a new technique to handle this difficulty, in the following called \textit{shifting and reordering}. 
While we cannot guarantee that all critical items are packed in the optimal packing area, we can place the critical items with height larger than $\nicefrac{1}{4}\OPT$ on three shelves using the area $W \times (5/4 +\eps)\OPT$ (see Figure~\ref{fig:comparisonB}). 

A challenge which arises using this new strategy is the fact that some of the other items have to be discarded due to slicing. 
Since the packing area is already extended by the factor $(1/4 +\eps)$, it is no longer possible to simply place these discarded items on top of the packing, as shown in Figure~\ref{fig:comparisonB}. 
The discarded items have to be placed carefully into the gaps generated by the shifting and reordering technique. We prove that for each possible optimal packing the corresponding rearranged packing contains suitable gaps for these items.

In Section~\ref{sec:simpleCase}, we describe this shifting and reordering technique and the specific structure it generates for the simplified case that just these critical items have to be placed integrally while all other items are allowed be partitioned into vertical slices, which do not have to be placed contiguously. 

For the structural result, all items have to be placed integral; thus, we cannot slice all noncritical items. 
Nevertheless, we may still slice certain narrow items. 
We use a lemma from~\cite{JansenR16}, which states the possibility to partition the packing into few rectangular areas. 
In this partition, we have the property that each critical item is contained in an area exclusively containing critical and sliceable items. 
Up to three critical items can overlap each of the vertical borders of these areas and these overlapping items may not be shifted horizontally or vertically. 
We managed to extend the new strategy to these areas although the strategy becomes much more involved in this extension. 

Combining our new techniques to place critical items on three shelves, find suitable gaps for discarded non-critical items and handle the exclusive slicing of narrow items together enables us to prove the structural result from Lemma~\ref{lma:structureLemma}. 
In the algorithm, we guess the structure from our structure result and use dynamic programming to place the items into this structure.

The strength of the structure result is, that it applies to all optimal solutions with the property that they consist of rectangular objects placed into a rectangle that is extendable on one side. 
Optimal solutions of the three considered problems all have this property.
Thanks to this fact, we where able to obtain algorithms which find $5/4 +\eps$ approximations for each of the three problems by carefully adapting the dynamic program.

\section{The simplified Case}
\label{sec:simpleCase}
To demonstrate the central new idea which leads to the improved structural result -- the shifting and reordering technique -- we consider the following simplified case. 
In this scenario, we have to pack items with a tall height integrally, while we are allowed to slice all other items vertically. 

Let a packing with height $H$ be given. 
We define tall items as the items which have a height larger than $\nicefrac{1}{4} H$. 
Further assume that there is an arithmetic grid with $\NumGridpoints +1$ horizontal grid lines with distance $H/\NumGridpoints$ such that each tall item starts and ends at the grid lines. 
For now, we can think of this grid as the integral grid with $H +1$ grid lines.
Later, we can reduce the grid lines by rounding the heights of the items. 
We are interested in a fractional packing of the non-tall items. 
Therefore, we replace each non-tall item $i$ by exactly $w(i)$ items with height $h(i)$ and width $1$. 
This step is called slicing. 
We define a box as a rectangular sub area of the packing area.
\begin{lemma}
	By adding at most $\nicefrac{1}{4} H$ to the packing height and slicing non-tall items, we can rearrange the items such that we generate at most $\nicefrac{3}{2} \NumGridpoints$ containers which contain tall items with the same height only, and at most $\nicefrac{9}{4}\NumGridpoints +1$ container for sliced items.
	\label{lma:simpleCase}
\end{lemma}

\begin{figure}
	\centering
	\begin{subfigure}[t]{.24\textwidth}
		\centering
		\begin{tikzpicture}
		\pgfmathsetmacro{\w}{2.6}
		\pgfmathsetmacro{\h}{4.0}
		\draw[white] (0,0) rectangle (\w,5*\h/4);
		\draw[fill = myLightGray] (0,0) rectangle (\w,\h);
		
		\foreach \x/\y/\xx/\yy in {
			0.0  /0.55/0.1  /0.8 ,
			0.025/0.28/0.075/0.53,
			0.05 /0.02/0.1  /0.28,
			0.1  /0.1 /0.15 /0.7 ,
			0.15 /0.62/0.25 /0.91,
			0.175/0.34/0.275/0.59,
			0.2  /0.01/0.3  /0.3 ,
			0.25 /0.71/0.35 /0.98,
			0.3  /0.42/0.4  /0.69,
			0.35 /0.03/0.45 /0.4 ,
			0.4  /0.55/0.45 /0.95,
			0.45 /0.32/0.575/0.58,
			0.5  /0.6 /0.6  /0.9 ,
			0.55 /0.05/0.65 /0.3 ,
			0.625/0.32/0.675/0.6 ,
			0.65 /0.62/0.7  /0.91,
			0.675/0.01/0.775/0.4 ,
			0.7  /0.54/0.8  /0.97,
			0.8  /0.01/0.95 /0.28,
			0.8  /0.3 /0.85 /0.9 ,
			0.9  /0.7 /0.95 /0.97,
			0.925/0.33/0.975/0.66
		}{
			\draw[fill = myGray] (\x*\w,\y*\h) rectangle (\xx*\w,\yy*\h);
		}
		
		\draw[dashed] (-0.1*\w,\h/4) -- (1.1*\w,\h/4);
		\draw[dashed] (-0.1*\w,\h/2) -- (1.1*\w,\h/2);
		\draw[dashed] (-0.1*\w,3*\h/4) -- (1.1*\w,3*\h/4);
		
		
		\end{tikzpicture}
		\subcaption{An optimal packing.}
		\label{fig:simpleShiftA}
	\end{subfigure}%
	\hfill%
	\begin{subfigure}[t]{.24\textwidth}
		\centering
		\begin{tikzpicture}
		\pgfmathsetmacro{\w}{2.6}
		\pgfmathsetmacro{\h}{4.0}
		
		\draw[] (0,0) rectangle (\w,\h);
		
		\foreach \x/\y/\xx/\yy in {
			0.0  / 0.0 / 0.025/ 0.75,
			0.025/ 0.53/ 0.05 / 0.75,
			0.025/ 0.0 / 0.05 / 0.28,
			0.05 / 0.26/ 0.075/ 0.28,
			0.05 / 0.53/ 0.075/ 0.75,
			0.075/ 0.26/ 0.1  / 0.75,
			0.1  / 0.6 / 0.15 / 1.0 ,
			0.15 / 0.0 / 0.175/ 0.71,
			0.175/ 0.0 / 0.2  / 0.34,
			0.175/ 0.59/ 0.2  / 0.71,
			0.2  / 0.29/ 0.25 / 0.34,
			0.2  / 0.59/ 0.25 / 0.71,
			0.25 / 0.29/ 0.275/ 0.34,
			0.25 / 0.59/ 0.275/ 0.73,
			0.275/ 0.29/ 0.3  / 0.73,
			0.3  / 0.0 / 0.35 / 0.42,
			0.3  / 0.69/ 0.35 / 0.73,
			0.35 / 0.37/ 0.4  / 0.42,
			0.35 / 0.69/ 0.4  / 1.0,
			0.4  / 0.37/ 0.45 / 0.6 ,
			0.45 / 0.58/ 0.5  / 1.0 ,
			0.45 / 0.0 / 0.5  / 0.32,
			0.5  / 0.58/ 0.55 / 0.7 ,
			0.5  / 0.0 / 0.55 / 0.32,
			0.55 / 0.25/ 0.575/ 0.32,
			0.55 / 0.58/ 0.575/ 0.7 ,
			0.575/ 0.25/ 0.6  / 0.7 ,
			0.6  / 0.25/ 0.625/ 1.0 ,
			0.625/ 0.25/ 0.65 / 0.32,
			0.625/ 0.6 / 0.65 / 1.0 ,
			0.65 / 0.6 / 0.675/ 0.71,
			0.65 / 0.0 / 0.675/ 0.32,
			0.675/ 0.39/ 0.7  / 0.71,
			0.7  / 0.39/ 0.775/ 0.57,
			0.775/ 0.0 / 0.8  / 0.57,
			0.8  / 0.27/ 0.85 / 0.4 ,
			0.85 / 0.27/ 0.9  / 1.0 ,
			0.9  / 0.27/ 0.925/ 0.73,
			0.925/ 0.27/ 0.95 / 0.33,
			0.925/ 0.66/ 0.95 / 0.73,
			0.95 / 0.66/ 0.975/ 1.0 ,
			0.95 / 0.0 / 0.975/ 0.33,
			0.975/ 0.0 / 1.0  / 1.0
		}{
			\draw[fill = myLightGray] (\x*\w,\y*\h) rectangle (\xx*\w,\yy*\h);
		}

		\foreach \x/\y/\xx/\yy in {
			0.0  /0.75/0.1  /1.0 ,
			0.025/0.28/0.075/0.53,
			0.05 /0.0 /0.1  /0.26,
			0.1  /0.0 /0.15 /0.6 ,
			0.15 /0.71/0.25 /1.0,
			0.175/0.34/0.275/0.59,
			0.2  /0.0 /0.3  /0.29,
			0.25 /0.73/0.35 /1.0,
			0.3  /0.42/0.4  /0.69,
			0.35 /0.0 /0.45 /0.37,
			0.4  /0.6 /0.45 /1.0 ,
			0.45 /0.32/0.575/0.58,
			0.5  /0.7 /0.6  /1.0 ,
			0.55 /0.0 /0.65 /0.25,
			0.625/0.32/0.675/0.6 ,
			0.65 /0.71/0.7  /1.0 ,
			0.675/0.0 /0.775/0.39,
			0.7  /0.57/0.8  /1.0 ,
			0.8  /0.0 /0.95 /0.27,
			0.8  /0.4 /0.85 /1.0 ,
			0.9  /0.73/0.95 /1.0 ,
			0.925/0.33/0.975/0.66
		}{
			\draw[fill = myGray] (\x*\w,\y*\h) rectangle (\xx*\w,\yy*\h);
		}

		\draw[dashed] (-0.1*\w,\h/4) -- (1.1*\w,\h/4);
		\draw[dashed] (-0.1*\w,\h/2) -- (1.1*\w,\h/2);
		\draw[dashed] (-0.1*\w,3*\h/4) -- (1.1*\w,3*\h/4);
		\draw[dashed] (-0.1*\w,\h) -- (1.1*\w,\h);
		
		\end{tikzpicture}
		\subcaption{The packing after the first shift.}
		\label{fig:simpleShiftB}
	\end{subfigure}%
	\hfill%
	\begin{subfigure}[t]{.24\textwidth}
		\centering
		\begin{tikzpicture}
		\pgfmathsetmacro{\w}{2.6}
		\pgfmathsetmacro{\h}{4.0}
		
		\draw[] (0,0) rectangle (\w,5*\h/4);
		
		\draw (0.19*\w,0.89*\h) -- (0.225*\w,0.89*\h);
		\node at (0.175*\w,0.89*\h) {\small 1};
		\draw[fill = myLightGray] (0.2*\w,0.75*\h) rectangle node[midway]{\small{1}}(0.25*\w,0.87*\h);
		
		\draw[fill = myLightGray] (0.3*\w,0.0*\h) rectangle node[midway]{\small 3} (0.35*\w,0.42*\h);
		\draw (0.325*\w,0.44*\h) -- (0.37*\w,0.44*\h);
		\node at (0.4*\w,0.44*\h) {\small 3};
		
		\draw[fill = myLightGray] (0.35*\w,0.89*\h) rectangle node[midway](A){} (0.4*\w,0.94*\h);
		\draw (0.37*\w,0.91*\h) -- (0.32*\w,0.86*\h) node{\small 3};
		\draw[fill = myLightGray] (0.35*\w,0.94*\h) rectangle node[midway]{\small 3} (0.4*\w,1.25*\h);
		
		\draw[fill = myLightGray] (0.8*\w,0.5*\h) rectangle node[midway]{\small 2}(0.85*\w,0.63*\h);
		
		\foreach \x/\y/\xx/\yy in {
			0.0  /0   /0.025/0.75,
			0.025/0.28/0.05 /0.5 ,
			0.025/0   /0.05 /0.28,
			0.05 /0.97/0.075/0.99,
			0.05 /0.75/0.075/0.97,
			0.075/0.26/0.1  /0.75,
			0.1  /0.85/0.15 /1.25,
			0.15 /0.0 /0.175/0.71,
			0.175/0.0 /0.2  /0.34,
			0.175/0.34/0.2  /0.46,
			0.2  /0.87/0.25 /0.92,
			0.25 /0.89/0.275/0.94,
			0.25 /0.75/0.275/0.89,
			0.275/0.31/0.3  /0.75,
			0.3  /0.42/0.35 /0.46,
			0.4  /0.52/0.45 /0.75,
			0.45 /0.83/0.5  /1.25,
			0.45 /0   /0.5  /0.32,
			0.5  /0.32/0.55 /0.44,
			0.5  /0   /0.55 /0.32,
			0.55 /0.87/0.575/0.94,
			0.55 /0.75/0.575/0.87,
			0.575/0.3 /0.6  /0.75,
			0.6  /0.5 /0.625/1.25,
			0.625/0.78/0.65 /0.85,
			0.625/0.85/0.65 /1.25,
			0.65 /0.32/0.675/0.43,
			0.65 /0.0 /0.675/0.32,
			0.675/0.44/0.7  /0.75,
			0.7  /0.57/0.775/0.75,
			0.775/0   /0.8  /0.57,
			0.85 /0.52/0.9  /1.25,
			0.9  /0.29/0.925/0.75,
			0.925/0.75/0.95 /0.81,
			0.925/0.81/0.95 /0.88,
			0.95 /0.91/0.975/1.25,
			0.95 /0   /0.975/0.33,
			0.975/0   /1    /1
		}{
			\draw[fill = myLightGray] (\x*\w,\y*\h) rectangle (\xx*\w,\yy*\h);
		}
		
		\foreach \x/\y/\xx/\yy in {
			0.025/0.5 /0.075/0.75,
			0.175/0.5 /0.275/0.75,
			0.3  /0.48/0.4  /0.75,
			0.45 /0.49/0.575/0.75,
			0.625/0.47/0.675/0.75,
			0.925/0.42/0.975/0.75
		}{
			\draw[fill = myGray] (\x*\w,\y*\h) rectangle (\xx*\w,\yy*\h);
		}
		
		\foreach \x/\y/\xx/\yy in {
			0.0  /0.75/0.1  /1.0 ,
			0.15 /0.71/0.25 /1.0,
			0.25 /0.73/0.35 /1.0,
			0.4  /0.6 /0.45 /1.0 ,
			0.5  /0.7 /0.6  /1.0 ,
			0.65 /0.71/0.7  /1.0 ,
			0.7  /0.57/0.8  /1.0 ,
			0.8  /0.4 /0.85 /1.0 ,
			0.9  /0.73/0.95 /1.0
		}{
			\begin{scope}[yshift= 0.25*\h cm]
			\draw[fill = myGray] (\x*\w,\y*\h) rectangle (\xx*\w,\yy*\h);
			\end{scope}
		}
		
		\foreach \x/\y/\xx/\yy in {
			0.05 /0.0 /0.1  /0.26,
			0.1  /0.0 /0.15 /0.6 ,
			0.2  /0.0 /0.3  /0.29,
			0.35 /0.0 /0.45 /0.37,
			0.55 /0.0 /0.65 /0.25,
			0.675/0.0 /0.775/0.39,
			0.8  /0.0 /0.95 /0.27
		}{
			\draw[fill = myGray] (\x*\w,\y*\h) rectangle (\xx*\w,\yy*\h);
		}
		
		\draw[dashed] (-0.1*\w,\h/4) -- (1.1*\w,\h/4);
		\draw[dashed] (-0.1*\w,\h/2) -- (1.1*\w,\h/2);
		\draw[dashed] (-0.1*\w,3*\h/4) -- (1.1*\w,3*\h/4);
		\draw[dashed] (-0.1*\w,\h) -- (1.1*\w,\h);
		
		\end{tikzpicture}
		\subcaption{The packing after the second shift.}
		\label{fig:simpleShiftC}
	\end{subfigure}%
	\hfill%
	\begin{subfigure}[t]{.24\textwidth}
		\centering
		\begin{tikzpicture}
		\pgfmathsetmacro{\w}{2.6}
		\pgfmathsetmacro{\h}{4.0}
		\draw[] (0,0) rectangle (\w,5*\h/4);
		
		\foreach \x/\y/\xx/\yy in {
			0.025/0.65/0.075/1.25,
			0.15 /0.82/0.25 /1.25,
			0.35 /0.85/0.4  /1.25,
			0.475/0.95/0.575/1.25,
			0.575/0.96/0.675/1.25,
			0.675/0.96/0.725/1.25,
			0.725/0.98/0.775/1.25,
			0.775/0.98/0.875/1.25,
			0.875/1   /0.975/1.25,
			0.175/0.0 /0.275/0.39,
			0.275/0.0 /0.375/0.37,
			0.45 /0.0 /0.55 /0.29,
			0.55 /0.0 /0.7  /0.27,
			0.7  /0.0 /0.75 /0.26,
			0.75 /0.0 /0.85 /0.25,
			0.875/0.0 /0.925/0.6 ,
			0.7  /0.42/0.75 /0.75,
			0.25 /0.5 /0.3  /0.75,
			0.3  /0.5 /0.4  /0.75,
			0.4  /0.49/0.525/0.75,
			0.525/0.48/0.625/0.75,
			0.625/0.47/0.675/0.75
		}{
			\draw[fill = myGray] (\x*\w,\y*\h) rectangle (\xx*\w,\yy*\h);
		}
		
		\foreach \x/\y/\xx/\yy in {
			0.0  /0.5 /0.025/1.25,
			0.025/0.5 /0.075/0.63,
			0.075/0.52/0.125/1.25,
			0.125/0.78/0.15 /1.25,
			0.25 /0.83/0.3  /1.25,
			0.3  /0.85/0.35 /1.25,
			0.4  /0.89/0.45 /1.25,
			0.45 /0.91/0.475/1.25,
			0.0  /0.0 /0.025/0.5,
			0.025/0.0 /0.075/0.46,
			0.075/0.0 /0.1  /0.46,
			0.1  /0.0 /0.15 /0.44,
			0.15 /0.0 /0.175/0.43,
			0.375/0.0 /0.4  /0.33,
			0.4  /0.0 /0.45 /0.32,
			0.85 /0.0 /0.875/0.57,
			0.925/0.0 /0.95 /0.71,
			0.95 /0.0 /0.975/0.75,
			0.975/0.0 /1    /1,
			0.125/0.57/0.2  /0.75,
			0.2  /0.52/0.25 /0.75,
			0.675/0.44/0.7  /0.75,
			0.75 /0.31/0.775/0.75,
			0.775/0.3 /0.8  /0.75,
			0.8  /0.29/0.825/0.75,
			0.825/0.26/0.85 /0.75,
			0.825/0.75/0.85 /0.88,
			0.85 /0.75/0.9  /0.92,
			0.9  /0.75/0.925/0.94,
			0.925/0.75/0.95 /0.94,
			0.95 /0.75/0.975/0.99
		}{
			\draw[fill = myLightGray] (\x*\w,\y*\h) rectangle (\xx*\w,\yy*\h);
		}

		\draw[dashed] (-0.1*\w,\h/4) -- (1.1*\w,\h/4);
		\draw[dashed] (-0.1*\w,\h/2) -- (1.1*\w,\h/2);
		\draw[dashed] (-0.1*\w,3*\h/4) -- (1.1*\w,3*\h/4);
		\draw[dashed] (-0.1*\w,\h) -- (1.1*\w,\h);
		
		\foreach \x/\y/\xx/\yy/\a in {
			0/0.5/0.125/1.25/2,
			0.975/0/1/1/,
			0.975/0/0.85/0.75/3,
			0.125/0.75/0.975/1.25/4
		}{
			\draw [red,very thick] (\x*\w, \y*\h) rectangle node[midway, black]{\small \a} (\xx*\w, \yy*\h);
		}
		
		\draw[red] (\w, 0.55*\h) -- (1.05*\w, 0.55*\h);
		\node at (1.075*\w, 0.55*\h){\small 1}; 
		
		\node at (0.45*\w,3*\h/8) {\small 5};
		
		\end{tikzpicture}
		\subcaption{The final reordered packing.}
		\label{fig:simpleShiftD}
	\end{subfigure}
	\caption{States of the item rearrangement. Dark rectangles represent tall items while light gray areas represent sliced items}
	\label{fig:simpleShift}
\end{figure}
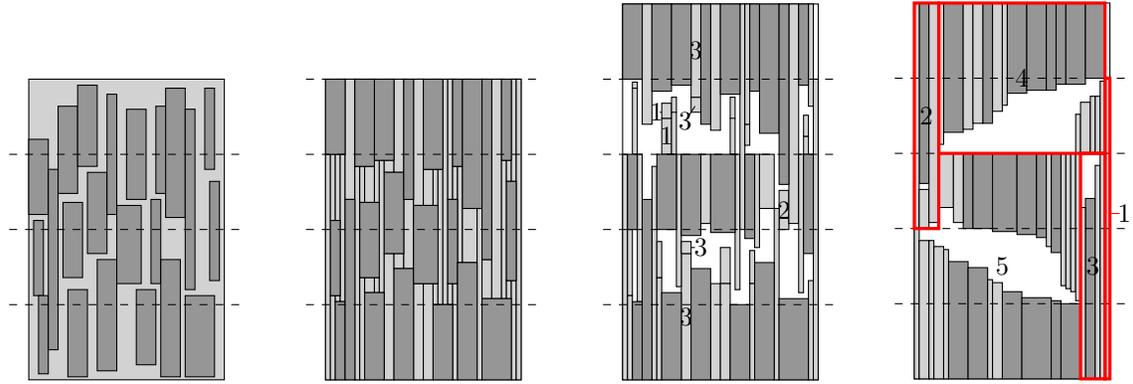

\begin{proof}
	In this proof, we will present a rearrangement strategy which provides the desired properties. 
	This strategy consists of two shifting steps and one reordering step. 
	In the shifting steps, we shift items in the vertical direction, while in the reordering step we change the item positions horizontally. 
	In the first shifting step, we ensure that tall items intersecting the horizontal lines $\nicefrac{1}{4} H$ or $\nicefrac{3}{4} H$ will touch the bottom or the top of the packing area, respectively. 
	In the second shift, we ensure that tall items not intersecting these lines have a common upper border as well. 
	Last, we reorder the items such that tall items with the same height are positioned next to each other if they have a common upper or lower border. 
	
	\begin{stepList}
		
		\item[First shift.]
		Note that there is no tall item completely below $\nicefrac{1}{4} H$ or completely above $\nicefrac{3}{4} H$ since each tall item has a height larger than $\nicefrac{1}{4} H$. 
		We shift each tall item $t$ intersecting the horizontal line $\nicefrac{1}{4} H$ down, such that its bottom border touches the bottom of the strip. 
		The sliced items below $t$ are shifted up exactly $h(t)$, such that they are now positioned above $t$.
		In the same way, we shift each tall item intersecting the horizontal line at $\nicefrac{3}{4} H$ but not the horizontal line at $\nicefrac{1}{4} H$ such that its upper border is positioned at $H$ and shift the sliced items down accordingly, see Figure~\ref{fig:simpleShiftB}. 
		
		\item[Introducing pseudo items.] 
		At this point, we introduce a set of containers for the sliced items, which we call pseudo items, see Figure~\ref{fig:simpleShiftB}. 
		We draw vertical lines at each left or right border of a tall item and erase these lines on any tall item.
		Each area between two consecutive lines which is bounded on top and bottom by a tall item or the packing area and contains sliced items represents a new item called pseudo item. 
		Note that no sliced item is intersecting any box border since they are positioned  on integral widths only. 
		Furthermore, when we shift a pseudo item, we shift all sliced items in this container accordingly and leave no item inside the container behind. 
		
		When constructing the pseudo items, we consider one special case. 
		Consider a tall item $t$ with height larger than $\nicefrac{3}{4} H$. 
		There can be no tall item positioned above or below $t$, and $t$ was shifted down. 
		For these items, we introduce one pseudo item of height $H$ and width $w(t)$ containing $t$ and all sliced items above. 
		Note that each pseudo item has a height, which is a multiple of $H/\NumGridpoints$.
		Furthermore, note that each tall or pseudo item touching the top or the bottom border of the packing area has a height larger than $\nicefrac{1}{4} H$.
		
		\item[Second shift.] 
		Next, we do a second shifting step consisting of three sub-steps. 
		First, we shift each tall or pseudo item intersected by the horizontal line at $\nicefrac{3}{4} H$ but not the horizontal line at $\nicefrac{1}{4} H$ exactly $\nicefrac{1}{4} H$ upwards. 
		Second, we shift each pseudo item positioned between the horizontal lines at $\nicefrac{1}{2} H$ and $\nicefrac{3}{4} H$, such that their lower border is positioned at the horizontal line $\nicefrac{3}{4} H$. 
		Last, we shift each tall or pseudo item intersected by the horizontal line at $\nicefrac{1}{2} H$ but not the horizontal line at $\nicefrac{1}{4} H$ or $\nicefrac{3}{4} H$ such that its upper border is positioned at the horizontal line $\nicefrac{3}{4} H$. 
		After this shifting, no item overlaps another item since we have shifted the items intersecting the line at $\nicefrac{3}{4} H$ exactly $\nicefrac{1}{4} H$, while each item below is shifted at most $\nicefrac{1}{4} H$.
		
		\item[Fusing pseudo items.]
		After the second shift, we will fuse and shift some pseudo items. 
		We want to establish the property that each tall and pseudo item has one border (upper or lower), which touches one of the horizontal lines at $0$, $\nicefrac{3}{4} H$, or $\nicefrac{5}{4} H$. 
		At the moment there can be some pseudo items between the horizontal lines $\nicefrac{1}{4} H$ and $\nicefrac{1}{2} H$, which do not touch one of the three lines.
		In the following, we study the three cases where those pseudo items can occur. 
		These items do only exist if there is a tall item touching the bottom of the packing and another tall item above this item with lower border at or below  $\nicefrac{1}{2} H$ before the second shifting step.
		Consider two consecutive vertical lines we had drawn to generate the pseudo items. 
		If a tall item overlaps the vertical strip between these lines, its right and left borders lie either on the strips borders or outside of the strip.
		
		\begin{caseList}
			\item 
			In the first considered case there are three tall items, $t_1$, $t_2$, and $t_3$ from bottom to top, which overlap the strip. 
			In this scenario $t_1$ must have its lower border at $0$, $t_2$ its upper border at $\nicefrac{3}{4} H$, and $t_3$ its upper border at $\nicefrac{5}{4} H$.
			As a consequence, there are at most two pseudo items: 
			One is positioned between $t_1$ and $t_2$, and the other between $t_2$ and $t_3$. 
			We will stack them, such that the lower border of the stack is positioned at $\nicefrac{3}{4} H$ and prove that this is possible without overlapping $t_3$. 
			The total height of both pseudo items is $H - h(t_1) - h(t_2) - h(t_3)$. 
			The total area not occupied by tall items is $H - h(t_1) - h(t_2) - h(t_3) +\nicefrac{1}{4} H$ since we have added $\nicefrac{1}{4} H$ to the packing height. 
			The distance between $t_1$ and $t_2$ is at most $\nicefrac{1}{4} H$ since $t_1$'s lower border is at $0$ and $t_2$'s upper border is at $\nicefrac{3}{4} H$ and both have a height larger than $\nicefrac{1}{4} H$. 
			Therefore, the distance between $t_2$ and $t_3$ is at least $H - h(t_1) - h(t_2) - h(t_3)$, see Figure~\ref{fig:simpleShiftC} at the items marked with 1.
			
			\item
			Now consider the case where there is one tall item $t_1$ touching the bottom, and one tall item $t_2$ with height at least $\nicefrac{1}{2} H$ touching $\nicefrac{5}{4} H$.
			Obviously, $t_2$ has a height of at most $\nicefrac{3}{4} H$. 
			Furthermore, there is at most one pseudo item, and it has to be positioned between $\nicefrac{1}{4} H$ and $\nicefrac{1}{2} H$.
			We shift this pseudo item up until its bottom touches $\nicefrac{1}{2} H$, see Figure~\ref{fig:simpleShiftC} at the item marked with 2. 
			This is possible without constructing any overlap, because the distance between $t_1$ and the horizontal line $\nicefrac{1}{2} H$ is less than $\nicefrac{1}{4} H$ and, therefore, the distance between the line $\nicefrac{1}{2} H$ and the lower border of the tall item is larger than the height of the pseudo item.
			
			After this step, we consider each tall item $t$ with height larger than $\nicefrac{1}{2} H$ touching $\nicefrac{5}{4} H$.
			We generate a new pseudo item with width $w(t)$ and height $\nicefrac{3}{4} H$, with upper border at $\nicefrac{5}{4} H$ and lower border at $\nicefrac{1}{2} H$, containing all pseudo items below $t$ touching $\nicefrac{1}{2} H$ with their lower border.
			
			\item
			In the last case we consider, there are two tall items $t_1$ and $t_2$ and two pseudo items; 
			one of the items $t_1$ and $t_2$ touches the top of the packing or the bottom, while the other ends at $\nicefrac{3}{4} H$. 
			Hence, the distance between the tall items has to be smaller than $\nicefrac{1}{4} H$. 
			Furthermore, one of the pseudo items has to touch the top or the bottom of the packing while the other is positioned between $t_1$ and $t_2$. 
			Since the distance between $t_1$ and $t_2$ is less than $\nicefrac{1}{4} H$ one of the distances between the packing border and the lower border of $t_1$ or the upper border of $t_2$ is at least $H - h(t_1) -h(t_2)$.
			Therefore, we can fuse both pseudo items by shifting the one between $t_1$ and $t_2$ such that it is positioned above or below the other one, see Figure~\ref{fig:simpleShiftC} at the items marked with 3.
		\end{caseList}
		
		\begin{observation}
			After the shifting and fusing, each tall or pseudo item touches one of the horizontal lines at $0$, $\nicefrac{3}{4} H$ or $\nicefrac{5}{4} H$.
		\end{observation}
		
		\item[Reordering the items.]
		In the last part of the rearrangement, we reorder the items horizontally to place pseudo and tall items with the same height next to each other. 
		In this reordering step, we create five areas each reserved for certain items. 
		To do so, we take vertical slices of the packing and move them to the left or the right in the strip. 
		A vertical slice is an area of the packing with width one and height of the considered packing area, i.e. $\nicefrac{5}{4} H$ in this case. 
		While rearranging these slices, it will never happen that two items overlap. 
		However, it can happen, that some of the tall items are placed fractionally afterwards. 
		This will be fixed in later steps.
		
		\begin{description}
			\item[Area 1:]
			First, we will extract all vertical slices containing (pseudo) items with height $H$. 
			Then, shifting all the remaining vertical slices to the left as much as possible, we create one box for pseudo items of height $H$ at the right, see Figure~\ref{fig:simpleShiftD} at Area 1. 
			In this area, we sort the pseudo items such that the pseudo items containing tall items with the same height are placed next to each other. 
			In this step, we did not place any tall item fractionally. 
			
			\item[Area 2:]
			Afterward, we take each vertical slice containing a (pseudo) item with height at least $\nicefrac{1}{2} H$ touching the horizontal line at $\nicefrac{5}{4} H$. 
			Remember, there might be pseudo items containing a tall item $t$ with height between $\nicefrac{1}{2}H$ and $\nicefrac{3}{4}H$.
			We shift these slices to the left of the packing and sort them in descending order of the tall items height $\iH{t}$, see Figure~\ref{fig:simpleShiftD} at Area 2.
			Afterward, we sort the pseudo items below these tall items, which are touching $\nicefrac{1}{2}H$ with their bottom in ascending order of their heights, which is possible without generating any overlapping.
			In this step, it can happen that we slice tall items which touch the bottom of the strip. 
			We will fix this slicing in one of the following steps, when we consider Area 5. 
			
			\item[Area 3:]
			Next, we look at vertical slices containing (pseudo) items $t$ with height at least $\nicefrac{1}{2} H$ touching the bottom of the strip. 
			We shift them to the right until they touch the Area 1 and sort these slices in ascending order of the heights $\iH{t}$, see Figure~\ref{fig:simpleShiftD} at Area 3. 
			Note that there are no pseudo or tall items with upper border at $\nicefrac{3}{4} H$ in these slices. 
			In this step, it can happen that we slice tall items touching the top of the packing. 
			This will be fixed in the next step.
			
			\item[Area 4:]
			Look at the area above $\nicefrac{3}{4} H$ and left of the Area 2 but right of Area 1, see Figure~\ref{fig:simpleShiftD} at Area 4. 
			In this area no item overlaps the horizontal line $\nicefrac{3}{4} H$. 
			Therefore, we have a rectangular area where each item either touches its bottom or its top and no item is intersected by the area's borders. 
			In~\cite{NadiradzeW16} it was shown that, in this case, we can sort the items touching the line $\nicefrac{3}{4} H$ in ascending order of their height and the items touching $\nicefrac{5}{4} H$ in descending order of heights and no item will overlap another item. 
			Now all items with the same height are placed next to each other, thus we have fixed the slicing of tall items on the top of the strip. 
			
			\item[Area 5:]
			In the last step, we will reorder the remaining items.
			Namely the items touching the bottom of the strip left of Area 3 and the items touching the horizontal line at $\nicefrac{3}{4} H$ with their top between Area 2 and Area 3.
			The items touching the bottom are sorted in descending order of their height and the items touching the horizontal line at $\nicefrac{3}{4} H$ are sorted in ascending order regarding their heights. 
			
			\begin{claim}
				After the reordering of Area 5 no item overlaps another.
			\end{claim}
			\begin{proofClaim}
				First, note that the items touching $\nicefrac{5}{4} H$ have a height of at most $\nicefrac{3}{4} H$. 
				Therefore, no item touching the bottom having height at most $\nicefrac{1}{2} H$ can overlap with these items. 
				Furthermore, note that before the reordering no item was overlapping another.
				Let us assume there are two items $b$ and $t$, which overlap at a point $(x,y)$ after this reordering. 
				Then all items left of $x$ touching $\nicefrac{3}{4} H$ have their lower border below $y$, while all items touching the bottom left of $x$ have their upper border above $y$. 
				Therefore, at every point left and right of $(x,y)$ in the Area 5 there is an item overlapping it. 
				Hence, the total width of items overlapping the horizontal line $y$ is larger than the width of the Area 5. 
				Therefore in the original ordering, there would have been items overlapping each other already since we did not add any items -- a contradiction.
				As a consequence in this new ordering, no two items overlap, which concludes the proof of the claim.
			\end{proofClaim}
		\end{description}
	\end{stepList}
	
	\begin{description}
		\item[Analyzing the number of constructed boxes.]
		In the last part of this proof, we analyze how many boxes we have created. 
		Each tall item with height at least $\nicefrac{3}{4} H$ touches the bottom and we create at most one box in Area 1 for each height. 
		Therefore, we create at most $\NumGridpoints/4$ boxes for these items. 
		Each tall item of height between $\nicefrac{1}{2} H$ and $\nicefrac{3}{4} H$ touches the bottom or the horizontal line $\nicefrac{5}{4} H$. 
		On each of these lines, we create at most one box for items with the same height. 
		Therefore, we create at most $2 \NumGridpoints/4$ boxes for these items. 
		Last, each tall item with height larger than $\nicefrac{1}{4} H$ but smaller than $\nicefrac{1}{2} H$ either touches the bottom of the packing, the horizontal line $\nicefrac{3}{4} H$ or the horizontal line $\nicefrac{5}{4} H$. 
		At each of these lines, we create at most one box for each height. 
		Therefore, we create at most $3 \NumGridpoints/4$ of these boxes. 
		In total, we create at most $\frac{3}{2}\NumGridpoints$ boxes for tall items.  
		
		Let us consider the number of boxes for sliced items. 
		Each pseudo item's height is a multiple of $H/\NumGridpoints$. 
		Therefore, we have at most $\NumGridpoints$ different sizes for pseudo items. 
		There are at most $4$ boxes for each height less than $\nicefrac{1}{4} H$. 
		One is touching $H$ with its top border in Area 1, one is touching $\nicefrac{3}{4} H$ with its bottom border in Area 4, one is touching $\nicefrac{3}{4} H$ with its top border in Area 5, and one is touching $\nicefrac{1}{2} H$ with its bottom border in Area 2. 
		Furthermore, there are at most $3$ boxes for each size between $\nicefrac{1}{4} H$ and $\nicefrac{1}{2} H$. 
		One is touching $\nicefrac{5}{4} H$ with its top border in Area 4, one is touching $\nicefrac{3}{4} H$ with its top border in Area 5, and one is touching $0$ with its bottom border in Area 5. 
		Additionally, there are at most $2$ boxes for each pseudo item size larger than $\nicefrac{1}{2} H$. 
		One is touching $\nicefrac{5}{4} H$ with its top border in Area 2, the other is touching $0$ with its bottom border in Area 3. 
		Last there is only one pseudo item with height larger than $\nicefrac{3}{4} H$ in Area 1. It has height $H$. 
		Since the grid is arithmetically, we have at most $\NumGridpoints/4$ sizes with height at most $\nicefrac{1}{4} H$, $\NumGridpoints/4$ sizes between $\nicefrac{1}{4} H$ and $\nicefrac{1}{2} H$ and at most $\nicefrac{1}{4}\NumGridpoints$ sizes between $\nicefrac{1}{2} H$ and $\nicefrac{3}{4} H$. 
		Therefore, we create at most $4 \cdot \nicefrac{1}{4}\NumGridpoints + 3 \cdot \nicefrac{1}{4}\NumGridpoints + 2 \cdot \nicefrac{1}{4}\NumGridpoints + 1 = \frac{9}{4}\NumGridpoints +1$ boxes for sliced items.
	\end{description}
\end{proof}

In this section, we have proven that in this simplified case it is possible to reorder the items such that they have a nice structure, where there are at most few boxes for each tall item height containing only items with this height.
However, we are interested in a simple structure for a packing, where no item is sliced.
The main key to find such a structure is presented in the next section. 

\section{Reordering in the General Case}
\label{sec:ReorderingInTheGeneralCase}
In the structural result, all items have to be placed integral; thus, we cannot slice the non-tall items as we do in the previews chapter. 
Nevertheless, we may still slice certain narrow items, because for them we have techniques to place them integrally afterward. 
We call these sliceable items vertical items.
As proven in~\cite{JansenR16}, it is possible to partition the packing area into a constant number of rectangular subareas, called boxes, such that boxes containing tall items will contain tall and vertical items only.
In this section, we will consider such boxes and show that it is possible to reorder the items in these boxes similarly as in Section~\ref{sec:simpleCase}.
A difficulty here is that up to three tall items can overlap the left and right box border. 
Since we do not want to slice these items, we fix their position and call them unmovable.
These unmovable items complicate the reordering in the box. 
We overcome this difficulty, by a more careful reordering of the items, while the shifting steps remain the same. 

The most interesting boxes are the ones with a height of at least $\nicefrac{3}{4}H$.
In these boxes, there can be up to three tall items above each other, while in smaller boxes there can be at most two. 
How to handle boxes with at most two tall items on top of each other was already proven in~\cite{JansenR16}. 
They are partitioned into fewer sub boxes than the boxes with height at least $\nicefrac{3}{4}H$. 
While the boxes larger than $\nicefrac{1}{2}H$ still need an extra height of $\nicefrac{1}{4}H$ to partition them into \subboxes, for the boxes smaller than $\nicefrac{1}{2}H$ their original height is sufficient.

The reordering in this section will be used in the proof of the structure Lemma~\ref{lma:structureLemma}. 
Before the reordering technique from this chapter will be used to prove the structural Lemma~\ref{lma:structureLemma} the boxes will already have certain properties.
Therefore in the following, we assume the following two properties. First boxes with height larger than $\nicefrac{3}{4}H$ can be extended by $\nicefrac{1}{4}H$ since there is no item positioned in this area; and second, no tall item overlaps the left or right box border at or above $S(B) + \iH{B} - \nicefrac{1}{4}H$, where $S(B)$ is the y-coordinate of the lower box border.
The details why we can assume these properties can be found in the proof of Lemma~\ref{lma:structureLemma}.

In the reordering strategy discussed in this section, we will use a version of a result from~\cite{JansenR16} explaining the reordering of boxes in which each item either touches the top or the bottom of the box, and there are up to two unmovable items at each box border:

\begin{lemma}[\cite{JansenR16}]
	\label{lma:reorderingPreviusly}
	Let $B$ be an area where each item either touches the bottom or the top of this area, with at most two unmovable items on each side. 
	Let $S_P$ be the number of heights of pseudo items in this area, $S_T$ be the number of heights of tall items in this area and $S_{P\cup T}$ be the number of heights in $P \cup T$. 
	It is possible to rearrange the items in this area creating at most $4S_TS_{P\cup T} $ boxes for tall items plus one for each unmovable item and at most $4S_PS_{P\cup T}$ boxes for pseudo items plus the pseudo items for extending the unmovable items.
\end{lemma}

In this section, we again assume that all tall items are placed on an arithmetic grid with $\NumGridpoints+1$ horizontal grid lines with distance $H/\NumGridpoints$. 
Furthermore, we assume that the box also starts and ends at these grid lines. 
Let $S(B)$ be the y-coordinate of the lower box border.

\begin{lemma}
	\label{lma:reorderingGeneral}
	Let $B$ be a box with height $\iH{B} > \nicefrac{3}{4} H$, such that no tall item overlaps the left or right box border at (or above) $S(B) + \iH{B} - \nicefrac{1}{4}H$.
	By adding at most $\nicefrac{1}{4} H$ to $B$'s height, we can rearrange the items in $B$ such that we generate at most $\mathcal{O}(\NumGridpoints^2)$ boxes for tall and at most $\mathcal{O}(\NumGridpoints^2)$ boxes for vertical items without moving the unmovable items. 
	The vertical items are sliced while each tall item is placed as a whole.
\end{lemma}

\begin{proof}
	In this proof, we present a reordering strategy for the items in these boxes.  
	Let $\iH{B}$ be the height of $B$. 
	For convenience, we will assume that the lower border of $B$ is at $0$. 
	If not, we shift all horizontal lines accordingly. 
	Notice that there are at most two tall items overlapping the left or right box border since we assumed that there is no tall item overlapping the border at $\iH{B} - \nicefrac{1}{4}H$.
	In the first step, we shift all movable items according to the first and second shifting step seen in the proof of Lemma~\ref{lma:simpleCase}. 
	However, the reordering works differently than before. 
	If there are items taller than $\nicefrac{1}{2} H$ touching the top of the box, we find the leftmost item $l$ and the rightmost item $r$ of them. 
	We introduce three areas: one left of $l$, one between $l$ and $r$ and one right of $r$. 
	While we reorder the leftmost and the rightmost area with known techniques, we need a new trick to reorder the middle part.
	
	
	\begin{stepList}
		\item[Shifting the items.]
		
		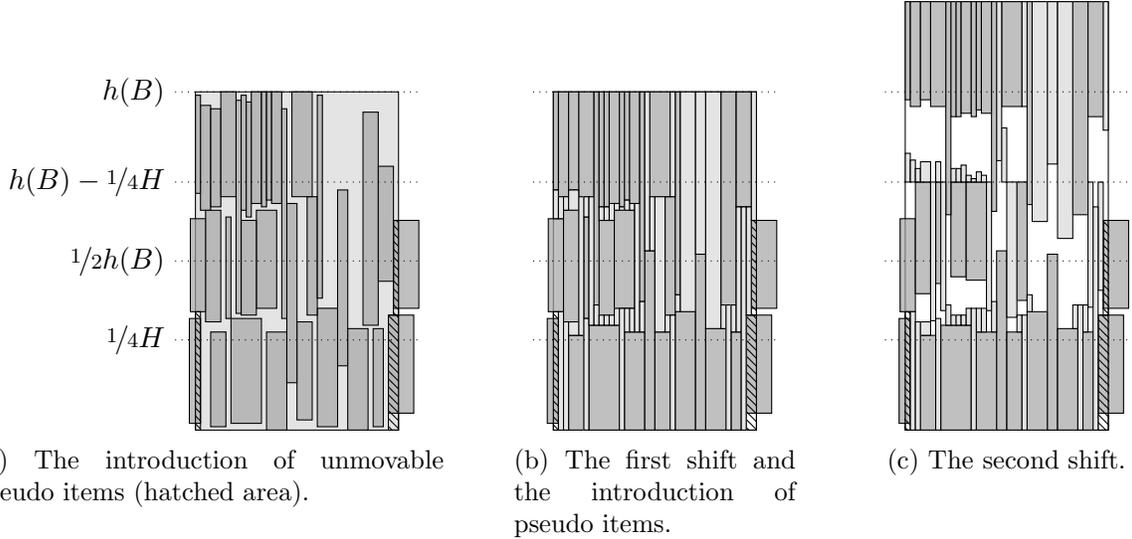
\begin{figure}
			\centering
			\begin{subfigure}[t]{0.4\textwidth}
				\centering
				\begin{tikzpicture}
				\pgfmathsetmacro{\w}{2.7}
				\pgfmathsetmacro{\h}{4.8}
				\pgfmathsetmacro{\hprime}{4.5}
				
				\drawVerticalItem{0*\w}{0}{1*\w}{\hprime};
				
				\foreach \x/\y/\xx/\yy in {
					-0.03 /0.02/0.025/0.33,
					-0.025/0.35/0.05/0.625,
					1.075 /0.05/0.95/0.34,
					1.1   /0.36/0.975/0.62,
					0.175 /0.02 /0.325/0.33,
					0.35  /0.0  /0.45 /0.29,
					0.45  /0.14 /0.5  /0.67,
					0.5   /0.03 /0.575/0.32,
					0.6   /0.01 /0.7  /0.36,
					0.7   /0.19 /0.75 /0.71,
					0.75  /0.0  /0.85 /0.3,
					0.875 /0.01 /0.925/0.3,
					0.3   /0.36 /0.4  /0.65,
					0.9   /0.44 /0.975/0.78,
					0.05  /0.32 /0.125/0.65,
					0.075 /0.29 /0.15 /0.01,
					0.15  /0.33 /0.175/0.63,
					0.225 /0.34 /0.3  /0.62,
					0.55  /0.34 /0.6  /0.69,
					0.0   /0.7  /0.025/0.99,
					0.025 /0.65 /0.075/0.96,
					0.075 /0.66 /0.125/0.95,
					0.125 /0.69 /0.2  /1,
					0.2   /0.345/0.225/0.975,
					0.225 /0.65 /0.25 /0.99,
					0.25  /0.63 /0.275/0.97,
					0.275 /0.67 /0.325/1,
					0.325 /0.66 /0.35 /1,
					0.35  /0.68 /0.375/1,
					0.375 /0.67 /0.425/1,
					0.475 /0.69 /0.575/1,
					0.425 /0.33 /0.45 /0.95,
					0.6   /0.39 /0.625/0.99,
					0.825 /0.31 /0.9  /0.94
				}
				{
					\drawTallItem{\x*\w}{\y*\hprime}{\xx*\w}{\yy*\hprime};
				}
				
				\foreach \x/\y/\xx/\yy in {
					0    /0   /0.025/0.35,
					0.950/0   /1    /0.34,
					0.975/0.34/1    /0.62
				}{
					\draw[pattern = north west lines] (\x*\w,\y*\hprime) rectangle (\xx*\w,\yy*\hprime);
				}
				
				\foreach \y/\z in {
					0.5*\hprime 		/ $\nicefrac{1}{2}\iH{B}$,
					0.25*\h     		/ $\nicefrac{1}{4}H$,
					\hprime - 0.25*\h	/ $\iH{B} - \nicefrac{1}{4}H$,
					\hprime				/ $\iH{B}$
				}{
					\draw[dotted] (-0.1*\w,\y)node[left]{\z} --(1.1*\w,\y);
				}
				
				\end{tikzpicture}
				\caption{The introduction of unmovable pseudo items (hatched area).}
				\label{fig:sub:IntroducingPseudoItems}
			\end{subfigure}
			\hfill
			\centering
			\begin{subfigure}[t]{0.24\textwidth}
				\centering
				\begin{tikzpicture}
				
				\pgfmathsetmacro{\w}{2.7}
				\pgfmathsetmacro{\h}{4.8}
				\pgfmathsetmacro{\hprime}{4.5}
				\draw (0*\w,0) rectangle (1*\w,\hprime);

				\foreach \x/\y/\xx/\yy in {
					-0.03 /0.02/0.025/0.33,
					-0.025/0.35/0.05/0.625,
					1.075 /0.05/0.95/0.34,
					1.1   /0.36/0.975/0.62,
					0.175 /0.0 /0.325/0.31,
					0.35  /0.0 /0.45 /0.29,
					0.45  /0.0 /0.5  /0.53,
					0.5   /0.0 /0.575/0.29,
					0.6   /0.0 /0.7  /0.35,
					0.7   /0.0 /0.75 /0.52,
					0.75  /0.0 /0.85 /0.3,
					0.875 /0.0 /0.925/0.29,
					0.075 /0.0 /0.15 /0.28,
					0.05  /0.32 /0.125/0.65,
					0.15  /0.33 /0.175/0.63,
					0.225 /0.34 /0.3  /0.62,
					0.3   /0.36 /0.4  /0.65,
					0.55  /0.34 /0.6  /0.69,
					0.0   /0.71 /0.025/1,
					0.025 /0.69 /0.075/1,
					0.075 /0.71 /0.125/1,
					0.125 /0.69 /0.2  /1,
					0.2   /0.37 /0.225/1,
					0.225 /0.66 /0.25 /1,
					0.25  /0.66 /0.275/1,
					0.275 /0.67 /0.325/1,
					0.325 /0.66 /0.35 /1,
					0.35  /0.68 /0.375/1,
					0.375 /0.67 /0.425/1,
					0.475 /0.69 /0.575/1,
					0.425 /0.38 /0.45 /1,
					0.6   /0.40 /0.625/1,
					0.825 /0.37 /0.9  /1,
					0.9   /0.66 /0.975/1
				}
				{
					\drawTallItem{\x*\w}{\y*\hprime}{\xx*\w}{\yy*\hprime};
				}
				
				\foreach \x/\y/\xx/\yy in {
					0    /0   /0.025/0.35,
					0.950/0   /1    /0.34,
					0.975/0.34/1    /0.62
				}{
					\draw[pattern = north west lines] (\x*\w,\y*\hprime) rectangle (\xx*\w,\yy*\hprime);
				}

				\foreach \x/\y/\xx/\yy in {
					0.0  /0.625/0.025/0.71,
					0.025/0.625/0.05 /0.69,
					0.025/0.0  /0.05 /0.35,
					0.05 /0.65 /0.075/0.69,
					0.05 /0.0  /0.075/0.32,
					0.075/0.65 /0.125/0.71,
					0.075/0.28 /0.125/0.32,
					0.125/0.28 /0.15 /0.69,
					0.15 /0.63 /0.175/0.69,
					0.15 /0.0  /0.175/0.33,
					0.175/0.31 /0.2  /0.69,
					0.2  /0.31 /0.225/0.37,
					0.225/0.31 /0.25 /0.34,
					0.225/0.62 /0.25 /0.66,
					0.25 /0.31 /0.275/0.34,
					0.25 /0.62 /0.275/0.66,
					0.275/0.31 /0.3  /0.34,
					0.275/0.62 /0.3  /0.67,
					0.3  /0.31 /0.325/0.36,
					0.3  /0.65 /0.325/0.67,
					0.325/0.0  /0.35 /0.36,
					0.325/0.65 /0.35 /0.66,
					0.35 /0.29 /0.375/0.36,
					0.35 /0.65 /0.375/0.68,
					0.375/0.29 /0.4  /0.36,
					0.375/0.65 /0.4  /0.67,
					0.4  /0.29 /0.425/0.67,
					0.425/0.29 /0.45 /0.38,
					0.45 /0.53 /0.475/1,
					0.475/0.53 /0.5  /0.69,
					0.5  /0.29 /0.55 /0.69,
					0.55 /0.29 /0.575/0.34,
					0.575/0.0  /0.6  /0.34,
					0.575/0.69 /0.6  /1,
					0.6  /0.35 /0.625/0.4,
					0.625/0.35 /0.7  /1,
					0.7  /0.52 /0.75 /1,
					0.75 /0.30 /0.825/1,
					0.825/0.30 /0.85 /0.37,
					0.85 /0.0  /0.875/0.37,
					0.875/0.29 /0.9  /0.37,
					0.9  /0.29 /0.925/0.66,
					0.925/0.0  /0.95 /0.66,
					0.95 /0.34 /0.975/0.66,
					0.975/0.62 /1    /1
				}
				{
					\drawVerticalItem{\x*\w}{\y*\hprime}{\xx*\w}{\yy*\hprime};
				}

				\foreach \y/\z in {
					0.5*\hprime 		/ $\nicefrac{1}{2}\iH{B}$,
					0.25*\h     		/ $\nicefrac{1}{4}H$,
					\hprime - 0.25*\h	/ $\nicefrac{1}{2}\iH{B}$,
					\hprime				/ $\iH{B}$
				}{
					\draw[dotted] (-0.1*\w,\y) -- (1.1*\w,\y);
				}
				
				
				\end{tikzpicture}
				\caption{The first shift and the introduction of pseudo items.}
				\label{fig:sub:TheFirstShift}
			\end{subfigure}
			\hfill
			\begin{subfigure}[t]{0.24\textwidth}
				\centering
				\begin{tikzpicture}
				
				\pgfmathsetmacro{\w}{2.7}
				\pgfmathsetmacro{\h}{4.8}
				\pgfmathsetmacro{\hprime}{4.5}
				
				\draw (0*\w,0) rectangle (1*\w,\hprime +\h/4);
				
				\foreach \x/\y/\xx/\yy in {
					-0.03 /0.02/0.025/0.33,
					-0.025/0.35/0.05/0.625,
					1.075 /0.05/0.95/0.34,
					1.1   /0.36/0.975/0.62
				}
				{
					\drawTallItem{\x*\w}{\y*\hprime}{\xx*\w}{\yy*\hprime};
				}
				
				\foreach \x/\y/\xx/\yy in {
					0    /0   /0.025/0.35,
					0.950/0   /1    /0.34,
					0.975/0.34/1    /0.62
				}{
					\draw[pattern = north west lines] (\x*\w,\y*\hprime) rectangle (\xx*\w,\yy*\hprime);
				}
				
				\foreach \x/\y/\xx/\yy in {
					0.175 /0.0 /0.325/0.31,
					0.35  /0.0 /0.45 /0.29,
					0.45  /0.0 /0.5  /0.53,
					0.5   /0.0 /0.575/0.29,
					0.6   /0.0 /0.7  /0.35,
					0.7   /0.0 /0.75 /0.52,
					0.75  /0.0 /0.85 /0.3,
					0.875 /0.0 /0.925/0.29,
					0.075 /0.0 /0.15 /0.28
				}
				{
					\drawTallItem{\x*\w}{\y*\hprime}{\xx*\w}{\yy*\hprime};
				}
				
				\foreach \x/\y/\xx/\yy in {
					0.05 /0.67/0.125/1,
					0.15 /0.70/0.175/1,
					0.225/0.72/0.3  /1,
					0.3  /0.71/0.4  /1,
					0.55 /0.65/0.6  /1
				}
				{
					\drawTallItem{\x*\w}{\y*\hprime-\h/4}{\xx*\w}{\yy*\hprime-\h/4};
				}
				
				\foreach \x/\y/\xx/\yy in {
					0.0   /0.71 /0.025/1,
					0.025 /0.69 /0.075/1,
					0.075 /0.71 /0.125/1,
					0.125 /0.69 /0.2  /1,
					0.2   /0.37 /0.225/1,
					0.225 /0.66 /0.25 /1,
					0.25  /0.66 /0.275/1,
					0.275 /0.67 /0.325/1,
					0.325 /0.66 /0.35 /1,
					0.35  /0.68 /0.375/1,
					0.375 /0.67 /0.425/1,
					0.475 /0.69 /0.575/1,
					0.425 /0.38 /0.45 /1,
					0.6   /0.40 /0.625/1,
					0.825 /0.37 /0.9  /1,
					0.9   /0.66 /0.975/1
				}
				{
					\drawTallItem{\x*\w}{\y*\hprime+\h/4}{\xx*\w}{\yy*\hprime+\h/4};
				}

				\foreach \x/\y/\xx/\yy in {
					0.45 /0.53/0.475/1,
					0.575/0.69/0.6  /1,
					0.625/0.35/0.7  /1,
					0.7  /0.52 /0.75 /1,
					0.75 /0.30/0.825/1,
					0.975/0.62/1    /1
				}
				{
					\drawVerticalItem{\x*\w}{\y*\hprime+\h/4}{\xx*\w}{\yy*\hprime+\h/4};
				}
				
				\foreach \x/\y/\xx/\yy in {
					0.0  /1/0.025/1.085,
					0.025/1/0.05 /1.065,
					0.05 /1/0.075/1.04 ,
					0.075/1/0.125/1.06 ,
					0.15 /1/0.175/1.06 ,
					0.225/1/0.25 /1.04 ,
					0.25 /1/0.275/1.04 ,
					0.275/1/0.3  /1.05 ,
					0.3  /1/0.325/1.02 ,
					0.325/1/0.35 /1.01 ,
					0.35 /1/0.375/1.03 ,
					0.375/1/0.4  /1.02 ,
					0.475/1/0.5  /1.16 
				}
				{
					\drawVerticalItem{\x*\w}{\y*\hprime-\h/4}{\xx*\w}{\yy*\hprime-\h/4};
				}
				
				\foreach \x/\y/\xx/\yy in {
					0.125/0.59/0.15 /1,
					0.175/0.62/0.2  /1,
					0.4  /0.62/0.425/1,
					0.5  /0.60/0.55 /1,
					0.9  /0.63/0.925/1,
					0.95 /0.68/0.975/1
				}
				{
					\drawVerticalItem{\x*\w}{\y*\hprime-\h/4}{\xx*\w}{\yy*\hprime-\h/4};
				}
				
				\foreach \x/\y/\xx/\yy in {
					0.075/0.28/0.125/0.32,
					0.2  /0.31/0.225/0.37,
					0.225/0.31/0.25 /0.34,
					0.25 /0.31/0.275/0.34,
					0.275/0.31/0.3  /0.34,
					0.3  /0.31/0.325/0.36,
					0.35 /0.29/0.375/0.36,
					0.375/0.29/0.4  /0.36,
					0.425/0.29/0.45 /0.38,
					0.55 /0.29/0.575/0.34,
					0.6  /0.35/0.625/0.4 ,
					0.825/0.30/0.85 /0.37,
					0.875/0.29/0.9  /0.37
				}
				{
					\drawVerticalItem{\x*\w}{\y*\hprime}{\xx*\w}{\yy*\hprime};
				}
				
				\foreach \x/\y/\xx/\yy in {
					0.025/0.0/0.05 /0.35,
					0.05 /0.0/0.075/0.32,
					0.15 /0.0/0.175/0.33,
					0.325/0.0/0.35 /0.36,
					0.575/0.0/0.6  /0.34,
					0.85 /0.0/0.875/0.37,
					0.925/0.0/0.95 /0.66
				}
				{
					\drawVerticalItem{\x*\w}{\y*\hprime}{\xx*\w}{\yy*\hprime};
				}
				
				\foreach \y/\z in {
					0.5*\hprime 		/ $\nicefrac{1}{2}\iH{B}$,
					0.25*\h     		/ $\nicefrac{1}{4}H$,
					\hprime - 0.25*\h	/ $\nicefrac{1}{2}\iH{B}$,
					\hprime				/ $\iH{B}$
				}{
					\draw[dotted] (-0.1*\w,\y) --(1.1*\w,\y);
				}
				
				
				\end{tikzpicture}
				\caption{The second shift.}
				\label{fig:sub:TheSecondShift}
			\end{subfigure}
			\caption{An overview of the shifting steps. First we define unmovable pseudo items; second we shift all the tall items intersecting the horizontal line $\nicefrac{1}{4}H$ and $\iH{B}- \nicefrac{1}{4}H$ to the bottom and top ans introduce pseudo items; last we extend the box by $\nicefrac{1}{4}H$, such that all tall items touch one of three horizontal lines.}
			\label{fig:schiftGeneral1}
		\end{figure}
		
		Let us first consider the unmovable items on the left box side. 
		There can be two of these, one overlapping the box at $\nicefrac{1}{4} H$, the other at $\nicefrac{1}{2}\iH{B}$. 
		In the case that there is just one item, we extend it to the bottom of the strip, generating one unmovable pseudo item. 
		If there are two items, $t_1$ at $\nicefrac{1}{4} H$ and $t_2$ at $\nicefrac{1}{2}\iH{B}$, we extend $t_1$ to the bottom. 
		Then, depending on which of the items $t_1$ or $t_2$ has its right border farther on the left, we extend $t_1$ to the bottom of $t_2$ or $t_2$ to the top of $t_1$, i.e., the one whose appearance inside the box is more narrow is extended, see Figure~\ref{fig:sub:IntroducingPseudoItems} at the hatched areas. 
		We do the same on the right side of the box. 
		
		Next, we perform the first shifting step, which works analogue to the one in the simple case with exception for the unmovable items which will not be shifted, see Figure~\ref{fig:sub:TheFirstShift}. 
		First, we shift each movable item crossed by the line $\nicefrac{1}{4} H$ down to the bottom of the box. 
		Afterward, we shift each movable item crossed by the line $\iH{B} - \nicefrac{1}{4} H$ to the top of the box. 
		We introduce pseudo items as described in the proof of Lemma~\ref{lma:simpleCase}, with the difference that each tall item $t$ with height larger than $3\iH{B}/4$ generates a pseudo item with height $\iH{B}$ and width $w(t)$. 
		
		Next, we do the second shifting step, see Figure~\ref{fig:sub:TheSecondShift}. 
		Each tall and pseudo item cut by the line $\iH{B} - \nicefrac{1}{4} H$ is shifted up exactly $\nicefrac{1}{4} H$. 
		Remember, there is no unmovable item intersecting this line.
		We shift each pseudo item between the lines $\iH{B} - \nicefrac{1}{4} H$ and $\nicefrac{1}{2}\iH{B}$ such that its bottom touches $\iH{B} - \nicefrac{1}{4} H$. 
		Afterward, we shift each not shifted movable tall and pseudo item crossed by the line $\nicefrac{1}{2}\iH{B}$ such that its top touches $\iH{B} - \nicefrac{1}{4} H$. 
		Again, no item overlaps another after this shift.
		
		Last, we will fuse the pseudo items as described in Lemma~\ref{lma:simpleCase}, see Figure~\ref{fig:sub:FusingPseudoItems}. 
		The fusion is possible since the considered distance in each of the Cases 1 to 3 is at most $\nicefrac{1}{4} H$, too. 
		After this fusion, we can assume that each item $t$ with height larger than $\nicefrac{1}{2}\iH{B}$ touching $\iH{B} +\nicefrac{1}{4} H$ has a height of exactly $\nicefrac{1}{2}\iH{B} + \nicefrac{1}{4} H$, see Case 2 in the proof of Lemma~\ref{lma:simpleCase}.

		Furthermore, we can assume that each item $t$ touching the bottom with height taller than $\nicefrac{1}{2}\iH{B}$ has height $\iH{B}-\nicefrac{1}{4} H$:
		There can be at most two items above $t$, one tall item and one pseudo item. 
		The pseudo item has its lower border at $\iH{B} -\nicefrac{1}{4} H$. 
		Therefore, we can extend the item $t$ to the horizontal line $\iH{B}-\nicefrac{1}{4} H$. 
		
		After this shift, each movable item has one border at one of the following horizontal lines $0$, $\iH{B}-\nicefrac{1}{4} H$, or $\iH{B} + \nicefrac{1}{4} H$. 
		Furthermore, only (pseudo) items with height $\nicefrac{1}{2}\iH{B} + \nicefrac{1}{4} H$ or larger are crossing the line $\iH{B}-\nicefrac{1}{4} H$.
		
		\item[Reordering.]
		Let us assume for simplicity that there is no (pseudo) item with height $\iH{B}$ in $B$. 
		Later we will see what happens if there are any of these ones. 
		In the following, we will reorder the items step by step, by considering a constant number of smaller subareas of the box.
		We number these subareas from one to nine.
		These subareas are generated symmetrically on the left and the right side of the box and we call them $B_{l,i}$ and $B_{r,i}$ accordingly for the $i$th subarea.
		In the following, we will describe the steps only for the boxes $B_{l,i}$. 
		
		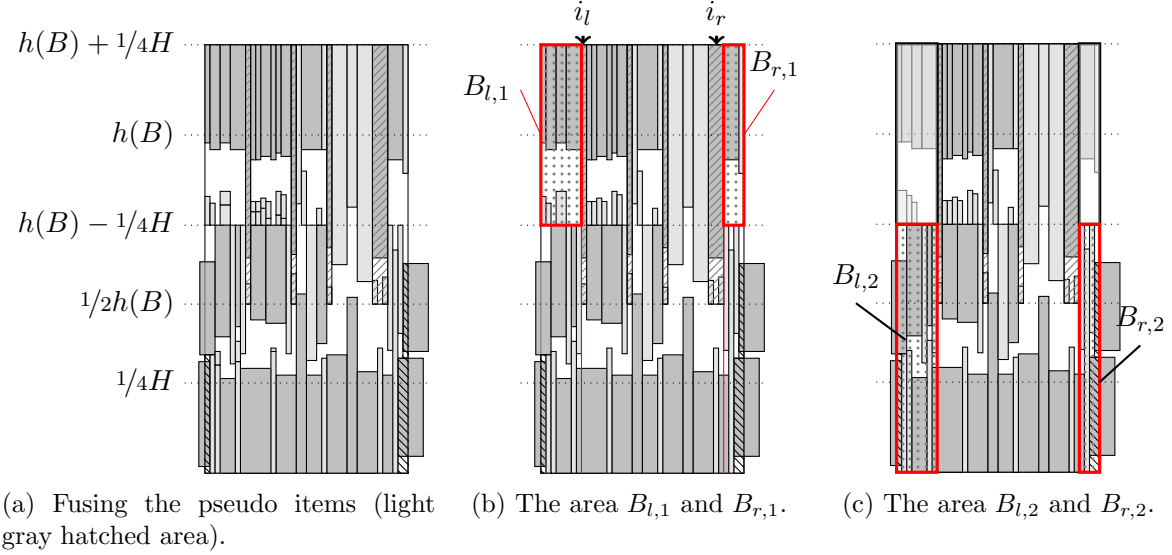
\begin{figure}
			\centering
			\begin{subfigure}[t]{0.37\textwidth}
				\centering
				\begin{tikzpicture}
				\pgfmathsetmacro{\w}{2.7}
				\pgfmathsetmacro{\h}{4.8}
				\pgfmathsetmacro{\hprime}{4.5}

				\draw (0*\w,0) rectangle (1*\w,\hprime +\h/4);
				
				\foreach \x/\y/\xx/\yy in {
					-0.03 /0.02/0.025/0.33,
					-0.025/0.35/0.05/0.625,
					1.075 /0.05/0.95/0.34,
					1.1   /0.36/0.975/0.62
				}
				{
					\drawTallItem{\x*\w}{\y*\hprime}{\xx*\w}{\yy*\hprime};
				}
				
				\foreach \x/\y/\xx/\yy in {
					0    /0   /0.025/0.35,
					0.950/0   /1    /0.34,
					0.975/0.34/1    /0.62
				}{
					\draw[pattern = north west lines] (\x*\w,\y*\hprime) rectangle (\xx*\w,\yy*\hprime);
				}
				
				\foreach \x/\y/\xx/\yy in {
					0.175 /0.0 /0.325/0.31,
					0.35  /0.0 /0.45 /0.29,
					0.45  /0.0 /0.5  /0.53,
					0.5   /0.0 /0.575/0.29,
					0.6   /0.0 /0.7  /0.35,
					0.7   /0.0 /0.75 /0.52,
					0.75  /0.0 /0.85 /0.3,
					0.875 /0.0 /0.925/0.29,
					0.075 /0.0 /0.15 /0.28
				}
				{
					\drawTallItem{\x*\w}{\y*\hprime}{\xx*\w}{\yy*\hprime};
				}
				
				\foreach \x/\y/\xx/\yy in {
					0.05 /0.67/0.125/1,
					0.15 /0.70/0.175/1,
					0.225/0.72/0.3  /1,
					0.3  /0.71/0.4  /1,
					0.55 /0.65/0.6  /1
				}
				{
					\drawTallItem{\x*\w}{\y*\hprime-\h/4}{\xx*\w}{\yy*\hprime-\h/4};
				}
				
				\foreach \x/\y/\xx/\yy in {
					0.0   /0.71 /0.025/1,
					0.025 /0.69 /0.075/1,
					0.075 /0.71 /0.125/1,
					0.125 /0.69 /0.2  /1,
					0.2   /0.37 /0.225/1,
					0.225 /0.66 /0.25 /1,
					0.25  /0.66 /0.275/1,
					0.275 /0.67 /0.325/1,
					0.325 /0.66 /0.35 /1,
					0.35  /0.68 /0.375/1,
					0.375 /0.67 /0.425/1,
					0.475 /0.69 /0.575/1,
					0.425 /0.38 /0.45 /1,
					0.6   /0.40 /0.625/1,
					0.825 /0.37 /0.9  /1,
					0.9   /0.66 /0.975/1
				}
				{
					\drawTallItem{\x*\w}{\y*\hprime+\h/4}{\xx*\w}{\yy*\hprime+\h/4};
				}
				
				\foreach \x/\y/\xx/\yy in {
					0.45 /0.53/0.475/1,
					0.575/0.69/0.6  /1,
					0.625/0.35/0.7  /1,
					0.7  /0.52/0.75 /1,
					0.75 /0.30/0.825/1,
					0.975/0.62/1    /1
				}
				{
					\drawVerticalItem{\x*\w}{\y*\hprime+\h/4}{\xx*\w}{\yy*\hprime+\h/4};
				}
				
				\foreach \x/\y/\xx/\yy in {
					0.0  /1/0.025/1.085,
					0.025/1/0.05 /1.065,
					0.075/1/0.125/1.06 ,
					0.075/1.06/0.125/1.10,
					0.225/1/0.25 /1.04 ,
					0.225/1.04/0.25 /1.07,
					0.25 /1/0.275/1.04 ,
					0.25 /1.04/0.275/1.07,
					0.275/1/0.3  /1.05 ,
					0.275/1.05/0.3  /1.08,
					0.3  /1/0.325/1.02 ,
					0.3  /1.02/0.325/1.07,
					0.35 /1/0.375/1.03 ,
					0.35 /1.03/0.375/1.10,
					0.375/1/0.4  /1.02 ,
					0.375/1.02/0.4  /1.09,
					0.475/1/0.5  /1.16 ,
					0.55 /1/0.575/1.05
				}
				{
					\drawVerticalItem{\x*\w}{\y*\hprime-\h/4}{\xx*\w}{\yy*\hprime-\h/4};
				}
				
				\foreach \x/\y/\xx/\yy in {
					0.125/0.59/0.15 /1,
					0.175/0.62/0.2  /1,
					0.4  /0.62/0.425/1,
					0.5  /0.60/0.55 /1,
					0.9  /0.63/0.925/1,
					0.95 /0.68/0.975/1
				}
				{
					\drawVerticalItem{\x*\w}{\y*\hprime-\h/4}{\xx*\w}{\yy*\hprime-\h/4};
				}
				
				\foreach \x/\y/\xx/\yy in {
					0.2  /1/0.225/1.06,
					0.425/1/0.45 /1.09,
					0.6  /1/0.625/1.05 ,
					0.825/1/0.85 /1.07,
					0.875/1/0.9  /1.08
				}
				{
					\drawVerticalItem{\x*\w}{\y*\hprime-\hprime/2}{\xx*\w}{\yy*\hprime-\hprime/2};
				}
				
				\foreach \x/\y/\xx/\yy in {
					0.025/0.0 /0.05 /0.35,
					0.05 /0.0 /0.075/0.32,
					0.05 /0.32/0.075/0.36,
					0.15 /0.0 /0.175/0.33,
					0.15 /0.33/0.175/0.39,
					0.325/0.0 /0.35 /0.36,
					0.325/0.36/0.35 /0.37,
					0.575/0.0 /0.6  /0.34,
					0.85 /0.0 /0.875/0.37,
					0.925/0.0 /0.95 /0.66
				}
				{
					\drawVerticalItem{\x*\w}{\y*\hprime}{\xx*\w}{\yy*\hprime};
				}
				
				\foreach \x/\y/\xx/\yy in {
					0.2  /0.225,
					0.425/0.45 ,
					0.6  /0.625,
					0.825/0.9
				}{
					\draw[pattern = north east lines,pattern color = gray] (\x*\w,0.5*\hprime) rectangle (\xx*\w,1*\hprime+\h/4);
				}
				
				\foreach \y/\z in {
					0.5*\hprime 		/ $\nicefrac{1}{2}\iH{B}$,
					0.25*\h     		/ $\nicefrac{1}{4}H$,
					\hprime - 0.25*\h	/ $\iH{B} - \nicefrac{1}{4}H$,
					\hprime				/ $\iH{B}$,
					\hprime + 0.25*\h	/ $\iH{B} + \nicefrac{1}{4}H$
				}{
					\draw[dotted] (-0.1*\w,\y)node[left]{\z} -- (1.1*\w,\y);
				}
				
				
				\end{tikzpicture}
				\caption{Fusing the pseudo items (light gray hatched area).}
				\label{fig:sub:FusingPseudoItems}
			\end{subfigure}
			\hfill
			\begin{subfigure}[t]{0.3\textwidth}
				\centering
				\begin{tikzpicture}
				
				\pgfmathsetmacro{\w}{2.7}
				\pgfmathsetmacro{\h}{4.8}
				\pgfmathsetmacro{\hprime}{4.5}
				\draw (0*\w,0) rectangle (1*\w,\hprime +\h/4);

				\foreach \x/\y/\xx/\yy in {
					-0.03 /0.02/0.025/0.33,
					-0.025/0.35/0.05/0.625,
					1.075 /0.05/0.95/0.34,
					1.1   /0.36/0.975/0.62
				}
				{
					\drawTallItem{\x*\w}{\y*\hprime}{\xx*\w}{\yy*\hprime};
				}
				
				\foreach \x/\y/\xx/\yy in {
					0    /0   /0.025/0.35,
					0.950/0   /1    /0.34,
					0.975/0.34/1    /0.62
				}{
					\draw[pattern = north west lines] (\x*\w,\y*\hprime) rectangle (\xx*\w,\yy*\hprime);
				}
				
				\foreach \x/\y/\xx/\yy in {
					0.175 /0.0 /0.325/0.31,
					0.35  /0.0 /0.45 /0.29,
					0.45  /0.0 /0.5  /0.53,
					0.5   /0.0 /0.575/0.29,
					0.6   /0.0 /0.7  /0.35,
					0.7   /0.0 /0.75 /0.52,
					0.75  /0.0 /0.85 /0.3,
					0.875 /0.0 /0.925/0.29,
					0.075 /0.0 /0.15 /0.28
				}
				{
					\drawTallItem{\x*\w}{\y*\hprime}{\xx*\w}{\yy*\hprime};
				}
				
				\foreach \x/\y/\xx/\yy in {
					0.05 /0.67/0.125/1,
					0.15 /0.70/0.175/1,
					0.225/0.72/0.3  /1,
					0.3  /0.71/0.4  /1,
					0.55 /0.65/0.6  /1
				}
				{
					\drawTallItem{\x*\w}{\y*\hprime-\h/4}{\xx*\w}{\yy*\hprime-\h/4};
				}
				
				\foreach \x/\y/\xx/\yy in {
					0.0   /0.71 /0.025/1,
					0.025 /0.69 /0.075/1,
					0.075 /0.71 /0.125/1,
					0.125 /0.69 /0.2  /1,
					0.2   /0.37 /0.225/1,
					0.225 /0.66 /0.25 /1,
					0.25  /0.66 /0.275/1,
					0.275 /0.67 /0.325/1,
					0.325 /0.66 /0.35 /1,
					0.35  /0.68 /0.375/1,
					0.375 /0.67 /0.425/1,
					0.475 /0.69 /0.575/1,
					0.425 /0.38 /0.45 /1,
					0.6   /0.40 /0.625/1,
					0.825 /0.37 /0.9  /1,
					0.9   /0.66 /0.975/1
				}
				{
					\drawTallItem{\x*\w}{\y*\hprime+\h/4}{\xx*\w}{\yy*\hprime+\h/4};
				}
				
				\foreach \x/\y/\xx/\yy in {
					0.45 /0.53/0.475/1,
					0.575/0.69/0.6  /1,
					0.625/0.35/0.7  /1,
					0.7  /0.52/0.75 /1,
					0.75 /0.30/0.825/1,
					0.975/0.62/1    /1
				}
				{
					\drawVerticalItem{\x*\w}{\y*\hprime+\h/4}{\xx*\w}{\yy*\hprime+\h/4};
				}
				
				\foreach \x/\y/\xx/\yy in {
					0.0  /1/0.025/1.085,
					0.025/1/0.05 /1.065,
					0.075/1/0.125/1.10,
					0.225/1/0.25 /1.07,
					0.25 /1/0.275/1.07,
					0.275/1/0.3  /1.08,
					0.3  /1/0.325/1.07,
					0.35 /1/0.375/1.10,
					0.375/1/0.4  /1.09,
					0.475/1/0.5  /1.16,
					0.55 /1/0.575/1.05
				}
				{
					\drawVerticalItem{\x*\w}{\y*\hprime-\h/4}{\xx*\w}{\yy*\hprime-\h/4};
				}
				
				\foreach \x/\y/\xx/\yy in {
					0.125/0.59/0.15 /1,
					0.175/0.62/0.2  /1,
					0.4  /0.62/0.425/1,
					0.5  /0.60/0.55 /1,
					0.9  /0.63/0.925/1,
					0.95 /0.68/0.975/1
				}
				{
					\drawVerticalItem{\x*\w}{\y*\hprime-\h/4}{\xx*\w}{\yy*\hprime-\h/4};
				}
				
				\foreach \x/\y/\xx/\yy in {
					0.2  /1/0.225/1.06,
					0.425/1/0.45 /1.09,
					0.6  /1/0.625/1.05 ,
					0.825/1/0.85 /1.07,
					0.875/1/0.9  /1.08
				}
				{
					\drawVerticalItem{\x*\w}{\y*\hprime-\hprime/2}{\xx*\w}{\yy*\hprime-\hprime/2};
				}
				
				\foreach \x/\y/\xx/\yy in {
					0.025/0.0 /0.05 /0.35,
					0.05 /0.0 /0.075/0.36,
					0.15 /0.0 /0.175/0.39,
					0.325/0.0 /0.35 /0.37,
					0.575/0.0 /0.6  /0.34,
					0.85 /0.0 /0.875/0.37,
					0.925/0.0 /0.95 /0.66
				}
				{
					\drawVerticalItem{\x*\w}{\y*\hprime}{\xx*\w}{\yy*\hprime};
				}
				
				\foreach \x/\y/\xx/\yy in {
					0.2  /0.225,
					0.425/0.45 ,
					0.6  /0.625,
					0.825/0.9
				}{
					\draw[pattern = north east lines, pattern color = gray] (\x*\w,0.5*\hprime) rectangle (\xx*\w,1*\hprime+\h/4);
				}
				
				\foreach \y/\z in {
					0.5*\hprime 		/ $\nicefrac{1}{2}\iH{B}$,
					0.25*\h     		/ $\nicefrac{1}{4}H$,
					\hprime - 0.25*\h	/ $\iH{B} - \nicefrac{1}{4}H$,
					\hprime				/ $\iH{B}$,
					\hprime + 0.25*\h	/ $\iH{B} + \nicefrac{1}{4}H$
				}{
					\draw[dotted] (-0.1*\w,\y) -- (1.1*\w,\y) ;
				}
				
				\draw[very thick, <-] (0.2075*\w,\hprime + \h/4) -- (0.2075*\w,\hprime + 1.1* \h/4) node[above]{$i_l$};
				\draw[very thick, <-] (0.865*\w,\hprime + \h/4) -- (0.865*\w,\hprime + 1.1* \h/4) node[above]{$i_r$};
				
				\draw[red] (0.2*\w,0) --(0.2 *\w,\hprime +\h/4);
				\draw[red] (0.9*\w,0) --(0.9 *\w,\hprime +\h/4);
				
				\draw[red, very thick, pattern = dots, pattern color = gray] (0*\w,\hprime-\h/4) rectangle (0.2 *\w,\hprime +\h/4);
				\draw[red] (0.00 *\w,\hprime) -- (-0.1 *\w,\hprime +\h/8) node[left,black]{$B_{l,1}$};
				
				\draw[red, very thick, pattern = dots, pattern color = gray] (0.9*\w,\hprime-\h/4) rectangle (1 *\w,\hprime +\h/4);
				\draw[red] (1.00 *\w,\hprime) --(1.15 *\w,\hprime +\h/8) node[above,black]{$B_{r,1}$};
				

				\end{tikzpicture}
				\caption{The area $B_{l,1}$ and $B_{r,1}$.}
				\label{fig:sub:AreaB1}
			\end{subfigure}
			\hfill
			\begin{subfigure}[t]{0.3\textwidth}
				\centering
				\begin{tikzpicture}
				\pgfmathsetmacro{\w}{2.7}
				\pgfmathsetmacro{\h}{4.8}
				\pgfmathsetmacro{\hprime}{4.5}
				\draw (0*\w,0) rectangle (1*\w,\hprime +\h/4);
				
				\draw (0*\w,0) rectangle (1*\w,\hprime +\h/4);

				\foreach \x/\y/\xx/\yy in {
					-0.03 /0.02/0.025/0.33,
					-0.025/0.35/0.05/0.625,
					1.075 /0.05/0.95/0.34,
					1.1   /0.36/0.975/0.62
				}
				{
					\drawTallItem{\x*\w}{\y*\hprime}{\xx*\w}{\yy*\hprime};
				}
				
				\foreach \x/\y/\xx/\yy in {
					0    /0   /0.025/0.35,
					0.950/0   /1    /0.34,
					0.975/0.34/1    /0.62
				}{
					\draw[pattern = north west lines] (\x*\w,\y*\hprime) rectangle (\xx*\w,\yy*\hprime);
				}
				
				\foreach \x/\y/\xx/\yy in {
					0.175 /0.0 /0.325/0.31,
					0.35  /0.0 /0.45 /0.29,
					0.45  /0.0 /0.5  /0.53,
					0.5   /0.0 /0.575/0.29,
					0.6   /0.0 /0.7  /0.35,
					0.7   /0.0 /0.75 /0.52,
					0.75  /0.0 /0.85 /0.3,
					0.875 /0.0 /0.925/0.29,
					0.075 /0.0 /0.15 /0.28
				}
				{
					\drawTallItem{\x*\w}{\y*\hprime}{\xx*\w}{\yy*\hprime};
				}
				
				\foreach \x/\y/\xx/\yy in {
					0.05 /0.67/0.125/1,
					0.15 /0.70/0.175/1,
					0.225/0.72/0.3  /1,
					0.3  /0.71/0.4  /1,
					0.55 /0.65/0.6  /1
				}
				{
					\drawTallItem{\x*\w}{\y*\hprime-\h/4}{\xx*\w}{\yy*\hprime-\h/4};
				}
				
				\foreach \x/\y/\xx/\yy in {
					0.0   /0.71 /0.025/1,
					0.025 /0.71 /0.075/1,
					0.075 /0.69 /0.125/1,
					0.125 /0.69 /0.2  /1,
					0.2   /0.37 /0.225/1,
					0.225 /0.66 /0.25 /1,
					0.25  /0.66 /0.275/1,
					0.275 /0.67 /0.325/1,
					0.325 /0.66 /0.35 /1,
					0.35  /0.68 /0.375/1,
					0.375 /0.67 /0.425/1,
					0.475 /0.69 /0.575/1,
					0.425 /0.38 /0.45 /1,
					0.6   /0.40 /0.625/1,
					0.825 /0.37 /0.9  /1,
					0.9   /0.66 /0.975/1
				}
				{
					\drawTallItem{\x*\w}{\y*\hprime+\h/4}{\xx*\w}{\yy*\hprime+\h/4};
				}
				
				\foreach \x/\y/\xx/\yy in {
					0.45 /0.53/0.475/1,
					0.575/0.69/0.6  /1,
					0.625/0.35/0.7  /1,
					0.7  /0.52/0.75 /1,
					0.75 /0.30/0.825/1,
					0.975/0.62/1    /1
				}
				{
					\drawVerticalItem{\x*\w}{\y*\hprime+\h/4}{\xx*\w}{\yy*\hprime+\h/4};
				}
				
				\foreach \x/\y/\xx/\yy in {
					0.000/1/0.05/1.10,
					0.05 /1/0.075/1.085,
					0.075/1/0.1 /1.065,
					0.225/1/0.25 /1.07,
					0.25 /1/0.275/1.07,
					0.275/1/0.3  /1.08,
					0.3  /1/0.325/1.07,
					0.35 /1/0.375/1.10,
					0.375/1/0.4  /1.09,
					0.475/1/0.5  /1.16,
					0.55 /1/0.575/1.05
				}
				{
					\drawVerticalItem{\x*\w}{\y*\hprime-\h/4}{\xx*\w}{\yy*\hprime-\h/4};
				}
				
				\foreach \x/\y/\xx/\yy in {
					0.125/0.59/0.15 /1,
					0.175/0.62/0.2  /1,
					0.4  /0.62/0.425/1,
					0.5  /0.60/0.55 /1,
					0.9  /0.63/0.925/1,
					0.95 /0.68/0.975/1
				}
				{
					\drawVerticalItem{\x*\w}{\y*\hprime-\h/4}{\xx*\w}{\yy*\hprime-\h/4};
				}
				
				\foreach \x/\y/\xx/\yy in {
					0.2  /1/0.225/1.06,
					0.425/1/0.45 /1.09,
					0.6  /1/0.625/1.05 ,
					0.825/1/0.85 /1.07,
					0.875/1/0.9  /1.08
				}
				{
					\drawVerticalItem{\x*\w}{\y*\hprime-\hprime/2}{\xx*\w}{\yy*\hprime-\hprime/2};
				}
				
				\foreach \x/\y/\xx/\yy in {
					0.025/0.0 /0.05 /0.35,
					0.05 /0.0 /0.075/0.36,
					0.15 /0.0 /0.175/0.39,
					0.325/0.0 /0.35 /0.37,
					0.575/0.0 /0.6  /0.34,
					0.85 /0.0 /0.875/0.37,
					0.925/0.0 /0.95 /0.66
				}
				{
					\drawVerticalItem{\x*\w}{\y*\hprime}{\xx*\w}{\yy*\hprime};
				}
				
				\foreach \x/\y/\xx/\yy in {
					0.2  /0.225,
					0.425/0.45 ,
					0.6  /0.625,
					0.825/0.9
				}{
					\draw[pattern = north east lines, pattern color = gray] (\x*\w,0.5*\hprime) rectangle (\xx*\w,1*\hprime+\h/4);
				}
				
				\foreach \y/\z in {
					0.5*\hprime 		/ $\nicefrac{1}{2}\iH{B}$,
					0.25*\h     		/ $\nicefrac{1}{4}H$,
					\hprime - 0.25*\h	/ $\iH{B} - \nicefrac{1}{4}H$,
					\hprime				/ $\iH{B}$,
					\hprime + 0.25*\h	/ $\iH{B} + \nicefrac{1}{4}H$
				}{
					\draw[dotted] (-0.1*\w,\y) -- (1.1*\w,\y) ;
				}
				
				\draw[very thick] (0*\w,\hprime-\h/4) rectangle (0.2 *\w,\hprime +\h/4);
				\draw[fill = white, opacity = 0.5] (0*\w,\hprime-\h/4) rectangle (0.2 *\w,\hprime +\h/4);
				
				\draw[very thick] (0.9*\w,\hprime-\h/4) rectangle (1 *\w,\hprime +\h/4);
				\draw[fill = white, opacity = 0.5] (0.9*\w,\hprime-\h/4) rectangle (1 *\w,\hprime +\h/4);
				
				\draw[red, very thick,pattern = dots,pattern color = gray] (0*\w,\hprime-\h/4) rectangle node[midway,opacity =1](A){} (0.2 *\w,0);
				\draw[thick] (A) --(-0.2*\w,3*\hprime/6) node[above]{$B_{l,2}$};
				\draw[red, very thick,pattern = dots,pattern color = gray] (0.9*\w,\hprime-\h/4) rectangle (1 *\w,0);
				\draw[thick] (1*\w,\h/4) --(1.2 *\w,3*\hprime/8) node[above]{$B_{r,2}$};
				
				\end{tikzpicture}
				\caption{The area $B_{l,2}$ and $B_{r,2}$.}
				\label{fig:sub:AreaB2}
			\end{subfigure}			
			\caption{We perform one more vertical shift, by fusing the pseudo items as seen in the previews section. Afterward, we introduce the first two rectangular subareas, where we can sort the items, such that they can be placed into few \subboxes.}
			\label{fig:GenralReordering1}
		\end{figure}
		
		\begin{description}
			\item[Area $B_{l,1}$:]
			Consider the leftmost (pseudo) item $i_l$ with height $\nicefrac{1}{2}\iH{B} + \nicefrac{1}{4} H$ touching $\iH{B}+\nicefrac{1}{4} H$ and let $i_r$ be the right most of these items. 
			Left of $i_l$ inside the box $B$, there is no item intersecting the horizontal line $\iH{B}-\nicefrac{1}{4} H$ since only (pseudo) items with height $\nicefrac{1}{2}\iH{B} + \nicefrac{1}{4} H$ touching $\iH{B}+\nicefrac{1}{4} H$ overlap this horizontal line, see Figure~\ref{fig:sub:AreaB1}. 
			Therefore, each item left of $i_l$ above $\iH{B}-\nicefrac{1}{4} H$ either touches $\iH{B}-\nicefrac{1}{4} H$ with its lower border or $\iH{B} + \nicefrac{1}{4} H$ with its upper border. 
			Since there is no item intersecting the left box border, we can sort the items left of $i_l$ touching $\iH{B} + \nicefrac{1}{4} H$ in descending order and the items touching $\iH{B} - \nicefrac{1}{4} H$ in ascending order of their heights, without constructing any overlap. 
			The same holds for the right side of $i_r$. 
			We call these areas $B_{l,1}$ and $B_{r,1}$. 
			
			\item[Area $B_{l,2}$:]
			We draw a vertical line at the left border of $i_l$ to the bottom of the box, see Figure~\ref{fig:sub:AreaB2}. 
			If this line cuts a tall item $l_b$ at the bottom, it defines a new unmovable item. 
			Let us consider the area between the line and the left box border below $\iH{B}-\nicefrac{1}{4} H$. 
			We call this area $B_{l,2}$. 
			In $B_{l,2}$ each item either touches the horizontal line at $0$ or $\iH{B}-\nicefrac{1}{4} H$ and on each side there are at most two tall unmovable items.
			We extend the unmovable item intersecting $\iH{B}/2$ on the top, such that it touches the horizontal line at $\iH{B}-\nicefrac{1}{4}H$ and reorder this box with the techniques from Lemma~\ref{lma:reorderingPreviusly}. 
			We do the the same on the right of $i_r$. 
			
			\item[Cases for $i_l$ and $i_r$:]
			If $i_l$ and $i_r$ do not exist there are no items overlapping the horizontal line $\iH{B}-H/4$ and we can partition the box in two areas $B_1$ and $B_2$.
			We reorder $B_1$ as described for $B_{l,1}$ and $B_2$ as described for $B_{l,2}$. 
			In the case that $l$ equals $r$ we introduce $B_{l,1}$, $B_{l,2}$, $B_{r,1}$ and $B_{r,2}$ as described, and order the tall items completely below $l$ such that items with the same height are positioned next to each other.
			
			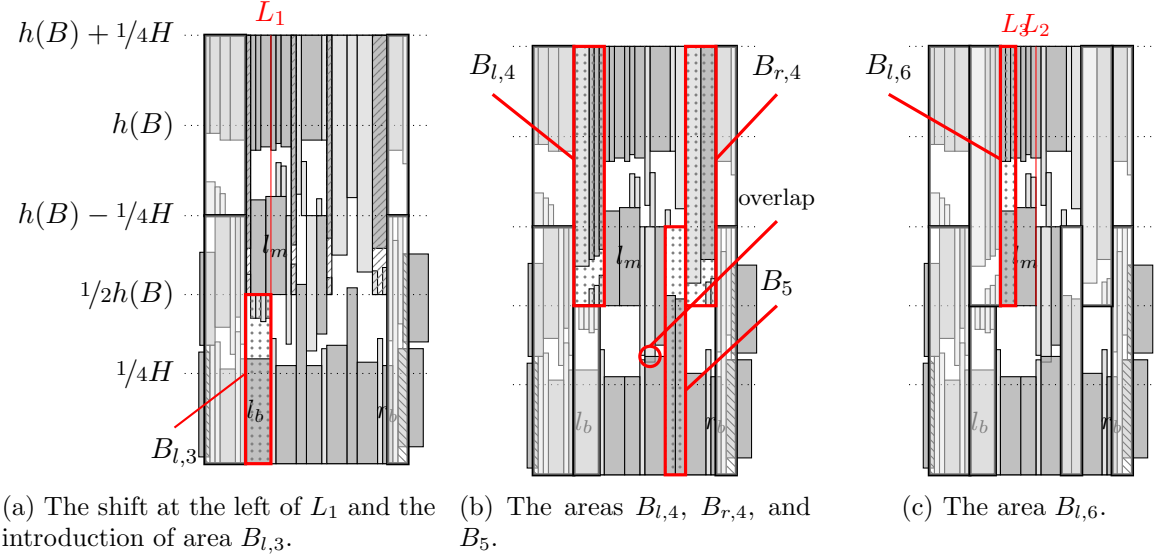
\begin{figure}
				\centering
				\begin{subfigure}[t]{0.37\textwidth}
					\centering
					\begin{tikzpicture}
					\pgfmathsetmacro{\w}{2.7}
					\pgfmathsetmacro{\h}{4.8}
					\pgfmathsetmacro{\hprime}{4.5}
					
					\draw (0*\w,0) rectangle (1*\w,\hprime +\h/4);
					
					\foreach \x/\y/\xx/\yy in {
						-0.03 /0.02/0.025/0.33,
						-0.025/0.35/0.05/0.625,
						1.075 /0.05/0.95/0.34,
						1.1   /0.36/0.975/0.62
					}
					{
						\drawTallItem{\x*\w}{\y*\hprime}{\xx*\w}{\yy*\hprime};
					}
					
					\foreach \x/\y/\xx/\yy in {
						0    /0   /0.025/0.35,
						0.950/0   /1    /0.34,
						0.975/0.34/1    /0.62
					}{
						\draw[pattern = north west lines] (\x*\w,\y*\hprime) rectangle (\xx*\w,\yy*\hprime);
					}
					
					\foreach \x/\y/\xx/\yy in {
						0.05 /0.67/0.125/1,
						0.15 /0.70/0.175/1,
						0.55 /0.65/0.6  /1
					}
					{
						\drawTallItem{\x*\w}{\y*\hprime-\h/4}{\xx*\w}{\yy*\hprime-\h/4};
					}
					
					\foreach \x/\y/\xx/\yy/\z in {
						0.225/1/0.3  /1.28/,
						0.3  /1/0.4  /1.29/$l_m$
					}
					{
						\drawTallItem[\small \z]{\x*\w}{\y*\hprime-0.5*\hprime}{\xx*\w}{\yy*\hprime- 0.5*\hprime};
					}
					\foreach \x/\y/\xx/\yy in {
						0.0   /0.71 /0.025/1,
						0.025 /0.71 /0.075/1,
						0.075 /0.69 /0.125/1,
						0.125 /0.69 /0.2  /1,
						0.2   /0.37 /0.225/1,
						0.225 /0.66 /0.25 /1,
						0.25  /0.66 /0.275/1,
						0.275 /0.67 /0.325/1,
						0.325 /0.66 /0.35 /1,
						0.35  /0.68 /0.375/1,
						0.375 /0.67 /0.425/1,
						0.475 /0.69 /0.575/1,
						0.425 /0.38 /0.45 /1,
						0.6   /0.40 /0.625/1,
						0.825 /0.37 /0.9  /1,
						0.9   /0.66 /0.975/1
					}
					{
						\drawTallItem{\x*\w}{\y*\hprime+\h/4}{\xx*\w}{\yy*\hprime+\h/4};
					}
					
					\foreach \x/\y/\xx/\yy in {
						0.45 /0.53/0.475/1,
						0.575/0.69/0.6  /1,
						0.625/0.35/0.7  /1,
						0.7  /0.52/0.75 /1,
						0.75 /0.30/0.825/1,
						0.975/0.62/1    /1
					}
					{
						\drawVerticalItem{\x*\w}{\y*\hprime+\h/4}{\xx*\w}{\yy*\hprime+\h/4};
					}
					
					\foreach \x/\y/\xx/\yy in {
						0.000/1/0.05/1.10,
						0.05 /1/0.075/1.085,
						0.075/1/0.1 /1.065,
						0.475/1/0.5  /1.16,
						0.55 /1/0.575/1.05
					}
					{
						\drawVerticalItem{\x*\w}{\y*\hprime-\h/4}{\xx*\w}{\yy*\hprime-\h/4};
					}
					
					\foreach \x/\y/\xx/\yy in {
						0.225/0.93/0.25 /1,
						0.25 /0.93/0.275/1,
						0.275/0.92/0.3  /1,
						0.3  /0.93/0.325/1,
						0.35 /1.29/0.375/1.39,
						0.375/1.29/0.4  /1.38
					}
					{
						\drawVerticalItem{\x*\w}{\y*\hprime-0.5*\hprime}{\xx*\w}{\yy*\hprime-0.5*\hprime};
					}

					\foreach \x/\y/\xx/\yy in {
						0.125/0.59/0.15 /1,
						0.175/0.62/0.2  /1,
						0.4  /0.62/0.425/1,
						0.5  /0.60/0.55 /1,
						0.9  /0.63/0.925/1,
						0.95 /0.68/0.975/1
					}
					{
						\drawVerticalItem{\x*\w}{\y*\hprime-\h/4}{\xx*\w}{\yy*\hprime-\h/4};
					}
					
					\foreach \x/\y/\xx/\yy in {
						0.2  /1/0.225/1.06,
						0.425/1/0.45 /1.09,
						0.6  /1/0.625/1.05,
						0.825/1/0.85 /1.07,
						0.875/1/0.9  /1.08
					}
					{
						\drawVerticalItem{\x*\w}{\y*\hprime-\hprime/2}{\xx*\w}{\yy*\hprime-\hprime/2};
					}
					
					\foreach \x/\y/\xx/\yy in {
						0.025/0.0 /0.05 /0.35,
						0.05 /0.0 /0.075/0.36,
						0.15 /0.0 /0.175/0.39,
						0.325/0.0 /0.35 /0.37,
						0.575/0.0 /0.6  /0.34,
						0.85 /0.0 /0.875/0.37,
						0.925/0.0 /0.95 /0.66
					}
					{
						\drawVerticalItem{\x*\w}{\y*\hprime}{\xx*\w}{\yy*\hprime};
					}
					
					\foreach \x/\y/\xx/\yy in {
						0.2  /0.225,
						0.425/0.45 ,
						0.6  /0.625,
						0.825/0.9
					}{
						\draw[pattern = north east lines, pattern color = gray] (\x*\w,0.5*\hprime) rectangle (\xx*\w,1*\hprime+\h/4);
					}
					
					\foreach \x/\y/\xx/\yy/\z in {
						0.075 /0.0 /0.15 /0.28/,
						0.175 /0.0 /0.325/0.31/$l_b$,
						0.35  /0.0 /0.45 /0.29/,
						0.45  /0.0 /0.5  /0.53/,
						0.5   /0.0 /0.575/0.29/,
						0.6   /0.0 /0.7  /0.35/,
						0.7   /0.0 /0.75 /0.52/,
						0.75  /0.0 /0.85 /0.30/,
						0.875 /0.0 /0.925/0.29/$r_b$
					}
					{
						\drawTallItem[\small \z]{\x*\w}{\y*\hprime}{\xx*\w}{\yy*\hprime};
					}

					\foreach \y/\z in {
						0.5*\hprime 		/ $\nicefrac{1}{2}\iH{B}$,
						0.25*\h     		/ $\nicefrac{1}{4}H$,
						\hprime - 0.25*\h	/ $\iH{B} - \nicefrac{1}{4}H$,
						\hprime				/ $\iH{B}$,
						\hprime + 0.25*\h	/ $\iH{B} + \nicefrac{1}{4}H$
					}{
						\draw[dotted] (-0.1*\w,\y) node[left]{\z}-- (1.1*\w,\y) ;
					}
					
					\foreach \x/\y/\xx/\yy in {
						0.0*\w/\hprime-0.25*\h/0.2*\w/\hprime + 0.25*\h,
						0.9*\w/\hprime-0.25*\h/1.0*\w/\hprime +0.25*\h,
						0.0*\w/\hprime-0.25*\h/0.2*\w/0,
						0.9*\w/\hprime-0.25*\h/1.0*\w/0
					}{
						\draw[very thick] (\x,\y) rectangle (\xx,\yy);
						\draw[fill = white, opacity = 0.5] (\x,\y) rectangle (\xx,\yy);
						
					}
					
					\draw[red] (0.325*\w,0.00*\hprime) rectangle (0.325*\w,\hprime +\h/4) node[above]{$L_1$};
					\draw[very thick,red,pattern = dots,pattern color = gray] (0.2*\w,\hprime/2) rectangle (0.325 *\w,0);
					\draw[thick,red] (0.2*\w,\h/4) --(-0.15 *\w,0.1*\h) node[below,black]{$B_{l,3}$};
					

					\end{tikzpicture}
					\caption{The shift at the left of $L_1$ and the introduction of area $B_{l,3}$.}
					\label{fig:sub:AreaB3}
				\end{subfigure}
				\hfill
				\begin{subfigure}[t]{0.3\textwidth}
					\centering
					\begin{tikzpicture}
					\pgfmathsetmacro{\w}{2.7}
					\pgfmathsetmacro{\h}{4.8}
					\pgfmathsetmacro{\hprime}{4.5}
					
					\draw (0*\w,0) rectangle (1*\w,\hprime +\h/4);
					
					\foreach \x/\y/\xx/\yy in {
						-0.03 /0.02/0.025/0.33,
						-0.025/0.35/0.05/0.625,
						1.075 /0.05/0.95/0.34,
						1.1   /0.36/0.975/0.62
					}
					{
						\drawTallItem{\x*\w}{\y*\hprime}{\xx*\w}{\yy*\hprime};
					}
					
					\foreach \x/\y/\xx/\yy in {
						0    /0   /0.025/0.35,
						0.950/0   /1    /0.34,
						0.975/0.34/1    /0.62
					}{
						\draw[pattern = north west lines] (\x*\w,\y*\hprime) rectangle (\xx*\w,\yy*\hprime);
					}
					
					\foreach \x/\y/\xx/\yy in {
						0.05 /0.67/0.125/1,
						0.15 /0.70/0.175/1,
						0.6 /0.65/0.65  /1
					}
					{
						\drawTallItem{\x*\w}{\y*\hprime-\h/4}{\xx*\w}{\yy*\hprime-\h/4};
					}
					
					\foreach \x/\y/\xx/\yy/\z in {
						0.35 /1/0.425  /1.28/,
						0.425/1/0.525  /1.29/$l_m$
					}
					{
						\drawTallItem[\small \z]{\x*\w}{\y*\hprime-0.5*\hprime}{\xx*\w}{\yy*\hprime-0.5* \hprime};
					}
					\foreach \x/\y/\xx/\yy in {
						0.0   /0.71 /0.025/1,
						0.025 /0.71 /0.075/1,
						0.075 /0.69 /0.125/1,
						0.125 /0.69 /0.2  /1,
						0.275 /0.37 /0.3  /1,
						0.3   /0.38 /0.325/1,
						0.325 /0.40 /0.35 /1,
						0.35  /0.66 /0.375/1,
						0.375 /0.66 /0.4  /1,
						0.4   /0.67 /0.45 /1,
						0.45  /0.66 /0.475/1,
						0.475 /0.68 /0.5  /1,
						0.5   /0.67 /0.55 /1,
						0.575 /0.69 /0.675/1,
						0.825 /0.37 /0.9  /1,
						0.9   /0.66 /0.975/1
					}
					{
						\drawTallItem{\x*\w}{\y*\hprime+\h/4}{\xx*\w}{\yy*\hprime+\h/4};
					}
					
					\foreach \x/\y/\xx/\yy in {
						0.2  /0.35/0.275/1,
						0.55 /0.53/0.575/1,
						0.675/0.69/0.7  /1,
						0.7  /0.52/0.75 /1,
						0.75 /0.30/0.825/1,
						0.975/0.62/1    /1
					}
					{
						\drawVerticalItem{\x*\w}{\y*\hprime+\h/4}{\xx*\w}{\yy*\hprime+\h/4};
					}
					
					\foreach \x/\y/\xx/\yy in {
						0.000/1/0.05/1.10,
						0.05 /1/0.075/1.085,
						0.075/1/0.1 /1.065,
						0.575/1/0.6  /1.16,
						0.65 /1/0.675/1.05
					}
					{
						\drawVerticalItem{\x*\w}{\y*\hprime-\h/4}{\xx*\w}{\yy*\hprime-\h/4};
					}
					
					\foreach \x/\y/\xx/\yy in {
						0.225/0.93/0.25 /1,
						0.25 /0.93/0.275/1,
						0.275/0.92/0.3  /1,
						0.3  /0.93/0.325/1,
						0.475 /1.29/0.5/1.39,
						0.5/1.29/0.525  /1.38
					}
					{
						\drawVerticalItem{\x*\w}{\y*\hprime-0.5*\hprime}{\xx*\w}{\yy*\hprime-0.5*\hprime};
					}

					\foreach \x/\y/\xx/\yy in {
						0.125/0.59/0.15 /1,
						0.175/0.62/0.2  /1,
						0.525/0.62/0.55 /1,
						0.55 /0.60/0.6 /1,
						0.9  /0.63/0.925/1,
						0.95 /0.68/0.975/1
					}
					{
						\drawVerticalItem{\x*\w}{\y*\hprime-\h/4}{\xx*\w}{\yy*\hprime-\h/4};
					}
					
					\foreach \x/\y/\xx/\yy in {
						0.275/1/0.300/1.05,
						0.300/1/0.325/1.06,
						0.325/1/0.350/1.09,
						0.85 /1/0.875/1.07,
						0.875/1/0.9  /1.08
					}
					{
						\drawVerticalItem{\x*\w}{\y*\hprime-\hprime/2}{\xx*\w}{\yy*\hprime-\hprime/2};
					}
					
					\foreach \x/\y/\xx/\yy in {
						0.025/0.0 /0.05 /0.35,
						0.05 /0.0 /0.075/0.36,
						0.15 /0.0 /0.175/0.39,
						0.325/0.0 /0.35 /0.37,
						0.525/0.0 /0.55 /0.34,
						0.85 /0.0 /0.875/0.37,
						0.925/0.0 /0.95 /0.66
					}
					{
						\drawVerticalItem{\x*\w}{\y*\hprime}{\xx*\w}{\yy*\hprime};
					}

					\foreach \x/\y/\xx/\yy/\z in {
						0.075 /0.0 /0.15 /0.28/,
						0.175 /0.0 /0.325/0.31/$l_b$,
						0.35  /0.0 /0.45 /0.29/,
						0.45  /0.0 /0.525/0.29/,
						0.55  /0.0 /0.65  /0.35/,
						0.65  /0.0 /0.7  /0.53/,
						0.7   /0.0 /0.75 /0.52/,
						0.75  /0.0 /0.85 /0.30/,
						0.875 /0.0 /0.925/0.29/$r_b$
					}
					{
						\drawTallItem[\small \z]{\x*\w}{\y*\hprime}{\xx*\w}{\yy*\hprime};
					}

					\foreach \y/\z in {
						0.5*\hprime 		/ $\nicefrac{1}{2}\iH{B}$,
						0.25*\h     		/ $\nicefrac{1}{4}H$,
						\hprime - 0.25*\h	/ $\iH{B} - \nicefrac{1}{4}H$,
						\hprime				/ $\iH{B}$,
						\hprime + 0.25*\h	/ $\iH{B} + \nicefrac{1}{4}H$
					}{
						\draw[dotted] (-0.1*\w,\y) -- (1.1*\w,\y) ;
					}
					
					\foreach \x/\y/\xx/\yy in {
						0.0*\w/\hprime-0.25*\h/0.2*\w/\hprime + 0.25*\h, 
						0.9*\w/\hprime-0.25*\h/1.0*\w/\hprime +0.25*\h, 
						0.0*\w/\hprime-0.25*\h/0.2*\w/0, 
						0.9*\w/\hprime-0.25*\h/1.0*\w/0, 
						0.2*\w/0/0.325*\w/0.5*\hprime 
					}{
						\draw[very thick] (\x,\y) rectangle (\xx,\yy);
						\draw[fill = white, opacity = 0.5] (\x,\y) rectangle (\xx,\yy);
						
					}

					\draw[very thick, red, pattern = dots, pattern color = gray] (0.2*\w,\hprime/2) rectangle (0.35 *\w,\hprime+\h/4);
					\draw[very thick,red] (0.2*\w,7*\h/8) -- (-0.2*\w,9*\hprime/8) node[above, black]{$B_{l,4}$} ;
					
					\draw[very thick, red, pattern = dots, pattern color = gray] (0.75*\w,\hprime/2) rectangle (0.9 *\w,\hprime+\h/4);
					\draw[very thick, red] (0.9*\w,7*\h/8) -- (1.2*\w,9*\hprime/8) node[above,black]{$B_{r,4}$} ;
					
					\draw[very thick, red, pattern = dots, pattern color = gray] (0.65*\w,0.0*\hprime) rectangle node[midway](C){} (0.75 *\w,\hprime -\h/4);
					\draw[very thick,red] (0.75*\w,0.25*\hprime) -- (1.2*\w, 0.5*\hprime) node[above,black]{$B_{5}$};
					
					\draw[very thick, red] (0.575*\w,0.35*\hprime) circle (0.05*\w);
					\draw[very thick,red] (0.575*\w,0.35*\hprime+0.05*\w) -- (1.2*\w, 0.75*\hprime) node[above,black]{\footnotesize overlap};
					
					\end{tikzpicture}
					\caption{The areas $B_{l,4}$, $B_{r,4}$, and $B_{5}$.}
					\label{fig:sub:AreaB4andB5}
				\end{subfigure}
				\hfill
				\begin{subfigure}[t]{0.29\textwidth}
					\centering
					\begin{tikzpicture}
					\pgfmathsetmacro{\w}{2.7}
					\pgfmathsetmacro{\h}{4.8}
					\pgfmathsetmacro{\hprime}{4.5}
					\draw (0*\w,0) rectangle (1*\w,\hprime +\h/4);
					
					\foreach \x/\y/\xx/\yy in {
						-0.03 /0.02/0.025/0.33,
						-0.025/0.35/0.05/0.625,
						1.075 /0.05/0.95/0.34,
						1.1   /0.36/0.975/0.62
					}
					{
						\drawTallItem{\x*\w}{\y*\hprime}{\xx*\w}{\yy*\hprime};
					}
					
					\foreach \x/\y/\xx/\yy in {
						0    /0   /0.025/0.35,
						0.950/0   /1    /0.34,
						0.975/0.34/1    /0.62
					}{
						\draw[pattern = north west lines] (\x*\w,\y*\hprime) rectangle (\xx*\w,\yy*\hprime);
					}
					
					\foreach \x/\y/\xx/\yy in {
						0.05 /0.67/0.125/1,
						0.15 /0.70/0.175/1,
						0.6 /0.65/0.65  /1
					}
					{
						\drawTallItem{\x*\w}{\y*\hprime-\h/4}{\xx*\w}{\yy*\hprime-\h/4};
					}
					
					\foreach \x/\y/\xx/\yy/\z in {
						0.35 /1/0.425  /1.28/,
						0.425/1/0.525  /1.29/$l_m$
					}
					{
						\drawTallItem[\small \z]{\x*\w}{\y*\hprime-0.5*\hprime}{\xx*\w}{\yy*\hprime-0.5* \hprime};
					}
					\foreach \x/\y/\xx/\yy in {
						0.0   /0.71 /0.025/1,
						0.025 /0.71 /0.075/1,
						0.075 /0.69 /0.125/1,
						0.125 /0.69 /0.2  /1,
						0.275 /0.37 /0.3  /1,
						0.3   /0.38 /0.325/1,
						0.325 /0.40 /0.35 /1,
						0.35  /0.66 /0.375/1,
						0.375 /0.66 /0.4  /1,
						0.4   /0.67 /0.45 /1,
						0.45  /0.66 /0.475/1,
						0.475 /0.68 /0.5  /1,
						0.5   /0.67 /0.55 /1,
						0.575 /0.69 /0.675/1,
						0.825 /0.37 /0.9  /1,
						0.9   /0.66 /0.975/1
					}
					{
						\drawTallItem{\x*\w}{\y*\hprime+\h/4}{\xx*\w}{\yy*\hprime+\h/4};
					}
					
					\foreach \x/\y/\xx/\yy in {
						0.2  /0.35/0.275/1,
						0.55 /0.53/0.575/1,
						0.675/0.69/0.7  /1,
						0.7  /0.52/0.75 /1,
						0.75 /0.30/0.825/1,
						0.975/0.62/1    /1
					}
					{
						\drawVerticalItem{\x*\w}{\y*\hprime+\h/4}{\xx*\w}{\yy*\hprime+\h/4};
					}
					
					\foreach \x/\y/\xx/\yy in {
						0.000/1/0.05/1.10,
						0.05 /1/0.075/1.085,
						0.075/1/0.1 /1.065,
						0.575/1/0.6  /1.16,
						0.65 /1/0.675/1.05
					}
					{
						\drawVerticalItem{\x*\w}{\y*\hprime-\h/4}{\xx*\w}{\yy*\hprime-\h/4};
					}
					
					\foreach \x/\y/\xx/\yy in {
						0.225/0.93/0.25 /1,
						0.25 /0.93/0.275/1,
						0.275/0.92/0.3  /1,
						0.3  /0.93/0.325/1,
						0.475 /1.29/0.5/1.39,
						0.5/1.29/0.525  /1.38
					}
					{
						\drawVerticalItem{\x*\w}{\y*\hprime-0.5*\hprime}{\xx*\w}{\yy*\hprime-0.5*\hprime};
					}

					\foreach \x/\y/\xx/\yy in {
						0.125/0.59/0.15 /1,
						0.175/0.62/0.2  /1,
						0.525/0.62/0.55 /1,
						0.55 /0.60/0.6 /1,
						0.9  /0.63/0.925/1,
						0.95 /0.68/0.975/1
					}
					{
						\drawVerticalItem{\x*\w}{\y*\hprime-\h/4}{\xx*\w}{\yy*\hprime-\h/4};
					}
					
					\foreach \x/\y/\xx/\yy in {
						0.275/1/0.300/1.05,
						0.300/1/0.325/1.06,
						0.325/1/0.350/1.09,
						0.85 /1/0.875/1.07,
						0.875/1/0.9  /1.08
					}
					{
						\drawVerticalItem{\x*\w}{\y*\hprime-\hprime/2}{\xx*\w}{\yy*\hprime-\hprime/2};
					}
					
					\foreach \x/\y/\xx/\yy in {
						0.025/0.0 /0.05 /0.35,
						0.05 /0.0 /0.075/0.36,
						0.15 /0.0 /0.175/0.39,
						0.325/0.0 /0.35 /0.37,
						0.525/0.0 /0.55  /0.34,
						0.85 /0.0 /0.875/0.37,
						0.925/0.0 /0.95 /0.66
					}
					{
						\drawVerticalItem{\x*\w}{\y*\hprime}{\xx*\w}{\yy*\hprime};
					}

					\foreach \x/\y/\xx/\yy/\z in {
						0.075 /0.0 /0.15 /0.28/,
						0.175 /0.0 /0.325/0.31/$l_b$,
						0.35  /0.0 /0.45 /0.29/,
						0.45  /0.0 /0.525/0.29/,
						0.55  /0.0 /0.65 /0.35/,
						0.65  /0.0 /0.7  /0.53/,
						0.7   /0.0 /0.75 /0.52/,
						0.75  /0.0 /0.85 /0.30/,
						0.875 /0.0 /0.925/0.29/$r_b$
					}
					{
						\drawTallItem[\small \z]{\x*\w}{\y*\hprime}{\xx*\w}{\yy*\hprime};
					}

					\foreach \y/\z in {
						0.5*\hprime 		/ $\nicefrac{1}{2}\iH{B}$,
						0.25*\h     		/ $\nicefrac{1}{4}H$,
						\hprime - 0.25*\h	/ $\iH{B} - \nicefrac{1}{4}H$,
						\hprime				/ $\iH{B}$,
						\hprime + 0.25*\h	/ $\iH{B} + \nicefrac{1}{4}H$
					}{
						\draw[dotted] (-0.1*\w,\y) -- (1.1*\w,\y) ;
					}
					
					\foreach \x/\y/\xx/\yy in {
						0.0*\w/\hprime-0.25*\h/0.2*\w/\hprime + 0.25*\h,  
						0.9*\w/\hprime-0.25*\h/1.0*\w/\hprime +0.25*\h,   
						0.0*\w/\hprime-0.25*\h/0.2*\w/0,                  
						0.9*\w/\hprime-0.25*\h/1.0*\w/0,                  
						0.2*\w/0/0.325*\w/0.5*\hprime,                    
						0.2*\w /0.5*\hprime/0.35*\w/\hprime + 0.25*\h,    
						0.75*\w/0.5*\hprime/0.9*\w/\hprime + 0.25*\h,     
						0.65*\w/0.0*\hprime/0.75 *\w/\hprime -0.25*\h      
					}{
						\draw[very thick] (\x,\y) rectangle (\xx,\yy);
						\draw[fill = white, opacity = 0.5] (\x,\y) rectangle (\xx,\yy);
						
					}

					\draw[red] (0.425*\w,\hprime/2) -- (0.425*\w,\hprime +\h/4) node[above]{\small $L_3$};
					\draw[red] (0.525*\w,\hprime/2) -- (0.525*\w,\hprime +\h/4) node[above]{\small $L_2$};
					
					\draw[very thick, red, pattern = dots,  pattern color = gray] (0.35*\w,\hprime/2) rectangle (0.425*\w,\hprime +\h/4);
					
					\draw[very thick,red] (0.35*\w,7*\h/8) -- (-0.2*\w,9*\hprime/8) node[above, black]{$B_{l,6}$} ;
					
					\end{tikzpicture}
					\caption{The area $B_{l,6}$.}
					\label{fig:sub:AreaB6}
				\end{subfigure}				
				\caption{When reordering the items to create areas $B_{l,4}$, $B_{r,4}$, and $B_{5}$, it can happen that some items overlap. Ths will be resolved when considering areas $B_{l,9}$ and $B_{r,9}$.}
				\label{fig:GenralReordering2}
			\end{figure}

			\item[Area $B_{l,3}$] 
			Now, we look at the area between the left border of $i_l$ and the right border of $i_r$. 
			We denote by $r(i)$ the right border of an item $i$.
			If $r(l_b)$ is to the right of $r(i_l)$, we draw a vertical line at  $r(l_b)$, called $L_1$. 
			If $L_1$ intersects a tall item with upper border at $\iH{B}-\nicefrac{1}{4} H$, we call this item $l_m$.
			Left of $L_1$ and right of $l$, we shift up each item touching $\iH{B}-\nicefrac{1}{4} H$ with its top (the item $l_m$ inclusively), such that its lower border touches $\nicefrac{1}{2}\iH{B}$, and shift down each pseudo item touching $\iH{B}-\nicefrac{1}{4} H$ with its lower border, such that it touches $\nicefrac{1}{2}\iH{B}$ with its upper border. All pseudo items right of $L_1$ above $l_m$ are shifted, such that they touch the top of $l_m$ with their bottom, see Figure~\ref{fig:sub:AreaB3}. 
			Note, that no pseudo item is intersected by the line $L_1$.
			
			\begin{claim}
				\label{clm:overlapingBl3}
				After this shift no item overlaps another.
			\end{claim}
			\begin{proofClaim}
				Consider an item $i$ that was shifted up such that it starts at $\nicefrac{1}{2}\iH{B}$. 
				Note that the distance between the upper border of $l_b$ and $\nicefrac{1}{2}\iH{B}$ is less than $\nicefrac{1}{4}H$ because the upper border of $l_b$ is above $\nicefrac{1}{4}H$. 
				Hence there has to be some free space left between the upper border of $s$ and the lower border of each item above since we added $\nicefrac{1}{4}H$ to the packing height.
				
				Now consider an item $i'$ that was down shifted such that it ends at $\nicefrac{1}{2}\iH{B}$. 
				Above this item there has to be a tall item $i''$ starting at $\nicefrac{1}{2}\iH{B}$, which has a height larger than $\nicefrac{1}{4}H$. 
				The item $i'''$ above $i''$, i.e., an item ending at $\iH{B}+\nicefrac{1}{4}H$, has a height larger than $\nicefrac{1}{4}H$ as well since all items ending at $\iH{B}+\nicefrac{1}{4}H$ have at least this height. 
				Therefore, the vertical distance between $i''$ and $i'''$ is smaller than $\nicefrac{1}{4}H$. 
				Since we have added $\nicefrac{1}{4}H$ to the packing height, the vertical distance between the bottom of $i'$ and the top of $l_b$ has to be larger than zero.
				This concludes the proof of this claim that no item overlaps another after the described shift. 
			\end{proofClaim}
			
			Let $I_{l,\nicefrac{1}{2}\iH{B}}$ be the set of shifted items now touching $\nicefrac{1}{2}\iH{B}$ with their bottom.
			All the items in $I_{l,\nicefrac{1}{2}\iH{B}}$ have a height of at most $\nicefrac{1}{2}\iH{B}$.
			The area left of $L_1$ and right of the left border of $l$ below $\nicefrac{1}{2}\iH{B}$ is called $B_{l,3}$. 
			This area contains pseudo items touching $\nicefrac{1}{2}\iH{B}$ and a part of $l_b$ at the bottom.
			We sort the pseudo items above $l_b$ touching $\nicefrac{1}{2}\iH{B}$ in descending order of their heights. 
			
			On the other hand, if $r(l_b)$ is left of $r(i_l)$, we introduce the line $L_1$, but do not shift any item. 
			On the right of $i_r$, we introduce the same line and area named $R_1$ and $B_{r,3}$ respectively. 
			
			\item[Simple cases.]
			It is possible that $l_b$ equals $r_b$, or one of the lines $L_1$ or $R_1$ intersects with $i_r$ or $i_l$ respectively, or that $L_1$ equals $R_1$. 
			In each of these cases, there is no item with height larger than $\iH{B}/2$ touching the bottom of the box between the lines $L_1$ and $R_1$.
			If there is no such item, we shift all the items touching $\iH{B}-\nicefrac{1}{4} H$ with their top between $L_1$ and $R_1$, such that they touch $\nicefrac{1}{2}\iH{B}$ with their bottoms and the pseudo items touching $\iH{B}-\nicefrac{1}{4} H$ with their bottom such that they touch $\nicefrac{1}{2}\iH{B}$ with their top (similar as we did with the items above $l_b$). 
			Now there is no item intersecting the horizontal line $\nicefrac{1}{2}\iH{B}$. 
			Hence, we can sort the items above $\nicefrac{1}{2}\iH{B}$ between $l$ and $r$ by their heights as well as the items below $\nicefrac{1}{2}\iH{B}$. 
			After this step, we do not need any further reordering.
			
			\item[Area $B_{l,4}$ and Area $B_5$:] 
			We now consider the case that there is an item with height taller than $\nicefrac{1}{2}\iH{B}$ at the bottom between $i_l$ and $i_r$ and, hence, we need further reordering.
			The objective is to reorder the items of height $\nicefrac{1}{2}\iH{B} +\nicefrac{1}{4} H$ touching $\iH{B}+\nicefrac{1}{4} H$, such that they build two blocks, one next to $i_l$ and one next to $i_r$.
			These blocks will be areas $B_{l,4}$ and $B_{r,4}$.
			To make this reordering possible, we have to define a border between $i_l$ and $i_r$ such that all these items left of this border are shifted to the item $i_l$ while all these items right of this border are shifted to the item $i_r$.
			Let $i$ be an item of height larger than $\nicefrac{1}{2}\iH{B}$ touching the bottom between $L_1$ and $R_1$.
			This item defines the border between $i_l$ and $i_r$.
			
			Consider items with height $\nicefrac{1}{2}\iH{B} +\nicefrac{1}{4} H$ touching $\iH{B}+\nicefrac{1}{4} H$, right of $i_l$ and left of $i$. 
			Note, that none of these items is positioned above $i$. 
			We shift those items left of $i$ to the left until they touch $i_l$ and those items right of $i$ to the right until they touch $i_r$. 
			All other items with parts above $\nicefrac{1}{2}\iH{B}$ are shifted to the right or left accordingly, see Figure~\ref{fig:sub:AreaB4andB5}.
			We sort the items with height $\nicefrac{1}{2}\iH{B} +\nicefrac{1}{4} H$ such that the pseudo items containing tall items with an equal height are positioned next to each other. 
			The area containing these items left of $i$, $i_l$ inclusively, is called $B_{l,4}$.
			
			While we shift the items with height $\nicefrac{1}{2}\iH{B} +\nicefrac{1}{4} H$ touching $\iH{B}+\nicefrac{1}{4} H$ such that they are close to $i_l$ and $i_r$, we shift all the items between $i_l$ and $i_r$ with height $\iH{B}-H/4$ touching the horizontal line at $0$, such that they are next to $i$ and shift the other items to the left or right accordingly. 
			These items form a new area around $i$ called $B_5$.
			
			In this step of creating the areas $B_{l,4}$, $B_{r,4}$, and $B_5$, it can happen that items touching $\iH{B}-\nicefrac{1}{4} H$ with their top are intersecting items touching $0$ with their bottom, see Figure~\ref{fig:sub:AreaB4andB5}. 
			We will fix this in a later step, when we consider area $B_{l,9}$ and $B_{r,9}$.

			\item[Area $B_{l,6}$:] 
			Note that the items in the set $I_{l,\nicefrac{1}{2}\iH{B}}$ are now placed next to each other (before it was possible that items with height $\nicefrac{1}{2}\iH{B} +\nicefrac{1}{4} H$ where positioned between them). 
			In addition, there is no item touching $\iH{B}+\nicefrac{1}{4} H$ above an item touching $\iH{B}-\nicefrac{1}{4} H$ with their bottom, which was not above this item before. 
			Furthermore, the total width of items with bottom border above $\nicefrac{1}{4} H$ and below $\nicefrac{1}{2}\iH{B}$ between $L_1$ and the right of $i$ has not changed.
			
			If $l_m$ exists, we draw a vertical line $L_2$ at the right of $l_m$ and a vertical line $L_3$ at the left of $l_m$, see Figure~\ref{fig:sub:AreaB6}. 
			Let $l_{t,r}$ and $l_{t,l}$ be the tall items touching $\iH{B}+\nicefrac{1}{4} H$ intersected by this line if there are any. 
			We look at the area left of $L_3$ and right of $B_{l,4}$, which is bounded at the top by $\iH{B}+\nicefrac{1}{4} H$ and at the bottom by $\nicefrac{1}{2}\iH{B}$. 
			We call this area $B_{l,6}$. 
			In this area, each item touches the bottom or the top, and there is at most one item $l_{t,l}$ intersecting the border, see Figure~\ref{fig:sub:AreaB6}. 
			We use the reordering in Lemma~\ref{lma:reorderingPreviusly} to reorder the items in $B_{l,6}$. 
			
			\begin{figure}[ht]
				\centering
				\begin{subfigure}[t]{0.37\textwidth}
					\centering
					\begin{tikzpicture}
					\pgfmathsetmacro{\w}{3}
					\pgfmathsetmacro{\h}{4.8}
					\pgfmathsetmacro{\hprime}{4.5}
					
					\draw (0*\w,0) rectangle (1*\w,\hprime +\h/4);
					
					\foreach \x/\y/\xx/\yy in {
						-0.03 /0.02/0.025/0.33,
						-0.025/0.35/0.05/0.625,
						1.075 /0.05/0.95/0.34,
						1.1   /0.36/0.975/0.62
					}
					{
						\drawTallItem{\x*\w}{\y*\hprime}{\xx*\w}{\yy*\hprime};
					}
					
					\foreach \x/\y/\xx/\yy in {
						0    /0   /0.025/0.35,
						0.950/0   /1    /0.34,
						0.975/0.34/1    /0.62
					}{
						\draw[pattern = north west lines] (\x*\w,\y*\hprime) rectangle (\xx*\w,\yy*\hprime);
					}
					
					\foreach \x/\y/\xx/\yy in {
						0.05 /0.67/0.125/1,
						0.15 /0.70/0.175/1,
						0.6 /0.65/0.65  /1
					}
					{
						\drawTallItem{\x*\w}{\y*\hprime-\h/4}{\xx*\w}{\yy*\hprime-\h/4};
					}
					
					\foreach \x/\y/\xx/\yy/\z in {
						0.35 /1/0.425  /1.28/,
						0.425/1/0.525  /1.29/$l_m$
					}
					{
						\drawTallItem[\small \z]{\x*\w}{\y*\hprime-0.5*\hprime}{\xx*\w}{\yy*\hprime-0.5* \hprime};
					}
					\foreach \x/\y/\xx/\yy in {
						0.0   /0.71 /0.025/1,
						0.025 /0.71 /0.075/1,
						0.075 /0.69 /0.125/1,
						0.125 /0.69 /0.2  /1,
						0.275 /0.37 /0.3  /1,
						0.3   /0.38 /0.325/1,
						0.325 /0.40 /0.35 /1,
						0.35  /0.66 /0.375/1,
						0.375 /0.66 /0.4  /1,
						0.4   /0.67 /0.45 /1,
						0.45  /0.66 /0.475/1,
						0.475 /0.68 /0.5  /1,
						0.5   /0.67 /0.55 /1,
						0.575 /0.69 /0.675/1,
						0.825 /0.37 /0.9  /1,
						0.9   /0.66 /0.975/1
					}
					{
						\drawTallItem{\x*\w}{\y*\hprime+\h/4}{\xx*\w}{\yy*\hprime+\h/4};
					}
					
					\foreach \x/\y/\xx/\yy in {
						0.2  /0.35/0.275/1,
						0.55 /0.53/0.575/1,
						0.675/0.69/0.7  /1,
						0.7  /0.52/0.75 /1,
						0.75 /0.30/0.825/1,
						0.975/0.62/1    /1
					}
					{
						\drawVerticalItem{\x*\w}{\y*\hprime+\h/4}{\xx*\w}{\yy*\hprime+\h/4};
					}
					
					\foreach \x/\y/\xx/\yy in {
						0.000/1/0.05/1.10,
						0.05 /1/0.075/1.085,
						0.075/1/0.1 /1.065,
						0.575/1/0.6  /1.16,
						0.65 /1/0.675/1.05
					}
					{
						\drawVerticalItem{\x*\w}{\y*\hprime-\h/4}{\xx*\w}{\yy*\hprime-\h/4};
					}
					
					\foreach \x/\y/\xx/\yy in {
						0.225/0.93/0.25 /1,
						0.25 /0.93/0.275/1,
						0.275/0.92/0.3  /1,
						0.3  /0.93/0.325/1,
						0.475 /1.29/0.5/1.39,
						0.5/1.29/0.525  /1.38
					}
					{
						\drawVerticalItem{\x*\w}{\y*\hprime-0.5*\hprime}{\xx*\w}{\yy*\hprime-0.5*\hprime};
					}

					\foreach \x/\y/\xx/\yy in {
						0.125/0.59/0.15 /1,
						0.175/0.62/0.2  /1,
						0.525/0.62/0.55 /1,
						0.55 /0.60/0.6 /1,
						0.9  /0.63/0.925/1,
						0.95 /0.68/0.975/1
					}
					{
						\drawVerticalItem{\x*\w}{\y*\hprime-\h/4}{\xx*\w}{\yy*\hprime-\h/4};
					}
					
					\foreach \x/\y/\xx/\yy in {
						0.275/1/0.300/1.05,
						0.300/1/0.325/1.06,
						0.325/1/0.350/1.09,
						0.85 /1/0.875/1.07,
						0.875/1/0.9  /1.08
					}
					{
						\drawVerticalItem{\x*\w}{\y*\hprime-\hprime/2}{\xx*\w}{\yy*\hprime-\hprime/2};
					}
					
					\foreach \x/\y/\xx/\yy in {
						0.025/0.0 /0.05 /0.35,
						0.05 /0.0 /0.075/0.36,
						0.15 /0.0 /0.175/0.39,
						0.325/0.0 /0.35 /0.37,
						0.525/0.0 /0.55  /0.34,
						0.85 /0.0 /0.875/0.37,
						0.925/0.0 /0.95 /0.66
					}
					{
						\drawVerticalItem{\x*\w}{\y*\hprime}{\xx*\w}{\yy*\hprime};
					}

					\foreach \x/\y/\xx/\yy/\z in {
						0.075 /0.0 /0.15 /0.28/,
						0.175 /0.0 /0.325/0.31/$l_b$,
						0.35  /0.0 /0.45 /0.29/,
						0.45  /0.0 /0.525/0.29/,
						0.55  /0.0 /0.65  /0.35/,
						0.65  /0.0 /0.7  /0.53/,
						0.7   /0.0 /0.75 /0.52/,
						0.75  /0.0 /0.85 /0.30/,
						0.875 /0.0 /0.925/0.29/$r_b$
					}
					{
						\drawTallItem[\small \z]{\x*\w}{\y*\hprime}{\xx*\w}{\yy*\hprime};
					}

					\foreach \y/\z in {
						0.5*\hprime 		/ $\nicefrac{1}{2}\iH{B}$,
						0.25*\h     		/ $\nicefrac{1}{4}H$,
						\hprime - 0.25*\h	/ $\iH{B} - \nicefrac{1}{4}H$,
						\hprime				/ $\iH{B}$,
						\hprime + 0.25*\h	/ $\iH{B} + \nicefrac{1}{4}H$
					}{
						\draw[dotted] (-0.1*\w,\y)node[left]{\z} -- (1.1*\w,\y) ;
					}

					\foreach \x/\y/\xx/\yy in {
						0.0*\w/\hprime-0.25*\h/0.2*\w/\hprime + 0.25*\h,  
						0.9*\w/\hprime-0.25*\h/1.0*\w/\hprime +0.25*\h,   
						0.0*\w/\hprime-0.25*\h/0.2*\w/0,                  
						0.9*\w/\hprime-0.25*\h/1.0*\w/0,                  
						0.2*\w/0/0.325*\w/0.5*\hprime,                    
						0.2*\w /0.5*\hprime/0.35*\w/\hprime + 0.25*\h,    
						0.75*\w/0.5*\hprime/0.9*\w/\hprime + 0.25*\h,     
						0.65*\w/0.0*\hprime/0.75 *\w/\hprime -0.25*\h,    
						0.35*\w/0.5*\hprime/0.425*\w/\hprime + 0.25*\h 
					}{
						\draw[very thick] (\x,\y) rectangle (\xx,\yy);
						\draw[fill = white, opacity = 0.5] (\x,\y) rectangle (\xx,\yy);
						
					}

					\draw[very thick,red] (0.425*\w,9*\h/8) -- (-0.2*\w,9*\hprime/8) node[left, black]{$B_{l,7}$};
					
					\draw[very thick, red, pattern = dots, pattern color = gray] (0.425*\w,1.29*\hprime-0.5*\hprime) rectangle (0.525*\w,\hprime +\h/4);
					
					\end{tikzpicture}
					\caption{The area $B_{l,7}$.}
					\label{fig:sub:AreaB7}
				\end{subfigure}
				\hfill
				\begin{subfigure}[t]{0.24\textwidth}
					\centering
					\begin{tikzpicture}
					\pgfmathsetmacro{\w}{3}
					\pgfmathsetmacro{\h}{4.8}
					\pgfmathsetmacro{\hprime}{4.5}
					\draw (0*\w,0) rectangle (1*\w,\hprime +\h/4);
					
					\foreach \x/\y/\xx/\yy in {
						-0.03 /0.02/0.025/0.33,
						-0.025/0.35/0.05/0.625,
						1.075 /0.05/0.95/0.34,
						1.1   /0.36/0.975/0.62
					}
					{
						\drawTallItem{\x*\w}{\y*\hprime}{\xx*\w}{\yy*\hprime};
					}
					
					\foreach \x/\y/\xx/\yy in {
						0    /0   /0.025/0.35,
						0.950/0   /1    /0.34,
						0.975/0.34/1    /0.62
					}{
						\draw[pattern = north west lines] (\x*\w,\y*\hprime) rectangle (\xx*\w,\yy*\hprime);
					}
					
					\foreach \x/\y/\xx/\yy in {
						0.05 /0.67/0.125/1,
						0.15 /0.70/0.175/1,
						0.6 /0.65/0.65  /1
					}
					{
						\drawTallItem{\x*\w}{\y*\hprime-\h/4}{\xx*\w}{\yy*\hprime-\h/4};
					}
					
					\foreach \x/\y/\xx/\yy/\z in {
						0.35 /1/0.425  /1.28/,
						0.425/1/0.525  /1.29/$l_m$
					}
					{
						\drawTallItem[\small \z]{\x*\w}{\y*\hprime-0.5*\hprime}{\xx*\w}{\yy*\hprime - 0.5* \hprime};
					}
					\foreach \x/\y/\xx/\yy in {
						0.0   /0.71 /0.025/1,
						0.025 /0.71 /0.075/1,
						0.075 /0.69 /0.125/1,
						0.125 /0.69 /0.2  /1,
						0.275 /0.37 /0.3  /1,
						0.3   /0.38 /0.325/1,
						0.325 /0.40 /0.35 /1,
						0.35  /0.66 /0.375/1,
						0.375 /0.66 /0.4  /1,
						0.4   /0.67 /0.45 /1,
						0.45  /0.66 /0.475/1,
						0.475 /0.68 /0.5  /1,
						0.5   /0.67 /0.55 /1,
						0.575 /0.69 /0.675/1,
						0.825 /0.37 /0.9  /1,
						0.9   /0.66 /0.975/1
					}
					{
						\drawTallItem{\x*\w}{\y*\hprime+\h/4}{\xx*\w}{\yy*\hprime+\h/4};
					}
					
					\foreach \x/\y/\xx/\yy in {
						0.2  /0.35/0.275/1,
						0.55 /0.53/0.575/1,
						0.675/0.69/0.7  /1,
						0.7  /0.52/0.75 /1,
						0.75 /0.30/0.825/1,
						0.975/0.62/1    /1
					}
					{
						\drawVerticalItem{\x*\w}{\y*\hprime+\h/4}{\xx*\w}{\yy*\hprime+\h/4};
					}
					
					\foreach \x/\y/\xx/\yy in {
						0.000/1/0.05/1.10,
						0.05 /1/0.075/1.085,
						0.075/1/0.1 /1.065,
						0.575/1/0.6  /1.16,
						0.65 /1/0.675/1.05
					}
					{
						\drawVerticalItem{\x*\w}{\y*\hprime-\h/4}{\xx*\w}{\yy*\hprime-\h/4};
					}
					
					\foreach \x/\y/\xx/\yy in {
						0.225/0.93/0.25 /1,
						0.25 /0.93/0.275/1,
						0.275/0.92/0.3  /1,
						0.3  /0.93/0.325/1,
						0.475 /1.29/0.5/1.39,
						0.5/1.29/0.525  /1.38
					}
					{
						\drawVerticalItem{\x*\w}{\y*\hprime-0.5*\hprime}{\xx*\w}{\yy*\hprime-0.5*\hprime};
					}

					\foreach \x/\y/\xx/\yy in {
						0.125/0.59/0.15 /1,
						0.175/0.62/0.2  /1,
						0.525/0.62/0.55 /1,
						0.55 /0.60/0.6 /1,
						0.9  /0.63/0.925/1,
						0.95 /0.68/0.975/1
					}
					{
						\drawVerticalItem{\x*\w}{\y*\hprime-\h/4}{\xx*\w}{\yy*\hprime-\h/4};
					}
					
					\foreach \x/\y/\xx/\yy in {
						0.275/1/0.300/1.05,
						0.300/1/0.325/1.06,
						0.325/1/0.350/1.09,
						0.85 /1/0.875/1.07,
						0.875/1/0.9  /1.08
					}
					{
						\drawVerticalItem{\x*\w}{\y*\hprime-\hprime/2}{\xx*\w}{\yy*\hprime-\hprime/2};
					}
					
					\foreach \x/\y/\xx/\yy in {
						0.025/0.0 /0.05 /0.35,
						0.05 /0.0 /0.075/0.36,
						0.15 /0.0 /0.175/0.39,
						0.325/0.0 /0.35 /0.37,
						0.525/0.0 /0.55  /0.34,
						0.85 /0.0 /0.875/0.37,
						0.925/0.0 /0.95 /0.66
					}
					{
						\drawVerticalItem{\x*\w}{\y*\hprime}{\xx*\w}{\yy*\hprime};
					}

					\foreach \x/\y/\xx/\yy/\z in {
						0.075 /0.0 /0.15 /0.28/,
						0.175 /0.0 /0.325/0.31/$l_b$,
						0.35  /0.0 /0.45 /0.29/,
						0.45  /0.0 /0.525/0.29/,
						0.55  /0.0 /0.65  /0.35/,
						0.65  /0.0 /0.7  /0.53/,
						0.7   /0.0 /0.75 /0.52/,
						0.75  /0.0 /0.85 /0.30/,
						0.875 /0.0 /0.925/0.29/$r_b$
					}
					{
						\drawTallItem[\small \z]{\x*\w}{\y*\hprime}{\xx*\w}{\yy*\hprime};
					}

					\foreach \y/\z in {
						0.5*\hprime 		/ $\nicefrac{1}{2}\iH{B}$,
						0.25*\h     		/ $\nicefrac{1}{4}H$,
						\hprime - 0.25*\h	/ $\iH{B} - \nicefrac{1}{4}H$,
						\hprime				/ $\iH{B}$,
						\hprime + 0.25*\h	/ $\iH{B} + \nicefrac{1}{4}H$
					}{
						\draw[dotted] (-0.1*\w,\y) -- (1.1*\w,\y) ;
					}

					\foreach \x/\y/\xx/\yy in {
						0.0*\w/\hprime-0.25*\h/0.2*\w/\hprime + 0.25*\h,  
						0.9*\w/\hprime-0.25*\h/1.0*\w/\hprime +0.25*\h,   
						0.0*\w/\hprime-0.25*\h/0.2*\w/0,                  
						0.9*\w/\hprime-0.25*\h/1.0*\w/0,                  
						0.2*\w/0/0.325*\w/0.5*\hprime,                    
						0.2*\w /0.5*\hprime/0.35*\w/\hprime + 0.25*\h,    
						0.75*\w/0.5*\hprime/0.9*\w/\hprime + 0.25*\h,     
						0.65*\w/0.0*\hprime/0.75 *\w/\hprime -0.25*\h,    
						0.35*\w/0.5*\hprime/0.425*\w/\hprime + 0.25*\h,    
						0.425*\w/0.5*\hprime/0.525*\w/\hprime +0.25*\h
					}{
						\draw[very thick] (\x,\y) rectangle (\xx,\yy);
						\draw[fill = white, opacity = 0.5] (\x,\y) rectangle (\xx,\yy);
						
					}

					
					\draw[very thick, red, pattern = dots] (0.525*\w,\hprime -\h/4) rectangle (0.75*\w,\hprime +\h/4);
					
					\draw[very thick,red] (0.75*\w,9*\h/8) -- (1.2*\w,9*\hprime/8) node[above, black]{$B_{l,8}$};
					\end{tikzpicture}
					\caption{The area $B_{8}$.}
					\label{fig:sub:AreaB8}
				\end{subfigure}
				\hfill
				\begin{subfigure}[t]{0.32\textwidth}
					\centering
					\begin{tikzpicture}
					\pgfmathsetmacro{\w}{3}
					\pgfmathsetmacro{\h}{4.8}
					\pgfmathsetmacro{\hprime}{4.5}
					
					\draw (0*\w,0) rectangle (1*\w,\hprime +\h/4);
					
					\draw (0*\w,0) rectangle (1*\w,\hprime +\h/4);
					
					\foreach \x/\y/\xx/\yy in {
						-0.03 /0.02/0.025/0.33,
						-0.025/0.35/0.05/0.625,
						1.075 /0.05/0.95/0.34,
						1.1   /0.36/0.975/0.62
					}
					{
						\drawTallItem{\x*\w}{\y*\hprime}{\xx*\w}{\yy*\hprime};
					}
					
					\foreach \x/\y/\xx/\yy in {
						0    /0   /0.025/0.35,
						0.950/0   /1    /0.34,
						0.975/0.34/1    /0.62
					}{
						\draw[pattern = north west lines] (\x*\w,\y*\hprime) rectangle (\xx*\w,\yy*\hprime);
					}
					
					\foreach \x/\y/\xx/\yy in {
						0.05 /0.67/0.125/1,
						0.15 /0.70/0.175/1,
						0.6 /0.65/0.65  /1
					}
					{
						\drawTallItem{\x*\w}{\y*\hprime-\h/4}{\xx*\w}{\yy*\hprime-\h/4};
					}
					
					\foreach \x/\y/\xx/\yy/\z in {
						0.35 /1/0.425  /1.28/,
						0.425/1/0.525  /1.29/$l_m$
					}
					{
						\drawTallItem[\small \z]{\x * \w}{\y * \hprime - 0.5 * \hprime}{\xx * \w}{\yy * \hprime - 0.5 * \hprime};
					}
					\foreach \x/\y/\xx/\yy in {
						0.0   /0.71 /0.025/1,
						0.025 /0.71 /0.075/1,
						0.075 /0.69 /0.125/1,
						0.125 /0.69 /0.2  /1,
						0.275 /0.37 /0.3  /1,
						0.3   /0.38 /0.325/1,
						0.325 /0.40 /0.35 /1,
						0.35  /0.66 /0.375/1,
						0.375 /0.66 /0.4  /1,
						0.4   /0.67 /0.45 /1,
						0.45  /0.66 /0.475/1,
						0.475 /0.68 /0.5  /1,
						0.5   /0.67 /0.55 /1,
						0.55  /0.69 /0.65 /1,
						0.825 /0.37 /0.9  /1,
						0.9   /0.66 /0.975/1
					}
					{
						\drawTallItem{\x*\w}{\y*\hprime+\h/4}{\xx*\w}{\yy*\hprime+\h/4};
					}
					
					\foreach \x/\y/\xx/\yy in {
						0.2  /0.35/0.275/1,
						0.65 /0.69/0.675/1,
						0.675/0.53/0.7/1,
						0.7  /0.52/0.75 /1,
						0.75 /0.30/0.825/1,
						0.975/0.62/1    /1
					}
					{
						\drawVerticalItem{\x*\w}{\y*\hprime+\h/4}{\xx*\w}{\yy*\hprime+\h/4};
					}
					
					\foreach \x/\y/\xx/\yy in {
						0.000/1/0.05/1.10,
						0.05 /1/0.075/1.085,
						0.075/1/0.1 /1.065,
						0.55 /1/0.575  /1.16,
						0.575/1/0.6    /1.05
					}
					{
						\drawVerticalItem{\x*\w}{\y*\hprime-\h/4}{\xx*\w}{\yy*\hprime-\h/4};
					}
					
					\foreach \x/\y/\xx/\yy in {
						0.225/0.93/0.25 /1,
						0.25 /0.93/0.275/1,
						0.275/0.92/0.3  /1,
						0.3  /0.93/0.325/1,
						0.475 /1.29/0.5/1.39,
						0.5/1.29/0.525  /1.38
					}
					{
						\drawVerticalItem{\x*\w}{\y*\hprime-0.5*\hprime}{\xx*\w}{\yy*\hprime-0.5*\hprime};
					}

					\foreach \x/\y/\xx/\yy in {
						0.125/0.59/0.15 /1,
						0.175/0.62/0.2  /1,
						0.525/0.62/0.55 /1,
						0.55 /0.60/0.6 /1,
						0.9  /0.63/0.925/1,
						0.95 /0.68/0.975/1
					}
					{
						\drawVerticalItem{\x*\w}{\y*\hprime-\h/4}{\xx*\w}{\yy*\hprime-\h/4};
					}
					
					\foreach \x/\y/\xx/\yy in {
						0.275/1/0.300/1.05,
						0.300/1/0.325/1.06,
						0.325/1/0.350/1.09,
						0.85 /1/0.875/1.07,
						0.875/1/0.9  /1.08
					}
					{
						\drawVerticalItem{\x*\w}{\y*\hprime-\hprime/2}{\xx*\w}{\yy*\hprime-\hprime/2};
					}
					
					\foreach \x/\y/\xx/\yy in {
						0.025/0.0 /0.05 /0.35,
						0.05 /0.0 /0.075/0.36,
						0.15 /0.0 /0.175/0.39,
						0.325/0.0 /0.35 /0.37,
						0.525/0.0 /0.55  /0.34,
						0.85 /0.0 /0.875/0.37,
						0.925/0.0 /0.95 /0.66
					}
					{
						\drawVerticalItem{\x*\w}{\y*\hprime}{\xx*\w}{\yy*\hprime};
					}

					\foreach \x/\y/\xx/\yy/\z in {
						0.075 /0.0 /0.15 /0.28/,
						0.175 /0.0 /0.325/0.31/$l_b$,
						0.35  /0.0 /0.45 /0.29/,
						0.45  /0.0 /0.525/0.29/,
						0.55  /0.0 /0.65  /0.35/,
						0.65  /0.0 /0.7  /0.53/,
						0.7   /0.0 /0.75 /0.52/,
						0.75  /0.0 /0.85 /0.30/,
						0.875 /0.0 /0.925/0.29/$r_b$
					}
					{
						\drawTallItem[\small \z]{\x*\w}{\y*\hprime}{\xx*\w}{\yy*\hprime};
					}

					\foreach \y/\z in {
						0.5*\hprime 		/ $\nicefrac{1}{2}\iH{B}$,
						0.25*\h     		/ $\nicefrac{1}{4}H$,
						\hprime - 0.25*\h	/ $\iH{B} - \nicefrac{1}{4}H$,
						\hprime				/ $\iH{B}$,
						\hprime + 0.25*\h	/ $\iH{B} + \nicefrac{1}{4}H$
					}{
						\draw[dotted] (-0.1*\w,\y) -- (1.1*\w,\y) ;
					}

					\foreach \x/\y/\xx/\yy in {
						0.0*\w/\hprime-0.25*\h/0.2*\w/\hprime + 0.25*\h,  
						0.9*\w/\hprime-0.25*\h/1.0*\w/\hprime +0.25*\h,   
						0.0*\w/\hprime-0.25*\h/0.2*\w/0,                  
						0.9*\w/\hprime-0.25*\h/1.0*\w/0,                  
						0.2*\w/0/0.325*\w/0.5*\hprime,                    
						0.2*\w /0.5*\hprime/0.35*\w/\hprime + 0.25*\h,    
						0.75*\w/0.5*\hprime/0.9*\w/\hprime + 0.25*\h,     
						0.65*\w/0.0*\hprime/0.75 *\w/\hprime -0.25*\h,    
						0.35*\w/0.5*\hprime/0.425*\w/\hprime + 0.25*\h,   
						0.425*\w/0.5*\hprime/0.525*\w/\hprime +0.25*\h,   
						0.525*\w/\hprime -0.25*\h/0.75*\w/\hprime +0.25*\h     
					}{
						\draw[very thick] (\x,\y) rectangle (\xx,\yy);
						\draw[fill = white, opacity = 0.5] (\x,\y) rectangle (\xx,\yy);
						
					}
					
					\draw[very thick,red,pattern=dots] (0.325*\w,0*\hprime) -- (0.325*\w,\hprime/2) -- (0.525*\w,\hprime/2) -- (0.525*\w,\hprime -\h/4) -- (0.65*\w,\hprime -\h/4) -- (0.65*\w,0*\hprime) -- (0.325*\w,0*\hprime);
					
					\draw[very thick,red,pattern=dots] (0.75*\w,0*\hprime) -- (0.75*\w,0.58*\hprime -\h/4) -- (0.75*\w,0.58*\hprime -\h/4)--(0.75*\w,\hprime/2) -- (0.875*\w,\hprime /2) -- (0.875*\w,0) --(0.7*\w,0);
					
					\draw[very thick,red] (0.325*\w,\h/4) -- (-0.15*\w,3*\hprime/8) node[above, black]{$B_{l,9}$};
					\draw[very thick,red] (0.875*\w,\h/4) -- (1.2*\w,3*\hprime/8) node[above, black]{$B_{r,9}$};
					
					\end{tikzpicture}
					\caption{The areas $B_{l,9}$ and $B_{r,9}$.}
					\label{fig:sub:AreaB9}
				\end{subfigure}
				\caption{Reordering the items}
				\label{fig:GenralReordering4}
			\end{figure}
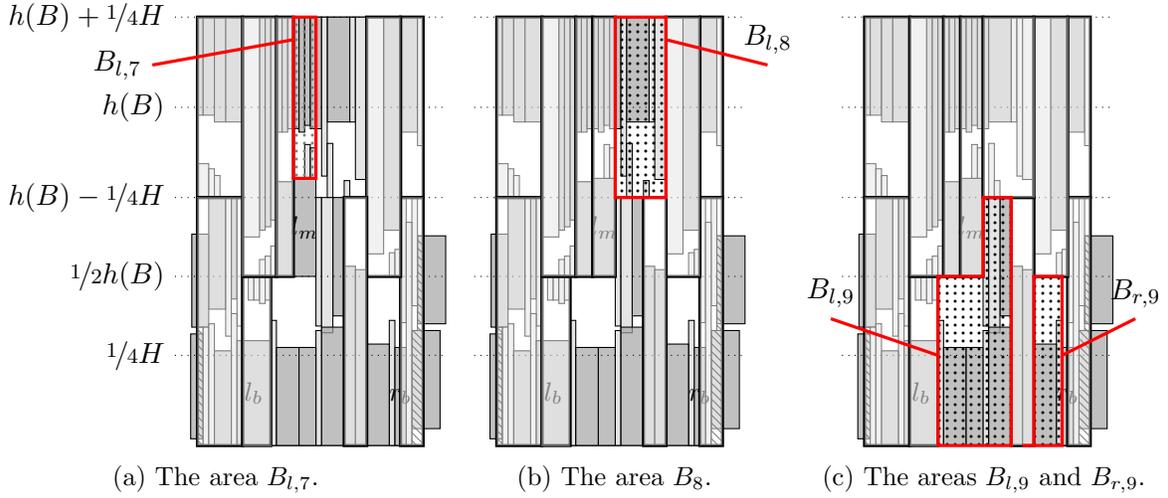
			
			\item[Area $B_{l,7}$:] 
			The area above $l_m$ is called $B_{l,7}$, see Figure~\ref{fig:sub:AreaB7}. 
			In this area, all items are touching $\iH{B}+\nicefrac{1}{4} H$ or the top of $l_m$. 
			All the items touching $l_m$ with their bottom (and not $\iH{B}+\nicefrac{1}{4} H$ with their top) are pseudo items. 
			We order the items touching $\iH{B}+\nicefrac{1}{4} H$ in ascending order of their heights and move the pseudo items below with them. 
			Now, we look at the overlapping items $l_{t,r}$ and $l_{t,l}$. 
			We move items with the height $h(l_{t,r})$ and $h(l_{t,l})$ next to these overlapping items. 
			This generates three areas for pseudo items. 
			The first is positioned below the first overlapping item together with the items with the same height, the second below the other overlapping item together with the items with the same height, and the last between these areas. 
			In each of these areas, we sort the pseudo items in descending order of their height.
			
			The areas $B_{l,6}$ and $B_{l,7}$ exist only if $l_m$ exists. If $l_m$ does not exist, we introduce the vertical line $L_2$ at the left border of the area $B_{l,4}$. We introduce $R_2$ and the areas $B_{r,6}$ and $B_{r,7}$ analogously on the left of $i_r$.

			\item[Area $B_8$:]
			Look at the area above $\iH{B}-\nicefrac{1}{4} H$ right of $L_2$ and left of $R_2$, see Figure~\ref{fig:sub:AreaB8}. 
			We call this area $B_8$. 
			There are at most two unmovable items overlapping this area. 
			One item $l_{t,r}$ on the left touching $\iH{B}+\nicefrac{1}{4} H$ and one item $r_{t,l}$ on the right touching $\iH{B}+\nicefrac{1}{4} H$.
			Since $B_8$ does not contain any item of height $\nicefrac{1}{2}\iH{B} +\nicefrac{1}{4} H$ or items from the sets $I_{l,\nicefrac{1}{2}\iH{B}}$ or $I_{r,\nicefrac{1}{2}\iH{B}}$, each item touches either the top or the bottom of this area. 
			Furthermore, all items touching the bottom are pseudo items. 
			Therefore, we can sort the items in this area as they are sorted in area $B_{l,7}$.
			
			\item[Area $B_{l,9}$:] 
			Last, we have to look at the items on the bottom between $L_1$ and $R_1$ as well as at the items touching $\iH{B}-\nicefrac{1}{4} H$ with their top between $L_2$ and $R_2$, see Figure~\ref{fig:sub:AreaB9}. 
			We consider the items touching the bottom between $L_1$ and the left border of $B_5$ and the items touching $\iH{B}-\nicefrac{1}{4} H$ with their top between $L_2$ and the left border of $B_5$. 
			The area containing these items is called $B_{l,9}$. 
			In $B_{l,9}$, we sort all items touching $\iH{B}-\nicefrac{1}{4} H$ in ascending order of their heights and the items at the bottom in descending order of their heights, such that the tallest on the bottom touches $l_b$ and the smallest touches the area $B_5$. 
			We do the same but mirrored on the right side of $i$ in the area $B_{r,9}$.

			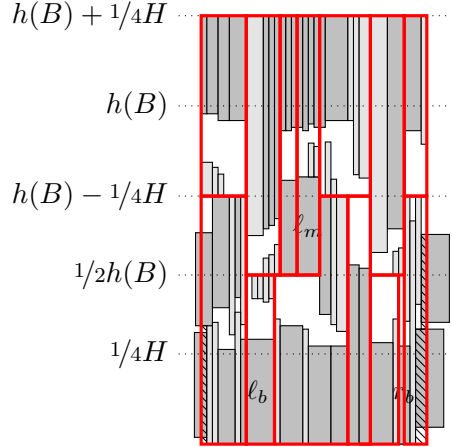
\begin{figure}[ht]
				\centering
				\begin{tikzpicture}
				\pgfmathsetmacro{\w}{3}
				\pgfmathsetmacro{\h}{4.8}
				\pgfmathsetmacro{\hprime}{4.5}
				
				\draw (0*\w,0) rectangle (1*\w,\hprime +\h/4);
				
				\foreach \x/\y/\xx/\yy in {
					-0.03 /0.02/0.025/0.33,
					-0.025/0.35/0.05/0.625,
					1.075 /0.05/0.95/0.34,
					1.1   /0.36/0.975/0.62
				}
				{
					\drawTallItem{\x*\w}{\y*\hprime}{\xx*\w}{\yy*\hprime};
				}
				
				\foreach \x/\y/\xx/\yy in {
					0    /0   /0.025/0.35,
					0.950/0   /1    /0.34,
					0.975/0.34/1    /0.62
				}{
					\draw[pattern = north west lines] (\x*\w,\y*\hprime) rectangle (\xx*\w,\yy*\hprime);
				}
				
				\foreach \x/\y/\xx/\yy in {
					0.05 /0.67/0.125/1,
					0.15 /0.70/0.175/1,
					0.525/0.65/0.575/1
				}
				{
					\drawTallItem{\x*\w}{\y*\hprime-\h/4}{\xx*\w}{\yy*\hprime-\h/4};
				}
				
				\foreach \x/\y/\xx/\yy/\z in {
					0.35 /1/0.425  /1.28/,
					0.425/1/0.525  /1.29/$\ell_m$
				}
				{
					\drawTallItem[\small \z]{\x*\w}{\y*\hprime-0.5*\hprime} {\xx*\w}{\yy*\hprime-0.5*\hprime};
				}
				\foreach \x/\y/\xx/\yy in {
					0.0   /0.71 /0.025/1,
					0.025 /0.71 /0.075/1,
					0.075 /0.69 /0.125/1,
					0.125 /0.69 /0.2  /1,
					0.275 /0.37 /0.3  /1,
					0.3   /0.38 /0.325/1,
					0.325 /0.40 /0.35 /1,
					0.35  /0.66 /0.375/1,
					0.375 /0.66 /0.4  /1,
					0.4   /0.67 /0.45 /1,
					0.45  /0.66 /0.475/1,
					0.475 /0.68 /0.5  /1,
					0.5   /0.67 /0.55 /1,
					0.55  /0.69 /0.65 /1,
					0.825 /0.37 /0.9  /1,
					0.9   /0.66 /0.975/1
				}
				{
					\drawTallItem{\x*\w}{\y*\hprime+\h/4}{\xx*\w}{\yy*\hprime+\h/4};
				}
				
				\foreach \x/\y/\xx/\yy in {
					0.2  /0.35/0.275/1,
					0.65 /0.69/0.675/1,
					0.675/0.53/0.7/1,
					0.7  /0.52/0.75 /1,
					0.75 /0.30/0.825/1,
					0.975/0.62/1    /1
				}
				{
					\drawVerticalItem{\x*\w}{\y*\hprime+\h/4}{\xx*\w}{\yy*\hprime+\h/4};
				}
				
				\foreach \x/\y/\xx/\yy in {
					0.000/1/0.05/1.10,
					0.05 /1/0.075/1.085,
					0.075/1/0.1 /1.065,
					0.55 /1/0.575  /1.16,
					0.575/1/0.6    /1.05
				}
				{
					\drawVerticalItem{\x*\w}{\y*\hprime-\h/4}{\xx*\w}{\yy*\hprime-\h/4};
				}
				
				\foreach \x/\y/\xx/\yy in {
					0.225/0.93/0.25 /1,
					0.25 /0.93/0.275/1,
					0.275/0.92/0.3  /1,
					0.3  /0.93/0.325/1,
					0.475 /1.29/0.5/1.39,
					0.5/1.29/0.525  /1.38
				}
				{
					\drawVerticalItem{\x*\w}{\y*\hprime-0.5*\hprime}{\xx*\w}{\yy*\hprime-0.5*\hprime};
				}

				\foreach \x/\y/\xx/\yy in {
					0.125/0.59/0.15 /1,
					0.175/0.62/0.2  /1,
					0.575/0.62/0.6  /1,
					0.6  /0.60/0.65 /1,
					0.9  /0.63/0.925/1,
					0.95 /0.68/0.975/1
				}
				{
					\drawVerticalItem{\x*\w}{\y*\hprime-\h/4}{\xx*\w}{\yy*\hprime-\h/4};
				}
				
				\foreach \x/\y/\xx/\yy in {
					0.275/1/0.300/1.05,
					0.300/1/0.325/1.06,
					0.325/1/0.350/1.09,
					0.85 /1/0.875/1.07,
					0.875/1/0.9  /1.08
				}
				{
					\drawVerticalItem{\x*\w}{\y*\hprime-\hprime/2}{\xx*\w}{\yy*\hprime-\hprime/2};
				}
				
				\foreach \x/\y/\xx/\yy in {
					0.025/0.0 /0.05 /0.35,
					0.05 /0.0 /0.075/0.36,
					0.15 /0.0 /0.175/0.39,
					0.325/0.0 /0.35 /0.37,
					0.45 /0.0 /0.475/0.34,
					0.85 /0.0 /0.875/0.37,
					0.925/0.0 /0.95 /0.66
				}
				{
					\drawVerticalItem{\x*\w}{\y*\hprime}{\xx*\w}{\yy*\hprime};
				}

				\foreach \x/\y/\xx/\yy/\z in {
					0.075 /0.0 /0.15 /0.28/,
					0.175 /0.0 /0.325/0.31/$\ell_b$,
					0.35  /0.0 /0.45 /0.35/,
					0.475 /0.0 /0.575/0.29/,
					0.575 /0.0 /0.65 /0.29/,
					0.65  /0.0 /0.7  /0.53/,
					0.7   /0.0 /0.75 /0.52/,
					0.75  /0.0 /0.85 /0.30/,
					0.875 /0.0 /0.925/0.29/$r_b$
				}
				{
					\drawTallItem[\small \z]{\x*\w}{\y*\hprime}{\xx*\w}{\yy*\hprime};
				}

				\foreach \y/\z in {
					0.5*\hprime 		/ $\nicefrac{1}{2}\iH{B}$,
					0.25*\h     		/ $\nicefrac{1}{4}H$,
					\hprime - 0.25*\h	/ $\iH{B} - \nicefrac{1}{4}H$,
					\hprime				/ $\iH{B}$,
					\hprime + 0.25*\h	/ $\iH{B} + \nicefrac{1}{4}H$
				}{
					\draw[dotted] (-0.1*\w,\y)node[left]{\z} -- (1.1*\w,\y) ;
				}
				
				\foreach \x/\y/\xx/\yy in {
					0.0*\w/\hprime-0.25*\h/0.2*\w/\hprime + 0.25*\h,  
					0.9*\w/\hprime-0.25*\h/1.0*\w/\hprime +0.25*\h,   
					0.0*\w/\hprime-0.25*\h/0.2*\w/0,                  
					0.9*\w/\hprime-0.25*\h/1.0*\w/0,                  
					0.2*\w/0/0.325*\w/0.5*\hprime,                    
					0.2*\w /0.5*\hprime/0.35*\w/\hprime + 0.25*\h,    
					0.75*\w/0.5*\hprime/0.9*\w/\hprime + 0.25*\h,     
					0.65*\w/0.0*\hprime/0.75 *\w/\hprime -0.25*\h,    
					0.35*\w/0.5*\hprime/0.425*\w/\hprime + 0.25*\h,   
					0.425*\w/0.5*\hprime/0.525*\w/\hprime +0.25*\h,   
					0.525*\w/\hprime -0.25*\h/0.75*\w/\hprime +0.25*\h     
				}{
					\draw[very thick, red] (\x,\y) rectangle (\xx,\yy);
				}
				
				\draw[very thick,red] (0.325*\w,0*\hprime) -- (0.325*\w,\hprime/2) -- (0.525*\w,\hprime/2) -- (0.525*\w,\hprime -\h/4) -- (0.65*\w,\hprime -\h/4) -- (0.65*\w,0*\hprime) -- (0.325*\w,0*\hprime);
				
				\draw[very thick,red] (0.75*\w,0*\hprime) -- (0.75*\w,0.58*\hprime -\h/4) -- (0.75*\w,0.58*\hprime -\h/4)--(0.75*\w,\hprime/2) -- (0.875*\w,\hprime /2) -- (0.875*\w,0) --(0.7*\w,0);
				
				\end{tikzpicture}
				\caption{Reordered Packing}
				\label{fig:GenralReordering5}
			\end{figure}
			
			\begin{claim}
				After this step there is no item which overlaps another in the area $B_{l,9}$.
			\end{claim}
			
			\begin{proofClaim}
				First, no item from the bottom will overlap the items with height $\nicefrac{1}{2}\iH{B} +\nicefrac{1}{4} H$ from the top since their lower border is at $\nicefrac{1}{2}\iH{B}$ and the items below have a height less than $\nicefrac{1}{2}\iH{B}$ (otherwise they wold be contained in area $B_5$). 
				
				Let us assume there is an item $b$ from the bottom intersecting an item $t$ from the top at an inner point $(x,y)$ in the area $B_{l,9}$. 
				As a consequence, for each $x'$ larger than $x$ up to the left border of $B_5$, the point $(x',y)$ is overlapped by an item touching $\iH{B}-\nicefrac{1}{4} H$. 
				On the other hand for each $x' \leq x$ but right of $L_1$ (i.e., $x'\geq L_i$) the point $(x',y)$ is overlapped by an item from the bottom of the box. 
				Note that the total width of items with lower border below $y$ and above $\nicefrac{1}{4} H$  between $L_1$ and the left border of $i$ has not changed after the shifting of items with height $\nicefrac{1}{2}\iH{B}+\nicefrac{1}{4} H$ on the top of the box (the items left of $l_m$ have their lower border at $\iH{B}/2$). 
				Additionally, the total width of items touching the bottom of the box with upper border above $y$ in this area has not changed either. 
				Therefore, the total width of items overlapping the horizontal line at $y$ in this area is larger than the width of this area. 
				As a result the items must had have an overlapping before the first horizontal shift -- a contradiction. 
				Hence, there is no item overlapping another item in this area, which concludes the proof of the claim. 
			\end{proofClaim}

			\item[Items with height $\iH{B}$:] 
			Last, let us consider the case that we have (pseudo) items with height $\iH{B}$. 
			In this case in the very first step, we choose one of the items with height $\iH{B}$ and shift all the other items with this height to the left or to the right such that they are positioned next to this item.
			This shifting can be done, without slicing any tall item. 
			Afterward these items form an area $B_{10}$, which just contains items of height $\iH{B}$. 
			We sort those pseudo items of height $H$ which contain tall items so that tall items with the same height are placed next to each other. 
			Then we will search for $i_l$ left of this area and for $i_r$ right of this area. 
			This area divides the box and represents the splitting item $i$. 
			In this case the box $B_8$ is split into two parts $B_{l,8}$ and $B_{r,8}$, as well as the area $B_5$ is divided into $B_{l,5}$ and $B_r,5$. 
			All the following steps are done as described above. 
		\end{description}
	\end{stepList}
	
	
	\begin{description}
		\item[Analyzing the number of constructed boxes.]
		In the worst case we have (pseudo) items with height $\iH{B}$ and both $l$ and $r$ exist. Furthermore, the left border of $l_b$ should be right of the left border of $l$ as well as the left border of $r_b$ should be left of $r$. Until now we did not need the assumption, that the tall items are placed on an arithmetic grid. However, to count the generated boxes, it is convenient to make this assumption.
		First, we will analyze the number of boxes for tall items we generate. 
		
		\begin{claim}
			The number of boxes for tall items is bounded by $2 \NumGridpoints^2 + 15\NumGridpoints/4 +8$, where $H/N$ is the distance between the grid lines.
		\end{claim}
		\begin{proofClaim}
			We proof this claim by considering the generated areas one after another and count the number of boxes generated in each of these areas.
			
			\begin{description}
				\item[Area $B_{l,1}$: $\NumGridpoints/4$ boxes.] 
				In the areas $B_{l,1}$ and $B_{r,1}$ there are tall items with heights between $\nicefrac{1}{4} H$ and $\nicefrac{1}{2} H$ on the top of the box. For each of these sizes we generate at most one box in each area. Therefore, both contain at most $\NumGridpoints/4$ boxes for tall items. 
				
				\item[Area $B_{l,2}$: $\NumGridpoints^2/2 + \NumGridpoints/4+3$ boxes.]
				In the boxes $B_{l,2}$ and $B_{r,2}$, we create at most one box for each item height larger than $\nicefrac{1}{2}\iH{B}$ and lower than $\iH{B}/3$. There are at most $\NumGridpoints/4$ sizes. 
				For the other occurring sizes we create by Lemma~\ref{lma:reorderingPreviusly} at most $4S_TS_{T\cup P} + 3$ boxes in total since there are at most three unmovable items overlapping this area. 
				We have $S_T \leq \NumGridpoints/4$ since the tall items have heights between $\nicefrac{1}{4} H$ and $\nicefrac{1}{2}\iH{B}$, $S_P \leq \NumGridpoints/2$ since they have heights smaller than $\nicefrac{1}{2}\iH{B}$, and $S_{T\cup P} \leq \NumGridpoints/2$ as a consequence. Therefore, we create at most $4 \frac{\NumGridpoints}{4}\frac{\NumGridpoints}{2} + \frac{\NumGridpoints}{4} +3= \NumGridpoints^2/2 + \NumGridpoints/4+3$ boxes for tall items in each of the areas $B_{l,2}$ and $B_{r,2}$. 
				
				\item[Area $B_{l,3}$: $0$ boxes.] 
				The area $B_{l,3}$ just contains the item $l_b$ as a tall item. Since this item overlaps the area $B_{l,2}$, we have already counted this item.
				
				\item[Area $B_{l,4}$: $\NumGridpoints/4$ boxes.] 
				The area $B_{l,4}$ contains just tall items with height between $\nicefrac{1}{2}\iH{B}$ and $\iH{B}-\nicefrac{1}{4} H$. For each size we create one box. Therefore, we create at most $\NumGridpoints/4$ boxes for tall items in this area.
				
				\item[Area $B_{l,5}$: $\NumGridpoints/4$ boxes.]
				The area $B_{l,5}$ contains the tall items with height between $\nicefrac{1}{2}\iH{B}$ and $\iH{B}-\nicefrac{1}{4} H$. 
				For each of these sizes we create at most one box, resulting in at most $\NumGridpoints/4$ boxes in this area.
				
				\item[Area $B_{l,5}$: $\NumGridpoints^2/2 +1$ boxes.] 
				In the areas $B_{l,6}$ and $B_{r,6}$ each tall and pseudo item has a size of less than $\nicefrac{1}{2}\iH{B}$. 
				Analogously to the boxes $B_{l,2}$ and $B_{r,2}$ we create by Lemma~\ref{lma:reorderingPreviusly} at most $\NumGridpoints^2/2 +1$ boxes for tall items per area $B_{l,3}$ and $B_{r,3}$ since there is at most one overlapping item. 
				
				\item[Area $B_{l,6}$: $\NumGridpoints/4$ boxes.] 
				The area $B_{l,7}$ or $B_{r,7}$ is the area containing $l_m$ or $r_m$ respectively.  
				Above $l_m$ and $r_m$ we create at most $\NumGridpoints/4$ boxes for tall items each since the tall items have a height of at most $\nicefrac{1}{2}\iH{B}$ and at least $\nicefrac{1}{4} H$. The box for the item overlapping $L_3$ is already counted. 
				
				\item[Area $B_m$: $\NumGridpoints/4$ boxes.]
				$B_8$ is divided into two boxes, if we have (pseudo) items with height $\iH{B}$.
				In each of these parts each tall item has height at most $\nicefrac{1}{2}\iH{B}$ and we create one box per item size. 
				Therefore, we create at most $\NumGridpoints/4$ boxes in this area in each part. 
				The boxes for the items overlapping $L_2$ or $R_2$ are already counted for area $B_{l,7}$.
				
				\item[Area $B_{l,7}$: $2\NumGridpoints/4$ boxes.] 
				We consider now the areas $B_{l,9}$ and $B_{r,9}$. 
				In these areas, all items have height of at most $\nicefrac{1}{2}\iH{B}$ and for each item height we create one box at the bottom and one box at $\iH{B}-\nicefrac{1}{4} H$. 
				Therefore, we create at most $2 \NumGridpoints/4$ boxes in each area.  
				
				\item[Area $B_{l,9}$: $\NumGridpoints/4$ boxes.] 
				Last, we create at most one box for each item with height larger than $\iH{B}-\nicefrac{1}{4} H$ resulting in at most $\NumGridpoints/4$ boxes for these items.
			\end{description}
			
			In total the number of generated boxes is bounded by 
			$2(\NumGridpoints/4 + \NumGridpoints^2/2 + \NumGridpoints/4+3 + \NumGridpoints/4+\NumGridpoints^2/2 +1 + \NumGridpoints/4 + 2\NumGridpoints/4 + \NumGridpoints/4) + \NumGridpoints/4 = 2 \NumGridpoints^2 + 15\NumGridpoints/4 +8$, which concludes the proof of the claim.
		\end{proofClaim}
		
		Let us consider the number of boxes for vertical items. 
		
		\begin{claim}
			The number of boxes for vertical items is bounded by  $4\NumGridpoints^2 + 31\NumGridpoints/4 +5$.
		\end{claim}
		
		\begin{proofClaim}
			We proof this claim by considering the generated areas one after another and count the number of boxes generated in each of these areas.
			\begin{description}
				\item[Area $B_{l,1}$: $\NumGridpoints/2$ boxes.] 
				In the areas $B_{l,1}$ and $B_{r,1}$, there are at most $\NumGridpoints/4$ boxes for items touching the bottom since they have height of at most $\nicefrac{1}{4} H$ and at most $\NumGridpoints/4$ boxes for items touching the top since they have height of at least $\nicefrac{1}{4} H$ and at most $\nicefrac{1}{2}\iH{B}$. 
				Therefore, in each of the areas $B_{l,1}$ and $B_{r,1}$ we generate at most $\NumGridpoints/2$ boxes.
				\item[Area $B_{l,2}$ -- $\NumGridpoints^2 +2$ boxes]: 
				In the areas $B_{l,1}$ and $B_{r,1}$ the pseudo items touching the bottom have sizes between $\nicefrac{1}{4} H$ and $\nicefrac{1}{2}\iH{B}$ and the items touching the top have sizes up to $\nicefrac{1}{2}\iH{B}$. By Lemma~\ref{lma:reorderingPreviusly} we generate at most $4S_PS_{P\cup T} \leq 4 \frac{\NumGridpoints}{2}\frac{\NumGridpoints}{2} = \NumGridpoints^2$ boxes plus the two boxes for extending the unmovable items in each area. Therefore, in the areas $B_{l,2}$ and $B_{r,2}$, we create at most $\NumGridpoints^2 +2$ boxes for pseudo items each.
				
				\item[Area $B_{l,3}$: $\NumGridpoints/4$ boxes.] 
				In the area $B_{l,3}$  and $B_{r,3}$ above the items $l_b$ and $r_b$ respectively, there are pseudo items with heights up to $\nicefrac{1}{4} H$. For each size we generate at most one box. Therefore, we generate at most $\NumGridpoints/4$ boxes in each of these areas. 
				
				\item[Area $B_{l,4}$: $\NumGridpoints/2$ boxes.] 
				In the area $B_{l,4}$ below the items with height larger than $\nicefrac{1}{2}\iH{B}$, we have areas for pseudo items with height at most $\nicefrac{1}{4} H$. 
				We have two blocks of these items, one at $l$ the other at $r$. 
				In each of these areas, we create at most $\NumGridpoints/4$ boxes for these items. 
				Furthermore, there can be pseudo items with heights between $\nicefrac{1}{2}\iH{B}$ and $\iH{B}-\nicefrac{1}{4} H$ for each of these heights we create at most one box resulting in $\NumGridpoints/4$ boxes for these items in each area $B_{l,4}$ and $B_{r,4}$.
				
				\item[Area $B_{l,6}$: $\NumGridpoints^2$ boxes.] 
				In the areas $B_{l,6}$ and $B_{r,6}$ the pseudo items have heights between $\nicefrac{1}{4} H$ and $\nicefrac{1}{2}\iH{B}$ on the top and heights up to $\nicefrac{1}{2}\iH{B}$ on the bottom. 
				Therefore, by Lemma~\ref{lma:reorderingPreviusly} we generate at most $\NumGridpoints^2$ boxes analogously to the boxes $B_{l,2}$. Here, we do not create another pseudo item since the item $l_{t,l}$ already touches the top of the area. Therefore, we generate at most $\NumGridpoints^2$ boxes in each of the areas $B_{l,5}$ and $B_{r,5}$.
				
				\item[Area $B_{l,7}$: $\NumGridpoints$ boxes.] 
				In the area $B_{l,7}$ above $l_m$, we have at most three areas for pseudo items touching $l_m$ with their lower border. 
				These items have a height of at most $\nicefrac{1}{4} H$. 
				Therefore, we create at most $3 \NumGridpoints/4$ for these pseudo items. 
				Furthermore, the pseudo items touching $\iH{B}+\nicefrac{1}{4} H$ with their top have a height between $\iH{B}/4$ and $\nicefrac{1}{2}\iH{B}$. 
				For each height we generate at most one box. 
				Therefore, we create at most $\NumGridpoints/4$ boxes for these items. 
				In total we generate at most $\NumGridpoints$ boxes for pseudo items in the areas $B_{l,7}$ and $B_{r,7}$ each. 
				
				\item[Area $B_8$: $\NumGridpoints/2$ boxes.] 
				In the area $B_8$ pseudo items with height up to $\nicefrac{1}{4} H$ touch the bottom and items with sizes between $\nicefrac{1}{4} H$ and $\nicefrac{1}{2}\iH{B}$ are touching the top. 
				For each size we generate at most one box. 
				Therefore, in $B_8$ we generate at most $\NumGridpoints/2$ boxes. 
				The area $B_8$ can be split in two by the items with height $\iH{B}$. 
				Therefore, we have to count the boxes in $B_8$ twice.
				
				\item[Area $B_{l,9}$: $3 \NumGridpoints/4$ boxes.] 
				In the area $B_{l,9}$ the tall items on the bottom have height between $\nicefrac{1}{4} H$ and $3\iH{B}/4$. 
				For each of these sizes we create at most one box, summing up to at most $\NumGridpoints/4$. 
				On the top of this area the pseudo items have heights up to $\nicefrac{1}{2}\iH{B}$ and we create one box per size, creating at most $\NumGridpoints/2$ boxes. 
				Therefore in the areas $B_{l,4}$, we have at most $3 \NumGridpoints/4$ boxes in total. 
				
				\item[Area $B_{l,5}$: $\NumGridpoints/4$ boxes.] 
				In the area $B_{l,5}$ above the tall items with height between $\nicefrac{1}{2}\iH{B}$ and $\iH{B}-\nicefrac{1}{4} H$ there are no pseudo items. 
				They where shifted up, to have their lower border at $\iH{B}-H/4$. 
				This area contains just pseudo items with height between $\nicefrac{1}{2}\iH{B}$ and $\iH{B}-\nicefrac{1}{4} H$ and for each size we create at most one box, hence at most $\NumGridpoints/4$ boxes.
				
				\item[Area $B_{l,9}$: $\NumGridpoints/4$ boxes.]
				Last, we consider the items with height larger than $3\iH{B}/4$ in area $B_{l,10}$. Above these items there can be pseudo items with heights up to $\nicefrac{1}{4} H$. For each size we create at most one box. Therefore, we create at most $\NumGridpoints/4$ boxes above these items. Furthermore, there can be at most one box contain a pseudo item with height $\iH{B}$.
			\end{description}
			
			In total we create at most 
			$2(\NumGridpoints/2 + \NumGridpoints^2 +2+ \NumGridpoints/4 + \NumGridpoints/2 +\NumGridpoints^2 +\NumGridpoints + \NumGridpoints/2 + 3 \NumGridpoints/4 + \NumGridpoints/4) + \NumGridpoints/4 +1= 
			4\NumGridpoints^2 + 31\NumGridpoints/4 +5$ boxes for vertical items, which concludes the proof of the claim and, therefore, the proof of the lemma.
		\end{proofClaim}
	\end{description}
\end{proof}

In the case that the considered box $B$ has a height of at most $\nicefrac{3}{4} H$, there are at most two tall items on top of each other. 
In this box, we can shift the tall items to the top and to the bottom and generate pseudo items as described in the previous proof. 
Pseudo items, which are positioned vertically between two tall items are removed and placed in an extra box. 
The extra box is placed into a gap, which will be generated by shifting all boxes with lower border above $\nicefrac{3}{4} H$ exactly $\nicefrac{1}{4} H$ upwards.

After removing these pseudo items, the tall and pseudo items still inside $B$ have height $h(B)$ or a height between $\nicefrac{1}{4} H$ and $h(B) -\nicefrac{1}{4} H$. 
Therefore, tall and pseudo items have at most $\NumGridpoints/4$ different heights in this area and the smallest items touching the top or the bottom have a height of at least $\nicefrac{1}{4} H$.
Furthermore, the difference between heights is at least $H/\NumGridpoints$. 
Therefore, we can conclude the following lemma from what is proven in~\cite{JansenR16}.

\begin{lemma}[\cite{JansenR16}]
	\label{lma:reorderingMediumBoxes}
	Let $B$ be a box with height $\nicefrac{1}{2} H < h(B) \leq \nicefrac{3}{4} H$. 
	We can rearrange the items in this area, such that we generate at most $\Oh(\NumGridpoints^2)$ boxes for tall items and at most $\Oh(\NumGridpoints^2)$ boxes for sliced vertical items plus one additional box of height $\nicefrac{1}{4} H$ and width $(1 - 1/\NumGridpoints)w(B)$.
\end{lemma}

If the considered box $B$ has a height of at most $\nicefrac{1}{2} H$ there can be just one tall item per vertical line through this box. 
In~\cite{JansenR16} this simple case was already studied and the following lemma is an adaption of what was proven for this scenario.

\begin{lemma}[\cite{JansenR16}]
	\label{lma:reorderingSmallBoxes}
	Let $B$ be a box with height $h(B) \leq \nicefrac{1}{2} H$. 
	We can rearrange the items in this area, such that we generate at most $\NumGridpoints/4 +1$ boxes for tall items and at most $\NumGridpoints/4 +1$ boxes for sliced vertical items.
\end{lemma}

In summary, the worst case where we generate the most sub boxes is if $h(B) > \nicefrac{3}{4}H$.

\section{Structure Result}
\label{sec:StructurResult}
In this section, we prove the key to achieve the approximation ratio $(5/4+\eps)\OPT$ -- the structural lemma.
Roughly it states that each optimal solution can be transformed such that it has a simple structure, see Lemma~\ref{lma:structureLemma}.
The hart of the proof -- to reorder the items inside the boxes of height taller than $\nicefrac{3}{4}$ -- was discussed in the previous section.
However, one challenge remains to be resolved: 
the placement of a constant number of extra boxes for vertical items that is used to provide an integral packing of those items after the rearrangement step. 
Unlike in the approaches in~\cite{NadiradzeW16},~\cite{GalvezGIK16} or~\cite{JansenR16}, we cannot place them on the top of the packing since we have to extend the packing beforehand by $\nicefrac{1}{4}H$ using a shifting step to establish the simple structure. 
Fortunately, this shift creates some free area. 
A careful analysis of this area shows that this it can be used to place the boxes inside.

In this section, we will assume that we are given an instance with set of items $\items$ and an $\eps \in \RR$ such that $1/\eps$ is integral.
Furthermore, we are given an optimal packing of the items $\items$ with height $\OPT$.

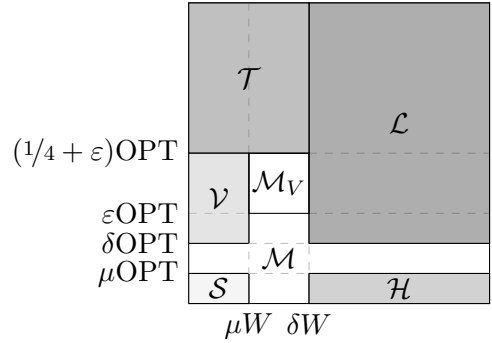
\begin{figure}[ht]
	\centering
	\setboolean{BlackAndWhite}{true}
	\provideboolean{paintWhite}
	\setboolean{paintWhite}{false}
	\begin{tikzpicture}
	\pgfmathsetmacro{\w}{4}
	\pgfmathsetmacro{\h}{4}
	
	\pgfmathsetmacro{\dW}{0.4}
	\pgfmathsetmacro{\mW}{0.2}
	
	\pgfmathsetmacro{\aOpt}{0.5}
	\pgfmathsetmacro{\eOpt}{0.3}
	\pgfmathsetmacro{\dOpt}{0.2}
	\pgfmathsetmacro{\mOpt}{0.1}
	
	\draw (0,0) rectangle (\w,\h);
	\draw (\dW*\w,0) node[below]{$\delta W$} -- (\dW*\w,\mOpt*\h);  
	\draw (\dW*\w,\dOpt*\h) -- (\dW*\w,\h); 
	
	\draw (\mW*\w,0) node[below]{$\mu W$} -- (\mW*\w,\mOpt*\h);
	\draw (\mW*\w,\dOpt*\h) -- (\mW*\w,\aOpt*\h);
	
	\draw (0,\aOpt*\h) node[left]{$(\nicefrac{1}{4} + \eps) \OPT$} -- (\dW*\w,\aOpt*\h);
	\draw (0,\dOpt*\h)node[left]{$\delta \OPT$} -- (\mW*\w,\dOpt*\h);
	\draw (\dW*\w,\dOpt*\h) -- (\w,\dOpt*\h);  
	\draw (0,\mOpt*\h)node[left]{$\mu \OPT$} -- (\mW*\w,\mOpt*\h);
	\draw (\dW*\w,\mOpt*\h) -- (\w,\mOpt*\h);

	\ifthenelse{\boolean{paintWhite}}{%
		\draw[dashed,lightgray] (\mW*\w,\aOpt*\h) -- (\mW*\w,\h);
		\draw[dashed,lightgray] (\dW*\w,\aOpt*\h) -- (\w,\aOpt*\h);
		\draw[dashed,lightgray] (0,\eOpt*\h) node[left,black]{$\eps \OPT$} -- (\mW*\w,\eOpt*\h);   
		\draw[dashed,lightgray] (\dW*\w,\eOpt*\h) -- (\w,\eOpt*\h); 
		\draw (\dW *\w, \dOpt *\h) rectangle node[midway] {$\itemsL$}(\w,\h);
		\draw (0, \aOpt *\h) rectangle node[midway] {$\itemsT$}(\dW *\w,\h);
		\draw (0, \dOpt *\h) rectangle node[midway] {$\itemsV$}(\mW *\w,\aOpt *\h);
		\draw (0, 0) rectangle node[midway] {$\itemsS$}(\mW *\w,\mOpt *\h);
		\draw (\dW *\w, 0) rectangle node[midway] {$\itemsH$}(\w,\mOpt *\h);
	}{%
		\draw[dashed,gray] (\mW*\w,\aOpt*\h) -- (\mW*\w,\h);
		\draw[dashed,gray] (\dW*\w,\aOpt*\h) -- (\w,\aOpt*\h);
		\draw[dashed,gray] (0,\eOpt*\h) node[left,black]{$\eps \OPT$} -- (\mW*\w,\eOpt*\h);   
		\draw[dashed,gray] (\dW*\w,\eOpt*\h) -- (\w,\eOpt*\h); 
		\drawLargeItem[$\itemsL$]{\dW *\w}{\dOpt *\h}{\w}{\h};
		\drawTallItem[$\itemsT$]{0}{\aOpt *\h}{\dW *\w}{\h};
		\drawVerticalItem[$\itemsV$]{0}{\dOpt *\h}{\mW *\w}{\aOpt *\h};
		\drawSmallItem[$\itemsS$]{0}{0}{\mW *\w}{\mOpt *\h};
		\drawHorizontalItem[$\itemsH$]{\dW *\w}{0}{\w}{\mOpt *\h};
	}
	\draw (\mW *\w, \eOpt *\h) rectangle node[midway] {$\itemsMV$}(\dW *\w,\aOpt *\h);	
	\draw[dashed,lightgray] (\mW *\w, \mOpt *\h) rectangle node[midway,black] {$\itemsM$}(\dW *\w,\dOpt *\h);
	\end{tikzpicture}
	\caption{Partition of the items. Each item can be represented by a point in this two-dimensional plane. The x-coordinate represents the items width while the y-coordinate represent its height.}
	\label{fig:partition}
\end{figure}

In the first simplification step, we partition the set of items $\items$, see Figure~\ref{fig:partition} for an overview. 
Let $\delta = \delta(\eps) \leq \eps$ and $\mu = \mu(\eps) <\delta$ be suitable constants depending on $\eps$, and let $\OPT$ be the height of an optimal packing.
We define 
\begin{itemize}
	\item $\itemsL := \sset{i \in \items}{\iH{i} > \delta \OPT, \iW{i} \geq \delta W}$ as the set of large items, 
	\item $\itemsT := \sset{i \in \items}{\iH{i} \geq (1/4 + \eps)\OPT, \iW{i} < \delta W}$ as the set of tall items, 
	\item $\itemsV := \sset{i \in \items}{\delta\OPT \leq \iH{i} < (1/4 + \eps)\OPT, \iW{i} \leq \mu W}$ as the set of vertical items, 
	\item $\itemsMV : = \sset{i \in \items}{\eps \OPT \leq \iH{i} < (1/4 + \eps)\OPT, \mu W < \iW{i} \leq \delta W}$ as the set of vertical medium items, 
	\item $\itemsH := \sset{i \in \items}{\iH{i} \leq  \mu \OPT, \delta W \leq \iW{i}}$ as the set of horizontal items, 
	\item $\itemsS := \sset{i \in \items}{\iH{i} \leq  \mu \OPT,  \iW{i} \leq \mu W}$ as the set of small items and
	\item  $\itemsM := \sset{i \in \items}{\iH{i} < \eps \OPT, \mu W < \iW{i} \leq \delta W} \cup \sset{i \in \items}{\mu \OPT < \iH{i} \leq \delta \OPT} = \items \setminus (\itemsL \cup \itemsT \cup \itemsV \cup \itemsMV \cup \itemsH \cup \itemsS)$ as the set of medium sized items.
\end{itemize}

We want to choose $\delta$ and $\mu$ such that the total area of the items in $\itemsM$ and $\itemsMV$ is small. The following Lemma states that we can find such suitable values for $\delta$ and $\mu$. 

\begin{lemma}
	\label{lma:DeltaMu}
	Let $f:\RR \to \RR$ be any function, such that $1/f(\eps)$ is integral. Consider the sequence $\sigma_0 = f(\eps)$, $\sigma_{i+1} = \sigma_i^2 f(\eps)$. There is a value $i \in \{0, \dots, (2/f(\eps))-1\}$ such that the total area of the items in $\mathcal{M} \cup \mathcal{M}_V$ is at most $f(\eps) W\OPT$, if we set $\delta := \sigma_i$ and $\mu := \sigma_{i+1}$.
\end{lemma}
\begin{proof}
	This can be proven by the pidginhole principle. The sequence $\sigma$ and the corresponding choice of $\delta$ and $\mu$ builds a sequence of $2/f(\eps)$ sets $\mathcal{M}_{\sigma_i}\cup \mathcal{M}_{V\sigma_i}$.
	Each item $i \in I$ can occur in at most two of these sets, either because of its width or its height.  
	Since the total area of all items is at most $W\cdot\OPT$ one of the sets must have an area which is at most $f(\eps)\cdot W \cdot \OPT$.
\end{proof}

For this application it is sufficient to choose $f(\eps) = \nicefrac{\eps^{13}}{k}$ for a constant $k \in \NN$ which has to fulfill certain properties, as can be seen later. 
Since $1/\eps \in \NN$, we have that $1/f(\eps) \in \NN$.
Let $\delta$ and $\mu$ be the values defined as in Lemma~\ref{lma:DeltaMu}.
Note that $\sigma_i = f(\eps) ^{(2^{i+1}-1)}$ and, therefore, $\delta \geq \sigma_{f(\eps)  -1}\geq (\eps^{13}/k)^l \in \eps^{\Oh(l)}$, where $l:=(2^{2k/\eps^{13}})$, i.e., $\delta \geq \eps^{2^{\Oh(1/\eps^{13})}}$.
In the following steps, we need $\delta$ to be of the form $\eps^x$ for some $x \in \NN$.
Therefore, define $\delta' := \eps^x$, such that $x \in \NN$ and $\delta' \leq \delta \leq \delta'/\eps$. 
Note that $\mu := \delta^2\eps^{13}/k \leq (\delta'/\eps)^2\eps^{13}/k = \delta'^2\eps^{11}/k$ and $\mu := \delta^2\eps^{13}/k \geq \delta'^2\eps^{13}/k$.
In the following we will use $\delta'$ for all the steps, but omit the prime for simplicity of notation. 
By this choice it still holds that the set of medium sized items has a total area of at most $(\eps^{13}/k)W\OPT$ because by reducing $\delta$ and not changing $\mu$ we only removed jobs from this set.

\begin{observation}
	Since each item in $\itemsMV$ has a height of at least $\eps \OPT$ and width of at least $\mu W \geq (\delta^2 \eps^{13}/k) W$, i.e., an area of at least $(\delta^2 \eps^{14}/k) W\OPT$, it holds that 
	$$|\itemsMV| \leq (\eps^{13}/k)W\OPT/ ((\delta^2 \eps^{14}/k) W\OPT) = 1/\delta^2\eps.$$
\end{observation}

After we have found the corresponding values for $\delta$ and $\mu$ and after we have partitioned the set of items accordingly, we round the height of all items with height at least $\delta \OPT$ as in the following Lemma.

\begin{lemma}[\cite{JansenR16}]
	\label{lma:rounding}
	Let $\delta = \eps^k$ for some value $k \in \NN$. At loss of a factor $(1+2\eps)$ in the approximation ratio, we can ensure that each item $i$ with height $\eps^{l-1}\OPT\geq h(i) \geq \eps^l\OPT$ for some $l\in \NN \leq k$ has height $k_i \eps^{l+1}\OPT$ for a value $k_i \in \{1/\eps,\dots 1/\eps^2 -1\}$. Furthermore, the items upper and lower border can be placed at multiples of $\eps^{l+1}\OPT$. 
\end{lemma}

This is possible since $\delta$ is of form $\eps^x$.
Note that after this rounding the packing has a height of $(1+2\eps)\OPT$ and all the items with processing time larger than $\delta \OPT$ will start at integral multiples of $\eps \delta \OPT$, while all times taller than $\eps\OPT$ will start at integral multiples of $\eps^2 \OPT$.
Furthermore, the number of item heights larger than $\delta \OPT$ is bounded by $\Oh(\log_{\eps}(1/\delta)/\eps^2)$ and the number of heights larger than $\eps\OPT$ is bounded by $\Oh(1/\eps^2)$.


After this rounding step, we remove all items in $\itemsM \cup \itemsS$ from the optimal packing. 
Later, we will show that these items can be placed back into the packing with the NFDH algorithm (see~\cite{CoffmanGJT80}) without increasing the packing height too much, see Lemma~\ref{lma:smallItems} and Lemma~\ref{lma:mediumItems}.

At this point, the considered packing has a height of at most $(1+2\eps)\OPT$ and contains the items $\itemsL \cup \itemsT \cup \itemsV \cup \itemsMV \cup \itemsH$.
When rearranging the packing, we are allowed to slice the items in $\itemsV$ vertically, while all the other items cannot be sliced.
Therefore, in order to use the techniques from Section~\ref{sec:ReorderingInTheGeneralCase}, we need to partition the packing area into \subboxes that divide the vertical and tall items from the residual ones.
The following lemma states that this division is possible by introducing a constant number of \subboxes. 
The lemma, as stated in \cite{JansenR16}, does not consider the set of medium vertical items $\itemsMV$.
The following adaptation, however, can be shown by a simple extension of the proof by handling them as if they where large. 
We refer to \cite{JansenR16} for details.

\begin{figure}[ht]
	\setboolean{BlackAndWhite}{true}
	\begin{subfigure}{.49\textwidth}
		\centering
		
		\begin{tikzpicture}
		\pgfmathsetmacro{\w}{4}
		\pgfmathsetmacro{\h}{6}
		\draw[white] (-0.25*\w,-\h/8) rectangle (1.25*\w,5.5*\h/4);
		\draw[] (0*\w,0*\h) rectangle (\w,\h);
		
		\begin{scope}[]
		\foreach \x in {0,0.5,1,1.5,...,10} {
			\draw[dashed,lightgray] (-0.1*\w,0.1*\x*\h) --(1.1*\w,0.1*\x*\h);
		}
		\draw[|-|] (1.1*\w,0*\h) -- node[midway, right] {$\eps\delta \OPT$}(1.1*\w,0.05*\h);
		\end{scope}
		\foreach \x/\y/\xx/\yy in {
			0.900/0.025/0.950/0.050,
			0.500/0.000/0.550/0.025,
			0.600/0.050/0.650/0.075,
			0.900/0.075/0.950/0.100,
			0.900/0.450/0.950/0.475,
			0.950/0.550/0.975/0.575,
			0.950/0.575/0.975/0.600,
			0.900/0.725/0.925/0.750,
			0.925/0.725/0.950/0.750,
			0.900/0.925/0.925/0.950,
			0.900/0.900/0.925/0.925,
			0.700/0.900/0.725/0.925,
			0.025/0.950/0.075/0.975,
			0.050/0.925/0.100/0.950,
			0.075/0.950/0.100/0.975,
			0.025/0.975/0.050/1.000,
			0.625/0.950/0.675/0.975,
			0.600/0.925/0.650/0.950,
			0.600/0.950/0.625/0.975,
			0.675/0.950/0.700/0.975,
			0.000/0.025/0.025/0.050,
			0.450/0.025/0.500/0.050,
			0.000/0.050/0.050/0.075,
			0.475/0.075/0.500/0.100,
			0.875/0.750/0.900/0.775,
			0.700/0.775/0.725/0.800,
			0.850/0.800/0.900/0.825,
			0.875/0.850/0.900/0.875,
			0.400/0.850/0.425/0.875,
			0.450/0.800/0.500/0.825,
			0.450/0.575/0.500/0.600,
			0.450/0.550/0.475/0.575,
			0.475/0.550/0.500/0.575,
			0.600/0.525/0.650/0.550,
			0.650/0.525/0.675/0.550,
			0.675/0.525/0.700/0.550,
			0.600/0.500/0.625/0.525,
			0.625/0.500/0.675/0.525,
			0.675/0.500/0.700/0.525,
			0.600/0.475/0.650/0.500,
			0.000/0.375/0.025/0.400,
			0.050/0.375/0.075/0.400,
			0.075/0.375/0.125/0.400,
			0.125/0.375/0.150/0.400,
			0.150/0.375/0.175/0.400,
			0.400/0.375/0.425/0.400,
			0.450/0.375/0.500/0.400,
			0.550/0.375/0.575/0.400,
			0.575/0.375/0.600/0.400,
			0.000/0.350/0.025/0.375,
			0.050/0.350/0.100/0.375,
			0.100/0.350/0.125/0.375,
			0.125/0.350/0.175/0.375,
			0.400/0.350/0.425/0.375,
			0.450/0.350/0.475/0.375,
			0.475/0.350/0.500/0.375,
			0.550/0.350/0.600/0.375,
			0.550/0.175/0.600/0.200,
			0.550/0.150/0.600/0.175,
			0.650/0.150/0.675/0.175,
			0.025/0.125/0.075/0.150,
			0.075/0.125/0.125/0.150,
			0.125/0.125/0.175/0.150,
			0.450/0.125/0.500/0.150,
			0.550/0.125/0.575/0.150,
			0.875/0.125/0.900/0.150,
			0.025/0.100/0.050/0.125,
			0.050/0.100/0.100/0.125,
			0.100/0.100/0.150/0.125,
			0.150/0.100/0.175/0.125,
			0.450/0.100/0.500/0.125
		}{
			\drawSmallItem{\x*\w}{\y*\h}{\xx*\w}{\yy*\h};
		}

		\foreach \x/\y/\xx/\yy in {
			0.40/0.40/0.600/0.55,
			0.65/0.20/0.900/0.45,
			0.10/0.70/0.250/0.90,
			0.00/0.40/0.300/0.60,
			0.70/0.50/0.900/0.70,
			0.50/0.55/0.700/0.80,
			0.10/0.60/0.275/0.70,
			0.05/0.15/0.175/0.35
		}{
			\drawLargeItem{\x*\w}{\y*\h}{\xx*\w}{\yy*\h};
		}
		\foreach \x/\y/\xx/\yy in {
			0.050/0.050/0.600/0.075,
			0.500/0.075/0.900/0.100,
			0.500/0.800/0.850/0.825,
			0.050/0.900/0.700/0.925,
			0.700/0.700/0.950/0.725,
			0.600/0.975/1.000/1.000,
			0.700/0.950/1.000/0.975,
			0.000/0.000/0.500/0.025,
			0.550/0.000/0.950/0.025,
			0.500/0.025/0.900/0.050,
			0.650/0.050/0.950/0.075,
			0.050/0.975/0.400/1.000,
			0.100/0.925/0.600/0.940,
			0.100/0.940/0.600/0.960,
			0.100/0.960/0.600/0.975,
			0.650/0.475/0.950/0.500,
			0.600/0.450/0.900/0.475,
			0.400/0.975/0.600/1.000,
			0.650/0.925/0.900/0.950,
			0.000/0.075/0.475/0.100,
			0.725/0.900/0.900/0.925,
			0.025/0.025/0.450/0.050,
			0.400/0.875/0.900/0.900,
			0.425/0.850/0.875/0.875,
			0.450/0.825/0.725/0.850,
			0.725/0.825/0.900/0.850,
			0.725/0.775/0.900/0.800,
			0.700/0.750/0.875/0.775,
			0.700/0.725/0.900/0.750,
			0.550/0.100/0.900/0.125,
			0.575/0.125/0.875/0.150,
			0.650/0.175/0.900/0.200,
			0.675/0.150/0.900/0.175
		}{
			\drawHorizontalItem{\x*\w}{\y*\h}{\xx*\w}{\yy*\h};
		}
		\foreach \x/\y/\xx/\yy in {
			0.900/0.10/0.950/0.25,
			0.925/0.75/0.950/0.95,
			0.975/0.40/1.000/0.50,
			0.975/0.50/1.000/0.60,
			0.950/0.40/0.975/0.55,
			0.200/0.10/0.225/0.35,
			0.900/0.25/0.950/0.45,
			0.900/0.50/0.950/0.70,
			0.900/0.75/0.925/0.90,
			0.050/0.75/0.075/0.90,
			0.050/0.60/0.075/0.75,
			0.025/0.60/0.050/0.65,
			0.250/0.70/0.275/0.90,
			0.275/0.85/0.300/0.90,
			0.300/0.80/0.325/0.90,
			0.325/0.80/0.375/0.90,
			0.350/0.60/0.375/0.80,
			0.375/0.60/0.400/0.65,
			0.450/0.60/0.500/0.80,
			0.350/0.10/0.375/0.30,
			0.450/0.15/0.500/0.35,
			0.275/0.25/0.300/0.40,
			0.275/0.10/0.300/0.25,
			0.425/0.10/0.450/0.25,
			0.550/0.20/0.600/0.35,
			0.225/0.10/0.250/0.15,
			0.250/0.35/0.275/0.40,
			0.200/0.35/0.225/0.40,
			0.425/0.25/0.450/0.40
		}{
			\drawVerticalItem{\x*\w}{\y*\h}{\xx*\w}{\yy*\h};
		}

		\foreach \x/\y/\xx/\yy in {
			0.950/0.00/1.000/0.40,
			0.300/0.10/0.350/0.80,
			0.500/0.10/0.550/0.40,
			0.950/0.60/1.000/0.95,
			0.350/0.30/0.375/0.60,
			0.175/0.10/0.200/0.40,
			0.000/0.10/0.025/0.35,
			0.025/0.15/0.050/0.40,
			0.225/0.15/0.250/0.40,
			0.250/0.10/0.275/0.35,
			0.000/0.60/0.025/1.00,
			0.025/0.65/0.050/0.95,
			0.075/0.60/0.100/0.90,
			0.275/0.60/0.300/0.85,
			0.375/0.65/0.400/0.90,
			0.375/0.35/0.400/0.60,
			0.375/0.10/0.425/0.35,
			0.400/0.55/0.450/0.85,
			0.600/0.15/0.650/0.45
		}{
			\drawTallItem{\x*\w}{\y*\h}{\xx*\w}{\yy*\h};
		}
		
		\end{tikzpicture}
		\caption{rounded optimal packing}
		\label{fig:sub:optimalPacking}
	\end{subfigure}
	\hfill
	\begin{subfigure}{.49\textwidth}
		\centering
		
		\begin{tikzpicture}
		\pgfmathsetmacro{\w}{4}
		\pgfmathsetmacro{\h}{6}
		\draw[white] (-0.25*\w,-\h/8) rectangle (1.25*\w,5.5*\h/4);
		\draw[] (0*\w,0*\h) rectangle (\w,\h);
		
		\begin{scope}[]
		\foreach \x in {0,0.5,1,1.5,...,10} {
			\draw[dashed,lightgray] (-0.1*\w,0.1*\x*\h) --(1.1*\w,0.1*\x*\h);
		}
		\end{scope}
		\foreach \x/\y/\xx/\yy in {
			0.40/0.40/0.600/0.55,
			0.65/0.20/0.900/0.45,
			0.10/0.70/0.250/0.90,
			0.00/0.40/0.300/0.60,
			0.70/0.50/0.900/0.70,
			0.50/0.55/0.700/0.80,
			0.10/0.60/0.275/0.70,
			0.05/0.15/0.175/0.35
		}{
			\drawLargeItem{\x*\w}{\y*\h}{\xx*\w}{\yy*\h};
		}
		\foreach \x/\y/\xx/\yy in {
			0.050/0.050/0.600/0.075,
			0.500/0.075/0.900/0.100,
			0.500/0.800/0.850/0.825,
			0.050/0.900/0.700/0.925,
			0.700/0.700/0.950/0.725,
			0.600/0.975/1.000/1.000,
			0.700/0.950/1.000/0.975,
			0.000/0.000/0.500/0.025,
			0.550/0.000/0.950/0.025,
			0.500/0.025/0.900/0.050,
			0.650/0.050/0.950/0.075,
			0.050/0.975/0.400/1.000,
			0.100/0.925/0.600/0.940,
			0.100/0.940/0.600/0.960,
			0.100/0.960/0.600/0.975,
			0.650/0.475/0.950/0.500,
			0.600/0.450/0.900/0.475,
			0.400/0.975/0.600/1.000,
			0.650/0.925/0.900/0.950,
			0.000/0.075/0.475/0.100,
			0.725/0.900/0.900/0.925,
			0.025/0.025/0.450/0.050,
			0.400/0.875/0.900/0.900,
			0.425/0.850/0.875/0.875,
			0.450/0.825/0.725/0.850,
			0.725/0.825/0.900/0.850,
			0.725/0.775/0.900/0.800,
			0.700/0.750/0.875/0.775,
			0.700/0.725/0.900/0.750,
			0.550/0.100/0.900/0.125,
			0.575/0.125/0.875/0.150,
			0.650/0.175/0.900/0.200,
			0.675/0.150/0.900/0.175
		}{
			\drawHorizontalItem{\x*\w}{\y*\h}{\xx*\w}{\yy*\h};
		}
		\foreach \x/\y/\xx/\yy in {
			0.900/0.10/0.950/0.25,
			0.925/0.75/0.950/0.95,
			0.975/0.40/1.000/0.50,
			0.975/0.50/1.000/0.60,
			0.950/0.40/0.975/0.55,
			0.200/0.10/0.225/0.35,
			0.900/0.25/0.950/0.45,
			0.900/0.50/0.950/0.70,
			0.900/0.75/0.925/0.90,
			0.050/0.75/0.075/0.90,
			0.050/0.60/0.075/0.75,
			0.025/0.60/0.050/0.65,
			0.250/0.70/0.275/0.90,
			0.275/0.85/0.300/0.90,
			0.300/0.80/0.325/0.90,
			0.325/0.80/0.375/0.90,
			0.350/0.60/0.375/0.80,
			0.375/0.60/0.400/0.65,
			0.450/0.60/0.500/0.80,
			0.350/0.10/0.375/0.30,
			0.450/0.15/0.500/0.35,
			0.275/0.25/0.300/0.40,
			0.275/0.10/0.300/0.25,
			0.425/0.10/0.450/0.25,
			0.550/0.20/0.600/0.35,
			0.225/0.10/0.250/0.15,
			0.250/0.35/0.275/0.40,
			0.200/0.35/0.225/0.40,
			0.425/0.25/0.450/0.40
		}{
			\drawVerticalItem{\x*\w}{\y*\h}{\xx*\w}{\yy*\h};
		}

		\foreach \x/\y/\xx/\yy in {
			0.950/0.00/1.000/0.40,
			0.300/0.10/0.350/0.80,
			0.500/0.10/0.550/0.40,
			0.950/0.60/1.000/0.95,
			0.350/0.30/0.375/0.60,
			0.175/0.10/0.200/0.40,
			0.000/0.10/0.025/0.35,
			0.025/0.15/0.050/0.40,
			0.225/0.15/0.250/0.40,
			0.250/0.10/0.275/0.35,
			0.000/0.60/0.025/1.00,
			0.025/0.65/0.050/0.95,
			0.075/0.60/0.100/0.90,
			0.275/0.60/0.300/0.85,
			0.375/0.65/0.400/0.90,
			0.375/0.35/0.400/0.60,
			0.375/0.10/0.425/0.35,
			0.400/0.55/0.450/0.85,
			0.600/0.15/0.650/0.45
		}{
			\drawTallItem{\x*\w}{\y*\h}{\xx*\w}{\yy*\h};
		}
		
		\draw[fill = white, fill opacity = 0.4] (0,0) rectangle (\w,\h);

		\begin{scope}[very thick,black]
		\draw[] (0,0) rectangle (0.95*\w,0.05*\h);
		\draw[] (0.00*\w,0.05*\h) rectangle (0.95*\w,0.10*\h);
		\draw[] (0.55*\w,0.10*\h) rectangle (0.90*\w,0.15*\h);
		\draw[] (0.65*\w,0.15*\h) rectangle (0.90*\w,0.20*\h);
		\draw[] (0.60*\w,0.45*\h) rectangle (0.95*\w,0.50*\h);
		\draw[] (0.70*\w,0.70*\h) rectangle (0.95*\w,0.75*\h);
		\draw[] (0.70*\w,0.75*\h) rectangle (0.90*\w,0.80*\h);
		\draw[] (0.45*\w,0.80*\h) rectangle (0.90*\w,0.85*\h);
		\draw[] (0.40*\w,0.85*\h) rectangle (0.90*\w,0.90*\h);
		\draw[] (0.05*\w,0.90*\h) rectangle (0.90*\w,0.95*\h);
		\draw[] (0.05*\w,0.95*\h) rectangle (1.00*\w,1.00*\h);
		\draw[] (0.0*\w,0.6*\h) rectangle (0.05*\w,1.0*\h);
		\draw[] (0.0*\w,0.1*\h) rectangle (0.05*\w,0.4*\h);
		\draw[] (0.05*\w,0.6*\h) rectangle (0.1*\w,0.9*\h);
		\draw[] (0.175*\w,0.1*\h) rectangle (0.3*\w,0.4*\h);
		\draw[] (0.25*\w,0.7*\h) rectangle (0.275*\w,0.9*\h);
		\draw[] (0.275*\w,0.6*\h) rectangle (0.3*\w,0.9*\h);
		\draw[] (0.3*\w,0.1*\h) rectangle (0.4*\w,0.9*\h);
		\draw[] (0.4*\w,0.55*\h) rectangle (0.45*\w,0.85*\h);
		\draw[] (0.4*\w,0.1*\h) rectangle (0.55*\w,0.4*\h);
		\draw[] (0.45*\w,0.55*\h) rectangle (0.50*\w,0.8*\h);
		\draw[] (0.55*\w,0.15*\h) rectangle (0.6*\w,0.4*\h);
		\draw[] (0.6*\w,0.15*\h) rectangle (0.65*\w,0.45*\h);
		\draw[] (0.9*\w,0.1*\h) rectangle (0.95*\w,0.45*\h);
		\draw[] (0.9*\w,0.5*\h) rectangle (0.95*\w,0.7*\h);
		\draw[] (0.9*\w,0.75*\h) rectangle (0.95*\w,0.95*\h);
		\draw[] (0.95*\w,0.00*\h) rectangle (1.00*\w,0.95*\h);
		\draw (0.40*\w,0.40*\h) rectangle (0.60*\w,0.55*\h);
		\draw (0.65*\w,0.20*\h) rectangle (0.90*\w,0.45*\h);
		\draw (0.10*\w,0.70*\h) rectangle (0.25*\w,0.90*\h);
		\draw (0.00*\w,0.40*\h) rectangle (0.30*\w,0.60*\h);
		\draw (0.70*\w,0.50*\h) rectangle (0.90*\w,0.70*\h);
		\draw (0.50*\w,0.55*\h) rectangle (0.70*\w,0.80*\h);
		\draw (0.10*\w,0.60*\h) rectangle (0.275*\w,0.70*\h);
		\draw (0.05*\w,0.15*\h) rectangle (0.175*\w,0.35*\h);
		\end{scope}
		\end{tikzpicture}
		\caption{partition into boxes}
		\label{fig:sub:partition}
	\end{subfigure}
	\caption{In this figure one can see an optimal packing (in~\ref{fig:sub:optimalPacking}) and its partition into the rectangular subareas (in~\ref{fig:sub:partition}). 
		Note that some of the horizontal, vertical and tall items overlap the box borders.}
	\label{fig:partitionPacking}
\end{figure}
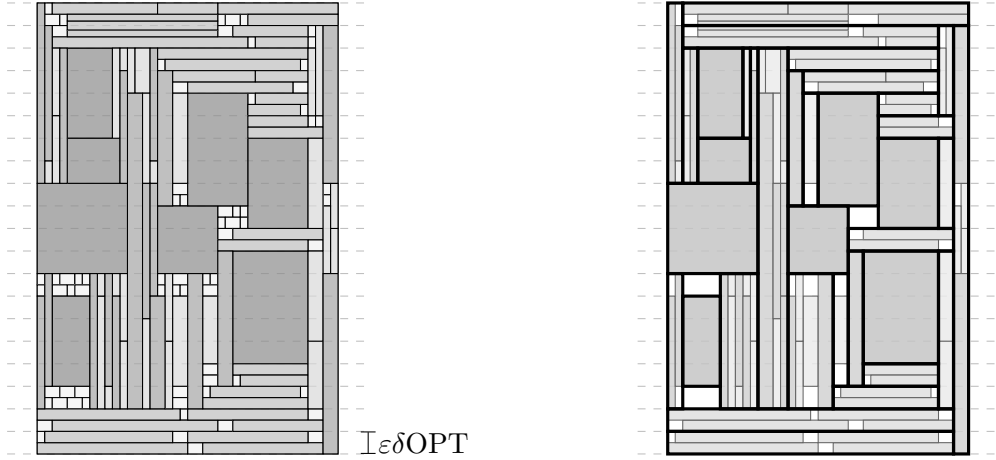

\begin{lemma}
	\label{lma:firstPartition}
	We can partition a rounded optimal packing, where the small and medium items are removed, into at most $ \mathcal{O}(1/\delta^2\eps)$ boxes such that the following conditions hold:
	\begin{itemize}
		\item There are $|\itemsL| + |\itemsMV| \leq \mathcal{O}(1/\delta^2\eps)$ boxes $\iBoxL$ each containing exactly one item from the set $\itemsL \cup \itemsMV$ and all items from this set are contained in these boxes. 
		\item There are at most $\mathcal{O}(1/\delta^2\eps)$ boxes $\iBoxH$ containing all horizontal items $\itemsH$, such that $\iBoxL \cap \iBoxH = \emptyset$. 
		The horizontal items can overlap horizontal box bofig:sub:optimalPackingrders, but never vertical box borders. 
		\item There are at most $\mathcal{O}(1/\delta^2\eps)$ boxes $\iBoxTV$ containing all items in $\itemsT \cup \itemsV$, such that $\iBoxTV \cap (\iBoxH \cup \iBoxL) = \emptyset$. 
		The items contained in these boxes can overlap vertical box borders, but never horizontal box borders. 
		\item The lower and upper border of each box is positioned at a multiple of $\eps\delta\OPT$.
	\end{itemize}
\end{lemma}

After applying the shifting and reordering technique from Section~\ref{sec:ReorderingInTheGeneralCase}, the vertical items will be sliced.
In the next lemma, we show that it is possible to place these items integral again. 
However, this integral placement comes at a cost. 
Namely, we have to introduce a constant number of narrow extra boxes for these items. 

\begin{lemma}
	\label{lma:verticalItems}
	Let $H_{\itemsV}$ be the set of different heights of vertical items and $\mu W$ the maximal width of a vertical item. 	
	Furthermore, let $\iBoxP$ be the set of boxes containing all sliced vertical items and only them.
	
	There exists a non fractional placement of the vertical items into the boxes $\iBoxP$ and at most $7(|H_{\itemsV}|+|\iBoxP|)$ additional boxes $\iBoxP'$ each of height at most $\nicefrac{1}{4} H$ and width $\mu W$, such that the boxes $\iBoxP \cup \iBoxP'$ are partitioned into at most $\Oh((|H_{\itemsV}|+|\iBoxP|)/\delta)$ \subboxes $\iBoxV$, containing only vertical items of the same height and at most $\Oh(|H_{\itemsV}|+|\iBoxP|)$ empty boxes $\iBoxS^{\mathcal{V}}$ with total area $\areaI{\iBoxS^{\mathcal{V}}} \geq \areaI{\iBoxP} - \areaI{\itemsV}$.
\end{lemma}
\begin{proof}	
	To prove this lemma, we first define a configuration LP.
	For this application, we define a configuration as follows.
	A configuration is a multiset of jobs, that can be placed on top of each other without exceeding the boundaries of a given box $B$, where we will place the given set of items. 
	More precisely $C = \sset{a_h:h}{ h \in H_{\itemsV}}$, the height of $C$ is given by $\iH{C} := \sum_{h \in H_{\itemsV}}h \cdot a_h$, and $\mathcal{C}_B$ is the set of configurations with heights at most $\iH{B}$. 
	Furthermore, we define for each $h \in H_{\itemsV}$ the value $w_h$ as the total width of all vertical items with height $h$.
	
	Consider the following configuration LP:
	\begin{align*}
	\sum_{C \in \mathcal{C}_B} X_C & = w(B) & \forall B \in \iBoxP\\
	\sum_{B \in \mathcal{B}_{\mathcal{V}}} \sum_{C \in \mathcal{C}_B} X_{C,B} a_{h,C} & = w_h & \forall  h \in H_{\itemsV}\\
	X_{C,B} & \geq 0 & \forall  B \in \mathcal{B}_{\mathcal{V}}, C \in \mathcal{C}_B
	\end{align*}
	It has, as each linear program, a basic solution with at most $|H_{\itemsV}|+|\iBoxP|$ non-zero components since it has at most this number of conditions. 
	
	Given such a basic solution, we place the corresponding configurations into the boxes. 
	Afterward, we place the items into the configurations, such that the last item overlaps the configuration border. 
	Each configuration has a height of at most $H$ since the boxes $\iBoxP$ have at most this height. 
	
	\begin{figure}[ht]
		\centering
		\begin{subfigure}[t]{.45\textwidth}
			\centering
			\begin{tikzpicture}
			
			\pgfmathsetmacro{\w}{4.5}
			\pgfmathsetmacro{\h}{4.5}
			
			\draw[very thick] (0,0) rectangle (\w,\h);
			\foreach \x/\xx/\i in {0.33/0/1,0.52/0.33/2,0.87/0.52/3,1/0.87/4}{
				\draw[very thick] (\x*\w,0) -- (\x*\w,\h);
				\draw[decorate,decoration={brace,amplitude=4pt}] (\x*\w,0*\h) -- (\xx*\w,0*\h) node[midway, below,yshift=-4pt,align=center]{$C_{\i}$};
			}
			
			\foreach \x in {0.05,0.1,0.15,...,0.95}{
				\draw[lightgray, dashed] (\x*\w,0) -- (\x*\w,\h);
			}
			
			\foreach \x/\y/\xx/\yy in {
				0   /0   /0.15/0.3,
				0.15/0   /0.25/0.3,
				0.25/0   /0.35/0.3,
				0   /0.3 /0.1 /0.6,
				0.1 /0.3 /0.25 /0.6,
				0.25 /0.3 /0.45/0.6,
				0   /0.6 /0.2 /0.9,
				0.2 /0.6 /0.4 /0.9,
				0.33/0   /0.43/0.5,
				0.43/0   /0.58/0.5,
				0.33/0.5 /0.48/0.6,
				0.48/0.5 /0.53/0.6,
				0.33/0.6 /0.38/0.7,
				0.38/0.6 /0.48/0.7,
				0.48/0.6 /0.58/0.7,
				0.33/0.7 /0.43/0.8,
				0.43/0.7 /0.53/0.8,
				0.52/0   /0.67/0.4,
				0.67/0   /0.77/0.4,
				0.77/0   /0.87/0.4,
				0.52/0.4 /0.62/0.6,
				0.62/0.4 /0.72/0.6,
				0.72/0.4 /0.82/0.6,
				0.82/0.4 /0.92/0.6,
				0.52/0.6 /0.67/0.7,
				0.67/0.6 /0.77/0.7,
				0.77/0.6 /0.92/0.7,
				0.87/0   /0.92/0.6,
				0.92/0   /1.02/0.6
			}{
				\drawVerticalItem{\x*\w}{\y*\h}{\xx*\w}{\yy*\h};
			}
			
			\end{tikzpicture}	
			\subcaption{Configurations inside a box $B$ with vertical items placed inside them.}
		\end{subfigure}
		\hfill
		\begin{subfigure}[t]{.45\textwidth}
			\centering
			\begin{tikzpicture}
			
			\pgfmathsetmacro{\w}{4.5}
			\pgfmathsetmacro{\h}{4.5}
			
			\draw[very thick] (0,0) rectangle (\w,\h);
			\foreach \x/\xx/\i in {0.3/0/1,0.45/0.3/2,0.8/0.45/3,0.9/0.8/4}{
				\draw[very thick] (\x*\w,0) -- (\x*\w,\h);
				\draw[decorate,decoration={brace,amplitude=4pt}] (\x*\w,0*\h) -- (\xx*\w,0*\h) node[midway, below,yshift=-4pt,align=center]{$C_{\i}$};
			}
			
			\foreach \x in {0.05,0.1,0.15,...,0.95}{
				\draw[lightgray, dashed] (\x*\w,0) -- (\x*\w,\h);
			}
			
			\foreach \x/\y/\xx/\yy in {
				0   /0   /0.15/0.3,
				0.15/0   /0.25/0.3,
				0   /0.3 /0.1 /0.6,
				0.1 /0.3 /0.25/0.6,
				0   /0.6 /0.2 /0.9,
				0.3 /0   /0.4 /0.5,
				0.3 /0.5 /0.45/0.6,
				0.3 /0.6 /0.35/0.7,
				0.35/0.6 /0.45/0.7,
				0.3 /0.7 /0.4 /0.8,
				0.45/0   /0.6 /0.4,
				0.6 /0   /0.7 /0.4,
				0.7 /0   /0.8 /0.4,
				0.45/0.4 /0.55/0.6,
				0.55/0.4 /0.65/0.6,
				0.65/0.4 /0.75/0.6,
				0.45/0.6 /0.6 /0.7,
				0.6 /0.6 /0.7 /0.7,
				0.8 /0   /0.85/0.6
			}{
				\drawVerticalItem{\x*\w}{\y*\h}{\xx*\w}{\yy*\h};
			}
			
			\foreach \x/\y/\xx/\yy in {
				0   /0.9 /0.3 /1,
				0.3 /0.8 /0.45/1,
				0.45/0.7 /0.8 /1,
				0.8 /0.6 /0.9 /1,
				0.9 /0.0 /1   /1		
			}{
				\draw[pattern = north west lines](\x*\w,\y*\h) rectangle (\xx*\w,\yy*\h);
			}	
			
			\end{tikzpicture}	
			\subcaption{The hatched areas are empty boxes, that can be used to place small items.}
		\end{subfigure}
		
		\caption{Configurations before and after removing the overlapping items and reducing the width to integrals.}
		\label{fig:VerticalItemsInsideTheConfigurations}
	\end{figure}
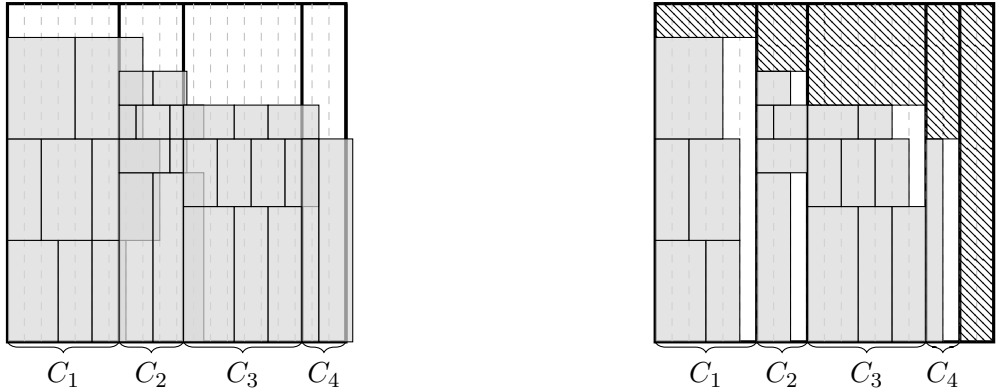
	
	We partition the set of overlapping items in each configuration into $7$ boxes with height $\nicefrac{1}{4} H$ and width $\mu W$ in the following way: 
	First, we stack the items in four boxes one by one on top of each other such that the last item overlaps the box on top. 
	Since the total height of the items is at most $H$, there are at most three overlapping items. 
	Each of them is placed into their own box. 
	We call the set of these boxes $\iBoxP'$.
	In total, we generate at most $7(|H_{\itemsV}|+|\iBoxP|)$ boxes of width $\mu W$.
	The items can be placed non-fractionally inside these boxes since they have a width of at most $\mu W$. 
	
	Note that the configuration width defined by the considered basic solution of the linear program might not be integral. 
	However, we can reduce the configuration width to the next smaller integer since we have removed all the overlapping items and hence only need an integral width.
	As a result we might get an empty configuration inside the strip, which has at least the width of the sum of all non integral fractions we removed from the configurations in the box. 
	This empty configuration has an integral width since the box has an integral width and all the other configurations have an integral width as well.
	
	Since the configurations have a height of at most $H$ and each item has a height of at least $\delta \OPT$, each configuration contains at most $H/(\delta \OPT) \in \Oh(1/\delta)$ items.
	Therefore, the set of boxes $\iBoxP \cup \iBoxP'$ is divided into at most $2(Y+|\iBoxP|)H/(\delta \OPT) \in \Oh((|H_{\itemsV}|+|\iBoxP|)/\delta)$ \subboxes containing only vertical items of the same height.
	
	Consider a configuration $C \in \mathcal{C}_B$ which has a non-zero entry $X_{C,B}$ in the considered solution. 
	Above this configuration there is a free area of height $\iH{B} - \iH{C}$ and width $X_{C,B}$ inside the box $B$, see Figure~\ref{fig:VerticalItemsInsideTheConfigurations}.
	Furthermore, in each box there might be a new empty configuration, which generates an empty box as well.
	Let $\iBoxS^{\mathcal{V}}$ be the set of these boxes. 
	There are at most $\Oh(|H_{\itemsV}|+|\iBoxP|)$; at most one above each configuration and one extra for each box.
	Since the configurations use exactly the area of the vertical items, the total area of these empty boxes has to be $\areaI{\iBoxS^{\mathcal{V}}} \geq  \areaI{\iBoxP} - \areaI{\itemsV}$.
\end{proof}


Consider the rounded optimal packing that is partitioned into the \subboxes by the first partitioning step in Lemma~\ref{lma:firstPartition}. 
This packing has a height of at most $(1+2\eps)\OPT$. 
We will rearrange the items inside this packing and partition the packing some further, such that the tall and vertical items are contained in boxes, that only contain items with the same height.

\begin{lemma} (Structure Lemma)
	\label{lma:structureLemma}
	By extending the packing height to $(5/4 + 5\eps)\OPT$ each rounded optimal packing can be rearranged and partitioned into $\mathcal{O}(1/(\delta^3\eps^5))$  boxes with the following properties:
	\begin{itemize}
		\item There are $|\mathcal{L}| + |\mathcal{M}_V| = \mathcal{O}(1/(\delta^2\eps))$ boxes $\iBoxL$ each containing exactly one item from the set $\mathcal{L} \cup \mathcal{M}_V$ and all items from this set are contained in these boxes.
		\item There are at most $\Oh(1/(\delta^2\eps)) $ boxes $\iBoxH$ 
		containing all horizontal items $\mathcal{H}$ with $\iBoxH\cap \iBoxL = \emptyset$. The horizontal items can overlap horizontal box borders, but never vertical box borders.
		\item There are at most $\mathcal{O}(1/(\delta^2\eps^5))$ boxes $\iBoxT$ containing tall items, such that each tall item $t$ is contained in a box with rounded height $h(t)$.
		\item There are at most $\mathcal{O}(1/(\delta^3\eps^5))$ boxes $\iBoxV$ containing vertical items, such that each vertical item $v$ is contained in a box with rounded height $h(v)$.
		\item There are at most $\Oh(1/(\delta^2\eps^5))$ boxes $\iBoxS$ for small items, such that the total area of these boxes combined with the total free area inside the horizontal boxes is at least as large as the total area of the small items.
		\item The lower and top border of each box is positioned at a multiple of $\eps\delta\OPT$.
	\end{itemize}
\end{lemma}

\begin{proof}
	In the following, we give a short overview of this proof. 
	We start with the partition from Lemma~\ref{lma:firstPartition} and define $H := (1+2\eps)\OPT$.
	Note that by this definition we have $H/4 \leq (1/4 +\eps)\OPT$ and thus each tall item has a height larger than $H/4$ as needed.	
	Since we already have seen how it is possible to reorder the items inside the boxes (see Section~\ref{sec:ReorderingInTheGeneralCase}), the main task in this proof is to find a place for the extra boxes for vertical items, which we need to place them integrally, see Lemma~\ref{lma:verticalItems}. 
	We consider three options to place these boxes. 
	First, we consider the widest tall items intersecting the horizontal line at $\nicefrac{1}{2} H$ and fix their position. 
	We aim to place the extra boxes on top of them if the total width of these items is large enough. 
	Otherwise, we know that all the tall items intersecting this line are very thin and we can find a way to place the extra boxes inside the boxes with height at least $\nicefrac{3}{4} H$, if the total width of these boxes is large enough. 
	The last option is to place them on top of the boxes with height between $\nicefrac{1}{2}H$ and $\nicefrac{3}{4} H$.
	
	Another task in this proof is to provide the condition assumed in Section~\ref{sec:ReorderingInTheGeneralCase}.
	Namely we have to ensure that the following conditions are provided:
	
	First, no box $B$ with height at least $\nicefrac{3}{4}H$ is allowed to be intersected at its border at the horizontal line $S(B)+\iH{B}-H/4$ by a tall item. 
	This can be done by introducing at most two further boxes of height at most $\nicefrac{3}{4}H$ per box $B$.
	
	Second, we need space above the tall boxes, to be able to extend them by $\nicefrac{1}{4} H$.
	Hence the next step is to shift up the boxes which have their lower border above $\nicefrac{3}{4}H$ by $\nicefrac{1}{4}H +\eps \OPT$. 
	We need the extra shift by $\eps  \OPT$ for technical reasons.
	
	Last, the tall and medium boxes have to start and end at the grid lines.  
	Since the tall items start and end at multiples of $\eps^2\OPT$, we choose these lines as the grid lines and change the start and endpoints of the tall and medium boxes accordingly at a small loss in the approximation ratio.
	
	When all these properties are fulfilled, we can apply Lemmas~\ref{lma:reorderingGeneral},~\ref{lma:reorderingMediumBoxes} and~\ref{lma:reorderingSmallBoxes} to reorder the items inside the boxes $\iBoxTV$.
	Afterward, we analyze the number of containers constructed for vertical items and find a place for the resulting set of additional containers, which we need by Lemma~\ref{lma:verticalItems} to place the vertical items non-fractional. 
	In the final step, we consider the boxes for horizontal and small items.
	
	\begin{stepList}
		\item[Fixing the position of the widest tall items intersecting $H/2$.]
		First, we look at the $1/(\delta^2\eps)$ widest tall items crossing the horizontal line at $\nicefrac{1}{2} H$. 
		We call the set of these items $T_{\nicefrac{1}{2} H}$.
		Each of these items defines a new unmovable item.  
		It splits the box containing it into three parts: The part left of this item, the part right of it and the part containing it. 
		The parts left and right will be reordered as any other box, while the part containing this item is reordered differently. 
		The item itself is not moved, while the part above and below have a height of less than $\nicefrac{1}{2} H$. 
		These parts define new boxes, which are small and hence can be reordered by Lemma~\ref{lma:reorderingSmallBoxes} such that they create at most $\Oh(N)$ \subboxes for tall and vertical items total. 
		These are less than the number of \subboxes created for one box of height larger than $\nicefrac{3}{4} H$. 
		Therefore, we can count this part as one box without making any error and assume that we add at most $2/(\delta^2\eps)$ boxes total.
		After this step, the total number of boxes containing both, tall and vertical items, is bounded by $\Oh(1/(\delta^2\eps))$.
		Furthermore, the number of vertical lines at box borders through the strip is bounded by 
		$\Oh(1/\delta^2\eps)$ as well.
		
		\item[Providing the conditions assumed for the reordering.]
		We have to provide three conditions: First, no item is allowed to overlap the tall box borders at the horizontal line at $S(B)+\iH{B}-H/4$; second, we need a gap of height $H/4$ between the upper border of each box of height at least $\nicefrac{3}{4} H$ as well as some extra free area above the medium sized boxes, to place the discarded pseudo items; third, the medium and tall boxes have to start and end at the grid lines. 
		
		\begin{description}
			\item[First condition: No overlapping at $S(B)+\iH{B}-H/4$.]
			To provide the first condition, we look at each box $B$ with height at least $\nicefrac{3}{4} H$, see Figure~\ref{fig:reorderingGeneral}. 
			Remember that in Lemma~\ref{lma:reorderingGeneral}, we had assumed that no tall item overlaps $B$'s left or right box border at $S(B)+\iH{B}-\nicefrac{1}{4}H$.
			We will establish this property by introducing two boxes for tall and vertical items of height less than $\nicefrac{3}{4} H$. 
			
			\begin{figure}
				\centering
				\begin{subfigure}[t]{0.38\textwidth}
					\centering
					\begin{tikzpicture}
					\pgfmathsetmacro{\w}{2.2}
					\pgfmathsetmacro{\h}{6}
					\pgfmathsetmacro{\hprime}{5}
					\draw[white] (-0.25*\w,0) rectangle (1.25*\w,5*\h/4);
					\drawVerticalItem{-0.25*\w}{0}{ 1.05*\w}{\hprime};
					
					\drawTallItem{-0.3*\w}{0.05*\hprime}{-0.22*\w}{0.345*\hprime};
					
					\foreach \x/\y/\xx/\yy/\z in {
						-0.27 / 0.36  / -0.21 / 0.62 /,
						-0.28 / 0.65  / 0.0   / 0.95 /$t$,
						-0.03 / 0.02  / 0.025 / 0.33 /,
						-0.025/ 0.35  / 0.05  / 0.625/,
						-0.21 / 0.38  /-0.175 / 0.64 /,
						-0.175/ 0.37  /-0.15  / 0.63 /,
						-0.15 / 0.34  /-0.1   / 0.60 /,
						-0.1  / 0.39  /-0.05  / 0.645/,
						-0.05 / 0.37  /-0.025 / 0.635/,
						-0.22 / 0.025 /-0.20  / 0.3  /,
						-0.175/ 0.01  /-0.15  / 0.28 /,
						-0.125/ 0.03  /-0.075 / 0.3  /,
						-0.03 / 0.020 / 0.025 / 0.33 /,
						-0.025/ 0.350 / 0.050 / 0.625/,			
						1.075 / 0.050 / 0.950 / 0.34 /,
						1.100 / 0.360 / 0.975 / 0.62 /,			
						0.200 / 0.975 / 0.225 / 0.345/,
						0.425 / 0.950 / 0.450 / 0.33 /,
						0.600 / 0.990 / 0.625 / 0.39 /,
						0.825 / 0.940 / 0.900 / 0.31 /,
						0.175 / 0.020 / 0.325 / 0.33 /,
						0.300 / 0.360 / 0.400 / 0.65 /,
						0.875 / 0.010 / 0.925 / 0.3  /,			
						0.000 / 0.700 / 0.025 / 0.99 /,
						0.025 / 0.650 / 0.075 / 0.96 /,
						0.075 / 0.660 / 0.125 / 0.95 /,
						0.125 / 0.690 / 0.200 / 1    /,			
						0.050 / 0.320 / 0.125 / 0.65 /,
						0.075 / 0.290 / 0.150 / 0.01 /,			
						0.225 / 0.650 / 0.250 / 0.99 /,
						0.250 / 0.630 / 0.275 / 0.97 /,
						0.275 / 0.670 / 0.325 / 1    /,
						0.325 / 0.660 / 0.350 / 1    /,
						0.350 / 0.680 / 0.375 / 1    /,
						0.375 / 0.670 / 0.425 / 1    /,			
						0.150 / 0.330 / 0.175 / 0.63 /,
						0.225 / 0.340 / 0.300 / 0.62 /,			
						0.350 / 0.000 / 0.450 / 0.29 /,
						0.475 / 0.030 / 0.575 / 0.32 /,
						0.600 / 0.010 / 0.700 / 0.36 /,
						0.750 / 0.000 / 0.850 / 0.3  /,			
						0.450 / 0.370 / 0.500 / 0.67 /,
						0.550 / 0.340 / 0.600 / 0.69 /,			
						0.475 / 0.690 / 0.575 / 1    /,			
						0.700 / 0.290 / 0.750 / 0.71 /,			
						0.900 / 0.440 / 0.975 / 0.78 /,
						1.075 / 0.650 / 1     / 0.95 /		
					}
					{
						\drawTallItem[\z]{\x*\w}{\y*\hprime}{\xx*\w}{\yy*\hprime};
					}
					
					\draw[dashed] (-0.35*\w,\hprime/2)node[left]{$\iH{B}/2$} --(1.25*\w,\hprime/2);
					\draw[dashed] (-0.35*\w,\h/4)node[left]{$H/4$} --(1.25*\w,\h/4);
					\draw[dashed] (-0.35*\w,\hprime- \h/4)node[left]{$\iH{B}-H/4$} --(1.25*\w,\hprime- \h/4);
					\draw[dashed] (-0.35*\w,\hprime)node[left]{$\iH{B}$} --(1.25*\w,\hprime);
					
					\end{tikzpicture}
					\caption{A box with items overlapping at $\iH{B}-H/4$.}
					\label{fig:sub:beforeChange}
				\end{subfigure}
				\hfill
				\begin{subfigure}[t]{0.3\textwidth}
					\centering
					\begin{tikzpicture}
					\pgfmathsetmacro{\w}{2.2}
					\pgfmathsetmacro{\h}{6}
					\pgfmathsetmacro{\hprime}{5}
					\draw[white, opacity = 0.9] (-0.25*\w,0) rectangle (1.25*\w,5*\h/4);
					\drawVerticalItem{-0.25*\w}{0}{ 1.05*\w}{\hprime};
					
					\foreach \x/\y/\xx/\yy/\z in {
						-0.27 / 0.36  / -0.21 / 0.62 /,
						-0.28 / 0.65  / 0.0   / 0.95 /$t$,
						-0.03 / 0.02  / 0.025 / 0.33 /,
						-0.025/ 0.35  / 0.05  / 0.625/,
						-0.21 / 0.38  /-0.175 / 0.64 /,
						-0.175/ 0.37  /-0.15  / 0.63 /,
						-0.15 / 0.34  /-0.1   / 0.60 /,
						-0.1  / 0.39  /-0.05  / 0.645/,
						-0.05 / 0.37  /-0.025 / 0.635/,
						-0.22 / 0.025 /-0.20  / 0.3  /,
						-0.175/ 0.01  /-0.15  / 0.28 /,
						-0.125/ 0.03  /-0.075 / 0.3  /,
						-0.03 / 0.020 / 0.025 / 0.33 /,
						-0.025/ 0.350 / 0.050 / 0.625/,			
						1.075 / 0.050 / 0.950 / 0.34 /,
						1.100 / 0.360 / 0.975 / 0.62 /,			
						0.200 / 0.975 / 0.225 / 0.345/,
						0.425 / 0.950 / 0.450 / 0.33 /,
						0.600 / 0.990 / 0.625 / 0.39 /,
						0.825 / 0.940 / 0.900 / 0.31 /,
						0.175 / 0.020 / 0.325 / 0.33 /,
						0.300 / 0.360 / 0.400 / 0.65 /,
						0.875 / 0.010 / 0.925 / 0.3  /,			
						0.000 / 0.700 / 0.025 / 0.99 /,
						0.025 / 0.650 / 0.075 / 0.96 /,
						0.075 / 0.660 / 0.125 / 0.95 /,
						0.125 / 0.690 / 0.200 / 1    /,			
						0.050 / 0.320 / 0.125 / 0.65 /,
						0.075 / 0.290 / 0.150 / 0.01 /,			
						0.225 / 0.650 / 0.250 / 0.99 /,
						0.250 / 0.630 / 0.275 / 0.97 /,
						0.275 / 0.670 / 0.325 / 1    /,
						0.325 / 0.660 / 0.350 / 1    /,
						0.350 / 0.680 / 0.375 / 1    /,
						0.375 / 0.670 / 0.425 / 1    /,			
						0.150 / 0.330 / 0.175 / 0.63 /,
						0.225 / 0.340 / 0.300 / 0.62 /,			
						0.350 / 0.000 / 0.450 / 0.29 /,
						0.475 / 0.030 / 0.575 / 0.32 /,
						0.600 / 0.010 / 0.700 / 0.36 /,
						0.750 / 0.000 / 0.850 / 0.3  /,			
						0.450 / 0.370 / 0.500 / 0.67 /,
						0.550 / 0.340 / 0.600 / 0.69 /,			
						0.475 / 0.690 / 0.575 / 1    /,			
						0.700 / 0.290 / 0.750 / 0.71 /,			
						0.900 / 0.440 / 0.975 / 0.78 /,
						1.075 / 0.650 / 1     / 0.95 /		
					}
					{
						\drawTallItem{\x*\w}{\y*\hprime}{\xx*\w}{\yy*\hprime};
					}
					
					\draw[pattern = north west lines] (-0.25*\w,0) rectangle (0*\w,0.65*\hprime);
					\draw[pattern = north west lines] (\w,0) rectangle (1.05*\w,0.65*\hprime);
					\draw[pattern = north west lines] (-0.25*\w,0.95*\hprime) rectangle (0*\w,\hprime);
					\draw[pattern = north west lines] (\w,0.95*\hprime) rectangle (1.05*\w,\hprime);
					
					\draw[red] (0*\w,0)-- (0*\w,\hprime);
					\draw[red] (1*\w,0)-- (1*\w,\hprime);
					
					\draw[dashed] (-0.35*\w,\hprime/2) --(1.25*\w,\hprime/2);
					\draw[dashed] (-0.35*\w,\h/4) --(1.25*\w,\h/4);
					\draw[dashed] (-0.35*\w,\hprime- \h/4) --(1.25*\w,\hprime- \h/4);
					\draw[dashed] (-0.35*\w,\hprime) --(1.25*\w,\hprime);
					
					\end{tikzpicture}
					\caption{The lines $L$ and $R$.}
					\label{fig:sub:LAndR}
				\end{subfigure}
				\hfill
				\begin{subfigure}[t]{0.3\textwidth}
					\centering
					\begin{tikzpicture}
					\pgfmathsetmacro{\w}{2.2}
					\pgfmathsetmacro{\h}{6}
					\pgfmathsetmacro{\hprime}{5}
					
					\draw[white, opacity = 0.9] (-0.25*\w,0) rectangle (1.25*\w,5*\h/4);
					
					\drawVerticalItem{0}{0}{\w}{\hprime};
					
					\foreach \x/\y/\xx/\yy in {
						-0.03 / 0.020 / 0.025 / 0.33,
						-0.025/ 0.350 / 0.050 / 0.625,			
						1.075 / 0.050 / 0.950 / 0.34,
						1.100 / 0.360 / 0.975 / 0.62,			
						0.200 / 0.975 / 0.225 / 0.345,
						0.425 / 0.950 / 0.450 / 0.33,
						0.600 / 0.990 / 0.625 / 0.39,
						0.825 / 0.940 / 0.900 / 0.31,
						0.175 / 0.020 / 0.325 / 0.33,
						0.300 / 0.360 / 0.400 / 0.65,
						0.875 / 0.010 / 0.925 / 0.3,			
						0.000 / 0.700 / 0.025 / 0.99,
						0.025 / 0.650 / 0.075 / 0.96,
						0.075 / 0.660 / 0.125 / 0.95,
						0.125 / 0.690 / 0.200 / 1,			
						0.050 / 0.320 / 0.125 / 0.65,
						0.075 / 0.290 / 0.150 / 0.01,			
						0.225 / 0.650 / 0.250 / 0.99,
						0.250 / 0.630 / 0.275 / 0.97,
						0.275 / 0.670 / 0.325 / 1,
						0.325 / 0.660 / 0.350 / 1,
						0.350 / 0.680 / 0.375 / 1,
						0.375 / 0.670 / 0.425 / 1,			
						0.150 / 0.330 / 0.175 / 0.63,
						0.225 / 0.340 / 0.300 / 0.62,			
						0.350 / 0.000 / 0.450 / 0.29,
						0.475 / 0.030 / 0.575 / 0.32,
						0.600 / 0.010 / 0.700 / 0.36,
						0.750 / 0.000 / 0.850 / 0.3,			
						0.450 / 0.370 / 0.500 / 0.67,
						0.550 / 0.340 / 0.600 / 0.69,			
						0.475 / 0.690 / 0.575 / 1,			
						0.700 / 0.290 / 0.750 / 0.71,			
						0.900 / 0.440 / 0.975 / 0.78			
					}
					{
						\drawTallItem{\x*\w}{\y*\hprime}{\xx*\w}{\yy*\hprime};
					}
					
					\draw[dashed] (-0.35*\w,\hprime/2) --(1.25*\w,\hprime/2);
					\draw[dashed] (-0.35*\w,\h/4) --(1.25*\w,\h/4);
					\draw[dashed] (-0.35*\w,\hprime- \h/4) --(1.25*\w,\hprime- \h/4);
					\draw[dashed] (-0.35*\w,\hprime) --(1.25*\w,\hprime);
					
					\end{tikzpicture}
					\caption{The box without overlappings at $\iH{B}-1/4 H$.}
					\label{fig:sub:AfterChange}
				\end{subfigure}					
				\caption{Eliminating overlaps at $\iH{B} - \nicefrac{1}{4} H$.}
				\label{fig:reorderingGeneral}
			\end{figure}
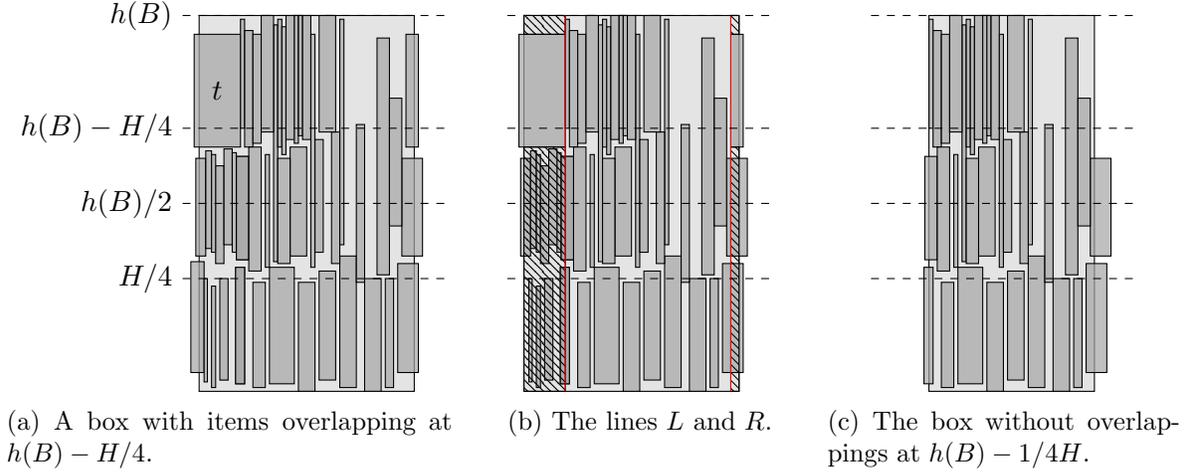
			
			Assume there is a tall item $t$ overlapping the left box border at $S(B)+ \iH{B} - \nicefrac{1}{4} H$, see Figure~\ref{fig:reorderingGeneral}. 
			We draw a vertical line $L$ at the right border of $t$ inside our box. 
			Tall items crossed by $L$ represent new unmovable items. 
			Obviously, $L$ is not intersected by a tall item at height $S(B) + \iH{B}- \nicefrac{1}{4} H$. 
			Consider the rectangular area between the left border of the box $B$ and $L$ bounded on top by $t$. 
			This area builds a new box for vertical and tall items with height less than $\nicefrac{3}{4} H$ and will later be reordered accordingly using Lemma~\ref{lma:reorderingMediumBoxes}.
			The rectangular area above $t$ between these vertical lines builds a pseudo item containing vertical items. 
			We repeat this step on the right side of the box. 
			
			
			In this step, we created for each of the tall boxes at most four new ones.
			Hence the number of boxes for tall and vertical items is still bounded by $\Oh(1/(\delta^2\eps))$. 
			This number, denoted as $\numBox$, will not increase in the following steps.
			Furthermore, the number of distinct vertical lines at each box border through the strip, denoted as $N_L$, is bounded by $\Oh(1/(\delta^2\eps))$.
			
			\item[Second condition: Free area above tall and medium boxes.]
			To ensure the second property, we draw a horizontal line at $\nicefrac{3}{4} H$ trough the strip and shift each box with lower border above or at this line exactly $\nicefrac{1}{4} H +\eps\OPT$ upwards. 
			We split each item that is overlapping the box border at the border during this shift.
			Note that no tall item is shifted since they start before $\nicefrac{3}{4} H$ and, thus, their boxes do too.	
			Hence the only items that will be split are vertical items (which are already sliced) and horizontal items, which we might slice horizontally.
			We fix this splitting in a later step. 
			After this shift, on top of each box with height at least $\nicefrac{3}{4} H$, there is a gap of height $\nicefrac{1}{4} H + \eps \OPT$ since these boxes end after $\nicefrac{3}{4} H$ and thus all the boxes above are shifted upwards. 
			
			Notice that we add an extra $\eps \OPT$ to the height. 
			In the later reordering of boxes with height at least $\nicefrac{3}{4} H$, we have shift the items crossing the line $\iH{B} -\nicefrac{1}{4} H$ exactly $\nicefrac{1}{4} H +\eps\OPT$ upwards as just $\nicefrac{1}{4} H$ like in the proof of Lemma~\ref{lma:reorderingGeneral} is not sufficient. 
			This extra height inside the tall boxes in $\iBoxTV$ is necessary to prove the existence of the gaps where we place the extra boxes for vertical items.
			
			Consider a box $B$ of height larger than $\nicefrac{1}{2} H$ and at most $\nicefrac{3}{4} H$. 
			By Lemma~\ref{lma:reorderingMediumBoxes}, we need an extra box with height $\nicefrac{1}{4} H$ and width $(1-(1/N))w(B) \leq (1-\eps^2)w(B)$ to rearrange the items in $B$.
			Due to the shifting, somewhere above this box, there is free area of height $\nicefrac{1}{4} H +\eps\OPT$ and width $w(B)$, which is possibly divided into several vertical slices. 
			Let us look at the free area above all the boxes with height between $\nicefrac{3}{4} H$ and $\nicefrac{1}{2} H$. 
			This free area is scattered into at most $N_L+1$ vertical pieces since there are at most $N_L$ vertical lines at box borders. 
			We allocate this free area above the boxes as contiguously as possible.
			For each piece of the free area we use, we introduce one box for vertical items (at most $N_L+1$).	
			Let $W_{1/2}$ be the total width of boxes with height larger than $\nicefrac{1}{2}H$ and at most $\nicefrac{2}{4}H$. The total width of the free area above these boxes, which we  have to use to place the pseudo items from inside the medium sized boxes is bounded by $(1-\eps^2)W_{1/2}$ and we have a total width of at least $\eps^2W_{1/2}$ to position the extra boxes needed to pack the vertical items non fractional.
			
			\item[Third condition: Alignment of tall and medium boxes to the grid-lines.]
			In Lemmas~\ref{lma:reorderingGeneral} and~\ref{lma:reorderingMediumBoxes} we assume, that each box with height larger than $\nicefrac{1}{2} H$ starts and ends at grid points. 
			In this step, we generate this property. 
			Grid lines be defined as the multiples of $\eps^2 \OPT$.
			Let $B$ be a box with height larger than $\nicefrac{1}{2} H$. 
			Look at the horizontal line $l$ at the smallest multiple of $\eps^2\OPT$ in this box. 
			The distance between $l$ and the bottom border is smaller than $\eps^2\OPT$. 
			In the box $B$, we will remove all the vertical items below and each item cut by $l$ and position them in an extra box at the end of the packing. 
			Since each item with height larger than $\eps\OPT$ starts and ends at multiples of $\eps^2\OPT$, the items cut by $l$ have a height of at most $\eps\OPT$. 
			We do the same on top of this box and for each other box. 
			We create above $\nicefrac{5}{4} H +\eps\OPT$ a box with height $2(\eps + \eps^2)\OPT$ and width $W$. 
			For each vertical line trough the strip, there is at most one box with height larger than $H/2$. 
			Hence, when shifting up these items such that they are positioned inside the new box while they maintain their relative positions, we do not provoke any overlapping.
			
		\end{description}
		
		\item[Reorder tall and vertical items inside the boxes.]
		After all necessary conditions are fulfilled, we apply the Lemmas~\ref{lma:reorderingGeneral}, \ref{lma:reorderingMediumBoxes}, and~\ref{lma:reorderingSmallBoxes} to reorder the items inside the boxes for tall and vertical items.
		Since we create the most \subboxes for tall boxes, we pessimistically assume that all the given boxes for tall and vertical items are tall, i.e., have a height larger than $\nicefrac{3}{4}H$.
		We create at most 
		$\Oh(1/\eps^4)$ \subboxes for tall and at most 
		$\Oh(1/\eps^4)$ \subboxes for vertical items per box for tall and vertical items; remember that $\NumGridpoints = \lceil (1+3\eps)/\eps^2 \rceil$. 
		To this point, we generate at most 
		$\Oh(1/(\delta^2\eps^5))$
		boxes for tall items in total.
		
		It is necessary to further divide the \subboxes inside the tall boxes in $\iBoxTV$ to enable the placement of the extra boxes for vertical items.
		Consider the boxes for tall and vertical items in $\iBoxTV$ that have a height larger than $\nicefrac{3}{4} H$. 
		In each of these boxes $B$, we draw a vertical line at the left border of each contained \subbox. 
		If a \subbox for vertical items inside $B$ is intersected by such a vertical line, we split the \subbox at this line. 
		Each of these lines intersects at most three boxes for vertical items since at each point there can be at most four boxes (for tall or vertical items) on top of each other inside $B$. 
		Hence, by splitting the vertical boxes this way, we introduce at most three new boxes for vertical items, per vertical line. 
		Since there are at most $\Oh(1/(\delta^2\eps^5))$ \subboxes for tall and vertical items, the number of vertical lines is bounded by $\Oh(1/(\delta^2\eps^5))$ as well. 
		And hence after the splitting the number of boxes for vertical items is still bounded by $\Oh(1/(\delta^2\eps^5))$.
		
		The area between two consecutive lines defines a strip, where the height of all the intersected boxes does not change. 
		We have at most 
		$\Oh(1/(\delta^2\eps^5))$ of these strips total. 
		We define $\numStrips$ 
		as the number of these strips.
		
		\item[Placing the extra boxes for vertical items.] 
		By Lemma~\ref{lma:verticalItems}, we need at most $\Oh(|H_{\itemsV}|+|\iBoxP|)$ additional boxes with height $\nicefrac{1}{4} H$ and width $\mu W$ to place the vertical items non-fractionally into the boxes, where $\iBoxP$ are the boxes for vertical items created so far. 
		We call the set of these additional boxes $\mathcal{B}_{\mu W}$.
		We can bound the variables in the following way. 
		There are at most $|\iBoxP| \in \Oh(1/(\delta^2\eps^5))$ boxes for vertical items and at most $|H_{\itemsV}| \leq 1/\delta\eps$ different heights of the items (which is a rather rough estimation). 
		Therefore, we need at most 
		$\numFrac \in \Oh(1/(\delta^2\eps^5))$ extra boxes $\mathcal{B}_{\mu W}$.
		
		We have to place the additional boxes inside the packing area $W \cdot \nicefrac{5}{4} H$. 
		In the following steps, we will prove that it is possible to place them by considering three possibilities. 
		Consider again the vertical lines at the box borders (not the \subbox borders). 
		These $N_L$ lines generate at most $N_L+1$ strips. 
		Let $W_T$ be the total width of the strips containing items from $T_{\nicefrac{1}{2} H}$, $W_H$ be the total width of the strips containing boxes with height at least $\nicefrac{3}{4} H$ and $W_R$ be the total width of all other strips. 
		In total we have $W_T + W_H + W_R = W$. 
		We can assume $\numBox \leq c_B/(\delta^2\eps)$, $\numStrips \leq c_S/(\delta^2\eps^5)$, $\numFrac \leq c_F/(\delta^2\eps^5)$ and $N_L \leq c_L/(\delta^2\eps)$ for some constants $c_B,c_S,c_F,c_L \in \NN$.
		At this point it is necessary to define the function $f$ to find the values $\delta$ and $\mu$ more precisely and we specify $f(\eps)$ by choosing $k \leq (4(c_B+c_F+c_L)c_S)$. 
		Hence it holds that $\mu \leq \delta^2\eps^{11}/(4(c_B+c_F+c_L)c_S)$.
		
		Consider the strips without boxes of height $\nicefrac{3}{4} H$ or the items in $T_{\nicefrac{1}{2} H}$. 
		These strips can contain boxes with height larger than $\nicefrac{1}{2} H$. 
		Therefore, we have free area with total width at least $\eps^2 W_R$ in these strips. 
		
		\begin{claim} 
			If $W_R \geq \eps^4 W$, we can place the $\numFrac$ boxes $\mathcal{B}_{\mu W}$ into these areas.
		\end{claim}
		\begin{proofClaim}
			The considered strips might contain boxes with height larger than $\nicefrac{1}{2}H$ and less than $\nicefrac{3}{4}H$. 
			Therefore, the free area in these strips will be partially used by the extra boxes for pseudo items for these boxes.
			Nevertheless, these strips contain free area with width at least $\eps^2W_R$ that we can use to place the extra boxes $\mathcal{B}_{\mu W}$, see Lemma~\ref{lma:reorderingMediumBoxes}.
			In each of these at most $(N_l +1)$ strips the free area is contiguous.
			However, we have to calculate a small error that might occur: 
			Each of the boxes in $\mathcal{B}_{\mu W}$ has a width of $\mu W$ and, therefore,  in each strip there is a residual width of up to $\mu W$ where we cannot place a box from the set $\mathcal{B}_{\mu W}$, see Figure~\ref{fig:WasteOfFreeArea}.
			
			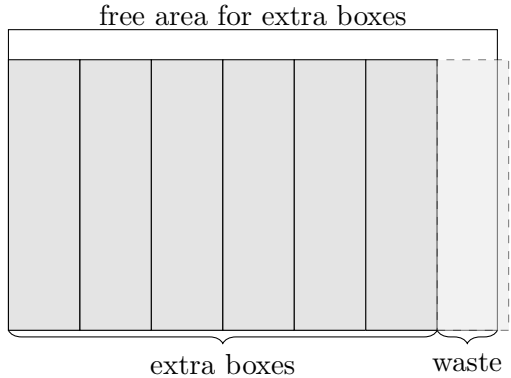
\begin{figure}[ht]
				\centering
				\begin{tikzpicture}
				\pgfmathsetmacro{\w}{0.5}
				\pgfmathsetmacro{\h}{4}
				\pgfmathsetmacro{\bW}{1.9}

				\draw (0,0) rectangle (13*\w,\h);
				\node[] at (6.5*\w,1.05*\h){free area for extra boxes};
				
				\foreach \x in {0,...,5}{
					\drawVerticalItem{\bW*\x*\w}{0}{\bW*\x*\w+\bW*\w}{0.9*\h};
				}
				\draw[opacity= 0.5, dashed,fill = white!90!black] (\bW*6*\w,0) rectangle (\bW*6*\w+\bW*\w,0.9*\h);
				
				\draw[decorate,decoration={brace,amplitude=5pt}] (6*\bW*\w,0*\h) -- (0*\w,0*\h) node[midway,below,yshift=-5pt,align=center]{extra boxes};
				\draw[decorate,decoration={brace,amplitude=5pt}] (13*\w,0*\h) -- (6*\bW*\w,0*\h) node[midway,below,yshift=-5pt,align=center]{waste};
				
				\end{tikzpicture}
				\caption{The waste of the free area, which can have a width of up to $\mu W$.}
				\label{fig:WasteOfFreeArea}
			\end{figure}
			
			On the positive side, we can use an area with total width of at least $\eps^2 W_R - (N_L +1) \mu W$ to place the boxes in $\mathcal{B}_{\mu W}$ since there are at most $N_L +1$ strips. 
			Therefore if $\eps^2W_R - (N_L +1)  \mu W \geq \numFrac \mu W$, we can place all the boxes. 
			Using $\mu := \delta^2\eps^{11}/(4(c_B+c_F+c_L)c_S)$, it holds that 
			\begin{align*}
			\numFrac \mu W +  (N_L +1) \mu W  
			= \mu W(c_F/\delta^2\eps^5 + c_L/\delta^2\eps +1) 
			\leq \eps^{11} W/\eps^5 
			\leq \eps^{6} W.
			\end{align*}
			Therefore, if $W_R \geq  \eps^4 W$, it holds that $\eps^2W_R - (N_L +1)  \mu W \geq \numFrac \mu W$ and we can place all the boxes $\mathcal{B}_{\mu W}$, which concludes the claim.
		\end{proofClaim}
		
		\begin{claim} 
			If $W_T \geq \eps^6 W/(4c_S)$, we can place the $\numFrac$ boxes $\mathcal{B}_{\mu W}$ in the strips containing the items in $T_{\nicefrac{1}{2} H}$.
		\end{claim}
		\begin{proofClaim}
			There are at most $N_L +1$ strips containing parts of the items in $T_{\nicefrac{1}{2} H}$. 
			In these strips the free area is contiguous and can be fully used since these strips do not contain boxes with height larger than $\nicefrac{1}{2}H$.
			Each box in $\mathcal{B}_{\mu W}$ has a width of exactly $\mu W$. 
			Hence, in each strip there is an area with width at most $\mu W$ which we cannot use to place the boxes. 
			Therefore, if $W_T - \mu W N_L \geq \numFrac \mu W$, we can place all the $\numFrac$ boxes into these strips. Using $\mu := \delta^2\eps^{11}/(4(c_B+c_F+c_L)c_S)$, it holds that 
			\begin{align*}
			\numFrac \mu W + \mu W N_L = \mu W(c_F/\delta^2\eps^5 + c_L/\delta^2\eps) \leq \eps^6 W/(4 c_S)
			\end{align*}
			Therefore, if $W_T \geq \eps^6 W/(4c_S)$, it holds that $W_T - \mu W N_L \geq \numFrac \mu W$ and we can place all the boxes $\mathcal{B}_{\mu W}$, which concludes the claim.
		\end{proofClaim}
		
		\begin{claim}
			If $W_T < \eps^6 W/(4c_S)$ and $W_R < \eps^4 W$, we can place all the boxes for vertical items inside the boxes of height at least $\nicefrac{3}{4} H$.
		\end{claim}
		
		\begin{proofClaim}
			In this case it holds that $W_H = W - (W_T+W_R) > (1 - 2 \eps^4)W \geq \eps W$. 
			Furthermore, each tall item not in $T_ {\nicefrac{1}{2} H}$ crossing $\nicefrac{1}{2} H$ has a width of at most $ W_T \cdot \delta^2\eps < \eps^7\delta^2 W/ 4 c_S := w_{\max}$.
			After the reordering in the boxes, there are at most $\numStrips$ strips in the boxes total. 
			We are interested in the total height of the free area inside a strip. 
			This area might be non-contiguous since there could occur an item in the middle of this strips and some free area above and below this item. 
			In the shifting step, we have added a total area of $W_H(\nicefrac{1}{4} H + \eps\OPT)$ to all of these strips. 
			Let $\check{W}_H$ be the total width of the strips containing free area with total height less than $\nicefrac{1}{4} H$ and let $\hat{W}_H$ be the total width of strips containing free area with height larger than $\nicefrac{1}{4} H$. 
			We want to use the strips containing free area of total height at least  $\nicefrac{1}{4} H$ to place the extra boxes.
			Therefore, we have to prove that these strips have a sufficient minimum total width; more precisely we prove the following remark:
			\begin{remark*}
				It holds that $\hat{W}_H \geq \eps W_H$.
			\end{remark*}	
			In each strip the total free area can have a height of at most $\nicefrac{3}{4} H +\eps\OPT$ since at the top and at the bottom there are always boxes with height at least $\nicefrac{1}{4} H$ or there has to be a box with height at least $\nicefrac{3}{4} H$ on the bottom. 
			It holds that $\check{W}_H + \hat{W}_H = W_H$. 
			Furthermore, it holds that 
			\[\nicefrac{1}{4} H \cdot \check{W}_H+ (\nicefrac{3}{4} H+\eps\OPT) \cdot \hat{W}_H \geq W_H(\nicefrac{1}{4} H + \eps\OPT)\]
			since the free area in $\check{W}_H$ has a total height of at most $\nicefrac{1}{4} H$ and the free area in $\hat{W}_H$ has a height of at most $(\nicefrac{3}{4} H+\eps\OPT)$ and the total free area is bounded by $W_H(\nicefrac{1}{4} H + \eps\OPT)$.
			As a consequence, we can prove that $\hat{W}_H$ has a sufficient minimum size.
			It holds that 
			\begin{align*}
			W_H(\nicefrac{1}{4} H + \eps\OPT) & \leq \nicefrac{1}{4} H \cdot \check{W}_H+ (\nicefrac{3}{4} H+\eps\OPT) \cdot \hat{W}_H\\
			& = \nicefrac{1}{4} H \cdot W_H+ (\nicefrac{1}{2} H + \eps\OPT) \cdot \hat{W}_H \\
			& =\nicefrac{1}{4} H  \cdot W_H+ ((1+2\eps)\OPT/2 + \eps\OPT) \cdot \hat{W}_H\\
			& =\nicefrac{1}{4} H  \cdot W_H+ ((1+4\eps)\OPT/2) \cdot \hat{W}_H,
			\end{align*}
			and, therefore, we can deduce
			\begin{align*}
			\eps W_H \leq ((1+4\eps)/2) \cdot \hat{W}_H.
			\end{align*}
			Thus, it holds that 
			\[\hat{W}_H \geq 2\eps W_H/(1+4\eps) \geq \eps W_H, \text{\ \ for } \eps \leq 1/4,\]
			which concludes the proof of the remark.
			
			Consequently, strips with total width of at least $\eps W_H$ contain free area with total height at least $\nicefrac{1}{4} H$.
			The free area in this strips can be scattered into at most two pieces. 
			We will fuse this free area by shifting the boxes for vertical or tall items. 
			Notice that we can shift the boxes for vertical items in each strip freely up and down since their box borders are at the strip borders by construction. 
			This is different for the \subboxes for tall items, which can be positioned between $\nicefrac{1}{2}\iH{B}$ and $\iH{B}-\nicefrac{1}{4} H$. 
			These \subboxes possibly contain tall items overlapping the strip's borders. 
			Remember that each tall item in this strip has a width of at most $w_{\max} = \eps^7\delta^2 W/ (4c_S)$.
			Hence, in each strip with width larger than $2 w_{\max} = \eps^7\delta^2 W/(2c_S)$, we can shift the middle part of these \subboxes such up or down such that the free area is connected. 
			We do not shift the \subboxes touching the bottom or the top of the box.
			In each strip, there is an area with width at most $\mu W$ which we cannot use to place the boxes. 
			Therefore, we can place all boxes for previously fractional vertical items, if 
			$\eps W_H - 2w_{max}\numStrips - \mu W \numStrips \geq \mu W \numFrac$.
			It holds that 
			\begin{align*}
			& 2w_{max}\numStrips + \mu W \numStrips + \mu W \numFrac \\
			& \leq (\eps^7\delta^2 W/(2c_S))(c_S/\delta^2\eps^5) +\mu W(c_S/\delta^2\eps^5 + c_F/\delta^2\eps^5) \\
			&\leq \eps^2W/2 + \eps^{6}W \leq \eps^2 W
			\end{align*}
			Thus, if $W_H \geq \eps W$, it holds that $\eps W_H - 2w_{max}\numStrips - \mu W \numStrips \geq \mu W \numFrac$ and we can place all boxes in this case which concludes the proof of this claim.
		\end{proofClaim}
		
		In this step, we create at most $2\numStrips \in \Oh(1/(\eps^5\delta^2))$ new boxes for tall items and no new box for vertical items. 
		The boxes for tall items already do contain just tall items with the same height. 
		hence, we introduce at most 
		$\Oh(1/(\eps^5\delta^2))$
		boxes for tall items in total. 	
		Furthermore, by Lemma~\ref{lma:verticalItems}, we create at most 
		$\Oh(1/(\eps^5\delta^3))$
		boxes for vertical items $\iBoxV$, such that each box $B\in \iBoxV$ contains just items with height $h(B)$.
		
	\end{stepList}
	\begin{description}
		\item[The boxes for small items.] 
		The free area inside the boxes from the partition in Lemma~\ref{lma:firstPartition} for horizontal, tall, and vertical items is at least as large as the total area for the small items since the small items where contained in the optimal packing area.
		
		\item[Bounding the packing height.] 
		Let us recapitulate what we added to the packing height during this process. 
		We started with a packing of height $\OPT$. 
		After the rounding of the items with height larger than $\delta$ and rounding the horizontal items, we received a packing with height $(1+2\eps)\OPT$. 
		With the shifting at the horizontal line $\nicefrac{3}{4}H$ we added $\nicefrac{1}{4}H +\eps\OPT \leq \nicefrac{1}{4}(1+2\eps)\OPT +\eps\OPT$ to the packing height. 
		Then, we shifted some vertical items to ensure that the boxes with height taller than $\nicefrac{1}{2} H$ start and end at multiples of $\eps^2\OPT$. 
		This added further $2(\eps + \eps^2)\OPT$ to the packing height. 
		In total we have added at most $\nicefrac{1}{4}(1+2\eps)\OPT + 2(\eps+\eps^2)\OPT \leq (1/4 + 3\eps)\OPT$ to the packing height (if $\eps \leq 1/2$) 
		such that the structured packing has a height of at most $(5/4 + 5\eps)\OPT \cdot W$.	
	\end{description}
\end{proof}

In the next step, we proof that there is an algorithm that can place the horizontal items inside their boxes.
This algorithm creates a constant number of sub boxes for small items.

\begin{lemma}
	\label{lma:horizontalItems}
	There is an algorithm with running time $(\log(1/\delta)/\eps)^{\Oh(1/\eps\delta^3)}$ that places the horizontal items into the boxes $\iBoxH$ and an extra box $B_H$ of height at most $\eps^{9}\OPT$ and width $W$.
	
	Furthermore, the algorithm creates at most $\Oh(1/\eps\delta^2)$ empty boxes $\iBoxS^{\itemsH}$ with total area $\areaI{\iBoxS^{\itemsH}} = \areaI{\iBoxH} - \areaI{\tilde{\itemsH}}$.
\end{lemma}
\begin{proof}
	The first step is to round the horizontal items.
	We stack horizontal items on top of each other ordered by their width, such that the widest item is positioned at the bottom. 
	This stack has a height of at most $\OPT/\delta$ since each item has a width of at least $\delta W$ and their total area is bounded by $\OPT\cdot W$. 
	We group the items in the stack to at most $1/\eps\delta^2$ groups, each of height $\eps\delta^2 \OPT/\delta = \eps\delta\OPT$ and round the items in the groups to the widest width occurring inside this group.
	This step reduces the number of different sizes to at most $1/\delta\eps^2$.
	The rounded horizontal items can be placed fractionally into the non-rounded items of the group containing the next larger items.
	The group containing the widest rounded items has to be placed on top of the packing.
	Therefore, the total height of items we put on the top of the packing has a height of at most $\delta\eps\OPT$.
	We define an extra box of width $W$ and height $\delta\eps\OPT$ for these items.	
	For simplicity of notation we assume in the following that $\iBoxH$ contains this extra box as well.	
	
	We place the rounded horizontal items into the boxes using a configuration LP.
	In this scenario, a configuration is a set of items that fit next to each other inside the boxes, i.e., a configuration $C$ is a multiset of the form $\sset{a_w:w}{w \in \mathcal{W}_{\itemsH}}$ and the width of a configuration is defined as $\iW{C} := \sum_{w \in \mathcal{W}_{\itemsH}}a_w w$.  
	Furthermore, $\conf_w$ denotes the set of configurations with width at most $w$, where $\mathcal{W}_{\itemsH}$ is the set of different width appearing in the set of rounded horizontal items $\itemsH$.
	Finally, we define $\iH{w}$ as the total height of all the items with width $w$.
	
	The set of configurations $\conf_{W}$ is bounded by $\Oh((\log(1/\delta)/\eps)^{1/\delta})$ because the items have a width of at least $\delta W$ and hence there can be at most $1/\delta$ items in each configuration.
	The following configuration LP is solvable since the rounded horizontal items fit fractionally into the boxes $\iBoxH$.
	\begin{align*}
	\sum_{C \in \mathcal{C}_{\iW{B}}} X_{C,B} & = \iH{B} & \forall B \in \mathcal{B}_{\mathcal{H}}\\
	\sum_{B \in \mathcal{B}_{\mathcal{H}}} \sum_{C \in \mathcal{C}_{\iW{B}}} X_{C,B} a_{w,C} & = \iH{w} & \forall  w \in \mathcal{W}_{\itemsH}\\
	X_{C,B} & \geq 0 & \forall  B \in \mathcal{B}_{\mathcal{H}}, C \in \mathcal{C}_{\iW{B}}
	\end{align*}
	
	We can solve this linear program by guessing the at most $|\mathcal{W}_{\itemsH}| + |\mathcal{B}_{\mathcal{H}}| = \Oh(1/(\eps\delta^2))$ non-zero entries of the basic solution and solve the resulting equality system using the Gauß-Jordan-Elimination.
	We use the first solution we find where all the variables are non-negative.
	Such a solution can be found in at most $\Oh(|\mathcal{C}_{W}|^{|\mathcal{W}_{\itemsH}| + |\mathcal{B}_{\mathcal{H}}|}\cdot (|\mathcal{W}_{\itemsH}| + |\mathcal{B}_{\mathcal{H}}|)^3) \leq (\log(1/\delta)/\eps)^{\Oh(1/\eps\delta^3)}$ operations since the configuration LP has to be solvable for the correct partition.
	
	We place the corresponding configurations into the corresponding boxes and place the original horizontal items greedily into the configurations, such that the last item overlaps the configuration border. 
	We place the original items one by one inside an area reserved by the configurations for their rounded counterparts until an item overlaps this area on the top.
	Then we proceed to the next area.
	Since the total processing time of these parts is exactly as large as the total processing time of the items with this rounded width, there are enough parts to place all them. 	
	
	In the next step, we remove the overlapping items and place them on top of the box. 
	Each of these removed items has a height of at most $\mu\OPT$. 
	We add at most $\eps^{10}\OPT$ to the packing height by shifting the overlapping items to the top of the packing, because first, a basic solution has at most $\Oh(1/(\delta^2\eps))$ configurations; second, all the items in one configuration can be placed next to each other; and third, $\mu \leq \delta^2\eps^{11}/k$ for a suitable large constant $k$.
	Together with the extra box that we need due to the rounding, the total added height is bounded by $\eps^{10}\OPT + \delta\eps\OPT \leq \eps^{9}\OPT$.
	
	Similar as before, we can reduce the height of each configuration to the next smaller integer since the horizontal items have an integral height. 
	This introduces at most one new configuration per box, i.e., the one which is empty.
	In each box $B$ to the right of each (used) configuration $C$ there might be some free area of width $\iW{B}-\iW{C}$ and height $X_{C,B}$. 
	This area defines one of the empty boxes $\iBoxS^{\itemsH}$.
	Since there are at most $|\mathcal{W}_{\itemsH}| + |\mathcal{B}_{\mathcal{H}}|$ configurations and at most $|\mathcal{B}_{\mathcal{H}}|$ boxes for horizontal items, we introduce at most $\Oh(1/\eps\delta^2)$ empty boxes $\iBoxS^{\itemsH}$.
	Furthermore, their total area has to be at least as large as $\areaI{\iBoxS^{\itemsH}} = \areaI{\iBoxH} - \areaI{\tilde{\itemsH}}$ since the configurations contain exactly the total area of the rounded horizontal items.  
	
	Let us now consider the boxes for horizontal items which we create in this step. 
	Each configuration contains at most $1/\delta$ positions for items. 
	For each of these positions we create one box that has the rounded width of the items for these positions and (integral) height that is the sum of all the heights of the items positioned inside this box and we create one additional box for the shifted item.
	Hence we introduced at most $\Oh(1/(\eps\delta^3))$ boxes for horizontal items, which only contain items with the same rounded width.
	
\end{proof}

\begin{lemma}
	\label{lma:smallItems}
	It is possible to place the small items inside the boxes generated by Lemma~\ref{lma:structureLemma} and the boxes generated by Lemma~\ref{lma:horizontalItems} and one extra box with width $W$ and height at most $2\eps^6 \OPT$.
\end{lemma}
\begin{proof}
	The free area inside the boxes from the partition form Lemma~\ref{lma:firstPartition} for horizontal tall and vertical items is at least as large as the total area for the small items since the small items where contained in the optimal packing area.
	By Lemmas~\ref{lma:verticalItems} and~\ref{lma:horizontalItems} we generate at most $\Oh(1/(\eps^5\delta^2))$ empty boxes, which we can use to place the small items.
	These boxes have a total area that is at least as large as the empty space in the original boxes $\iBoxH$ and $\iBoxTV$ from Lemma~\ref{lma:firstPartition}.
	Hence the total area of these empty boxes is at least as large as the area of the small items.
	
	Let $\iBoxS$ be the set of boxes and $|\iBoxS| = c/\delta^2\eps^5$ for some constant $c \in \NN$.
	We prove that we only need a small extra box to place all the items with the NFDH algorithm into these boxes. 
	
	First, we discard any box with height less than $\mu\OPT$ or width less than $\mu W$. 
	The total area of each discarded box is at most $\mu W\OPT$. 
	Let us consider a box $B$ with height and width larger than $\mu\OPT$ or $\mu W$ respectively. 
	In each shelf we use for the NFDH-Algorithm, we cannot use a total width of at most $\mu W$ to place the items. 
	Furthermore, the last shelf has a distance of at most $\mu\OPT$ to the upper border of the box. 
	Additionally, the free area between the shelfs has a total area of at most $\mu\OPT\cdot w(B)$. 
	Therefore, the total free area in $B$ is at most $\mu W \cdot h(B) + 2 \mu \OPT\cdot w(B)\leq 3 \mu W\OPT$. 
	As a result, the total area of items that could not be placed inside the boxes is at most $3 \mu W\OPT \cdot c/\delta^2\eps^5$. 
	Since $\mu \leq \eps^{11}\delta^2/k$ for some suitable constant $k$ it holds that $3 \mu W\OPT \cdot c/\delta^2\eps^5 \leq \eps^{6} W\OPT$, when choosing $k \geq c$.
	
	These items can be place with Steinberg's algorithm~\cite{Steinberg97} into a box with width $W$ and height $2\eps^6 \OPT$ since they have a height of at most $\mu \OPT$.
\end{proof}

In the last step, we prove that it is possible to place the medium sized items $\itemsM$.

\begin{lemma}
	\label{lma:mediumItems}	
	It is possible to place the medium items $\itemsM$ into a box width width $W$ and height at most $2\eps \OPT$
\end{lemma}
\begin{proof}
	First, we sort them by their processing time. 
	Afterward, we use the NFDH algorithm to place the jobs. 
	We know that $\areaI{\itemsM}$ is bounded by $\eps^{13}W \OPT$ and $p_{\max}$ is bounded by $\eps \OPT$.
	Therefore, by Lemma~\ref{thm:NFDH}, we can place these items with a packing height of at most  $\mathrm{NFDH}(\itemsM) \leq 2\eps^{13} \OPT + \eps \OPT \leq 2 \eps \OPT$.
	
\end{proof}

\section{Algorithms}
\label{sec:Algorithms}
In this section, we describe the three algorithms for Strip Packing without rotations, Strip Packing with rotations and Scheduling contiguous moldable jobs. Each optimal solution of these three problems can be rearranged, such that the structure looks like the structure in Lemma~\ref{lma:structureLemma}.
The algorithms all work roughly the same. 
First, we determine an upper bound for the approximation. 
Afterward, we use a binary search framework to find a $(5/4 +\eps)$ approximation. 
The routine called by the framework guesses the structure of the packing and tests with a dynamic program if the guess is feasible.

\subsection{Strip Packing without Rotations} 
\label{sec:StripPackingWithoutRotations}

The steps of the algorithm can be summarized as follows:
\begin{enumerate}[noitemsep]
	\item \label{algEnum1:LowerAndUpperBounds} Define $\eps' := 1/\lceil 10/\eps \rceil$, a lower bound $T := \max\{\areaI{\items}/W, \max\sset{\iH{i}}{i \in \items\}}$ and an upper bound $2T$ for the approximation and use these bounds to round and scale the item heights to values in $\{1, \dots, n/\eps'\}$. As a result we gain that $\OPT_{scaled}$ is an integer in $\{n/\eps',  \dots, 2n/\eps' +n\}$.
	\item \label{algEnum1:BinarySearch} Try values $T' \in \{n/\eps',  \dots, 2n/\eps' +n\}$ for the optimum in a binary search fashion and for each tested value perform the following steps.
	\item \label{algEnum1:SimplificationSteps} For $\OPT = T'$, $\eps = \eps'$, and the rounded and scaled instance perform the simplification steps as described in Section~\ref{sec:StructurResult}. More precisely: find the correct values for $\delta$ and $\mu$, and round the heights of all items with a height taller than $\delta \OPT$ with the techniques from Lemma~\ref{lma:rounding}, round the horizontal items using linear grouping as described in Lemma~\ref{lma:horizontalItems}.
	\item \label{algEnum1:PartitioningStep} Try each possible partition of the area $W \times (5/4+5\eps')\OPT$ into $\Oh(1/(\delta^3\eps'^5))$ boxes from Lemma~\ref{lma:structureLemma}.
	For each of these partitions perform the following steps:
	\item \label{algEnum1:DynamicProgram} Using a dynamic program try to place the vertical and horizontal items inside their boxes. If this fails the partition must have been wrong and we try the next partition.
	\item \label{algEnum1:PlacingResidualItems} If the vertical and tall items could be placed inside the boxes, try to place the horizontal items using the algorithm from Lemma~\ref{lma:horizontalItems}. If this fails discard the guessed partition. Otherwise save the packing and try the next smaller value for $T'$ in binary search fashion.
	\item \label{algEnum1:BinarySearchEnd} If all the partitions into boxes fail, try the next larger value for $T'$ in binary search fashion. 
	\item Finally (after the binary search for the correct $T'$) place the small items inside their boxes using NFDH and place the medium sized items on top of the packing using Steinberg using the best packing found. Return the packing.
\end{enumerate}
In the following we discuss the correctness of these steps.
\begin{description}
	\item[Step~\ref{algEnum1:LowerAndUpperBounds}: A first rounding step.]
	Given a value $\eps \in (0,1]$ and an instance $I$, we define $\eps':= \min\{1/4, 1/\lceil 10/\eps \rceil\}$ and use this $\eps'$ instead of $\eps$ in the following steps. 
	We estimate the optimal packing height $\OPT$ by using that Steinberg's algorithm~\cite{Steinberg97} can place all items into a strip with height of at most $2T := 2\max\{\areaI{\items}/W, \max\sset{\iH{i}}{i \in \items\}}$. 
	Therefore, we know that $\OPT \in [T,2T]$. 
	
	Next we round the heights of the items arithmetically by introducing a small error.
	\begin{lemma}
		\label{lma:arithmeticRoundingOfHeights}
		Let $T$ be a lower bound of the optimal packing height.
		With an additive loss of at most $\eps T$ in the approximation ratio, we can assume that each job has a processing time that is a multiple of $\eps T/n$.
	\end{lemma}
	\begin{proof}
		For this assumption, we have to add to the height of each item a height of at most $\eps T/n$. 
		Since there are at most $n$ items and these items are placed on top of each other in the worst case, this adds at most $\eps T$ to the optimal packing height of the considered instance.
	\end{proof}
	
	Note that after this rounding there are at most $n/\eps$ different sizes of the items, which are all multiples of $\eps T/n$. 
	We define the height of the optimal packing for this rounded instance as $\OPT_{rounded}$
	Without making an additional rounding error, we can assume that the items have a height in $\{1,2,\dots, n/\eps\}$, by scaling the rounded heights with $n/(T\eps)$ and scale them back, when constructing the packing.  
	Using this scaling, we know that $\OPT_{scaled} \in \{n/\eps, n/\eps +1, \dots, 2n/\eps + n\}$ since $T \leq \OPT \leq \OPT_{rounded} \leq \OPT +\eps T \leq 2T +\eps T$ and hence $n/\eps \leq \OPT_{scaled} \leq 2n/\eps + n$.
	
	\item[Step~\ref{algEnum1:BinarySearch}: Dual Approximation]
	To find the $(5/4+\eps)\OPT$ approximation we use the dual approximation framework introduced by Hochbaum and Shmoys~\cite{HochbaumS87}. 
	Given a value $T$, we can calculate a packing with height at most $(5/4+\eps)T$ or decide that there is no packing with height $T$. 
	To find the smallest value for $T$ where it is possible to find a packing with height $(5/4+\eps)T$ we can try all the appropriate values for $T$ in $\Oh(n/\eps)$.
	The other option is to try these values in a binary search fashion, such that we are searching for the smallest value from this set such that the residual algorithm finds a packing. 
	This search search for the optimal $T$ can then be done in $\Oh(\log(n/\eps))$ using the  dual approximation framework.

	\item[Step~\ref{algEnum1:SimplificationSteps}: Performing the simplification steps.] 
	Next, we compute the values $\delta$ and $\mu$ with the properties from Lemma~\ref{lma:DeltaMu} while assuming $\OPT = T$ for some $T \in \{n/\eps, \dots, 2n/\eps + n\}$. 
	Knowing these values, we partition and round the items accordingly, see Lemma~\ref{lma:rounding}. 
	As a trick to maintain integer sizes for the item heights, we can scale the instance with $1/(\eps\delta) \in \NN$ before rounding the items with Lemma~\ref{lma:rounding}. 
	In this way $\eps^x T_{scaled}$ will be integral for each $x \in \{1, \dots, \log_{\eps}{1/(\eps\delta)}\}$ since $T_{scaled}$ has the form $T/\eps\delta$.
	We will write $T$ instead of $T_{scaled}$ in the following.
	
	\item[Step~\ref{algEnum1:PartitioningStep}: Guessing the partition into boxes]
	Afterward, we guess the structure of the transformed optimal solution via Lemma~\ref{lma:structureLemma} for items in $\itemsL, \itemsM, \itemsT$ and $\itemsV$ and the boxes from Lemma~\ref{lma:firstPartition} for the horizontal items $\itemsH$.
	This partition into boxes has a height of at most $(\frac{5}{4} + 5\eps)(1+\eps)T$. 
	For each of the at most $\mathcal{O}(1/(\delta^3\eps^5))$ boxes, we guess the lower left corner and the upper right corner. 
	For each box there are at most $\mathcal{O}((W/\delta\eps)^2)$ possibilities to guess these positions since the x-coordinates are in $\{0,\dots,W-1\}$, and the y-coordinates are in $\{0,\dots, \mathcal{O}(1/\delta\eps)\}$. 
	Therefore, there are at most $(W/\delta\eps)^{\mathcal{O}(1/\delta^3\eps^5)}$ possible guessing steps. 
	A guessed structure is feasible if we can place the items into the corresponding boxes.
	
	\item[Step~\ref{algEnum1:DynamicProgram}: The dynamic program]
	For each of the guessed partitions into boxes, we test if we can place the items in $\mathcal{V}\cup\mathcal{T}$ inside these boxes by using the following dynamic program:
	For each rounded height $h$ we generate a vector $(w_{h,1},\dots,w_{h,k_{h}})$. 
	It represents the $k_{h}$ boxes for items $i \in \mathcal{V}\cup \mathcal{T}$ with height $\iH{i} = h$. 
	Each entry is bounded by the width of the corresponding box. 
	
	For each rounded height $h$ the program enumerates all items $i \in \mathcal{V}\cup \mathcal{T}$ with height $\iH{i} = h$.
	We start with the vector $(w_{h,1}=0,\dots,w_{h,k_{h}}=0)$. 
	For each item $i \in \mathcal{T}\cup \mathcal{V}$ with $\iH{i} = h$, we make $k_{h}$ copies of each so far generated vector $(w_{h,1},\dots,w_{h,k_{h}})$ and add the value $\iW{i}$ to a different component in each copy.
	If we enlarge one component above its maximal value or we get the same vector a second time, we discard this vector. 
	The guess of boxes for items with height $h$ was feasible if there is still a valid vector after enumerating all the items with height $h$.
	
	The width of each box for tall and vertical items $\mathcal{T} \cup \mathcal{V}$ is bounded by $W$. 
	Therefore, we enumerate at most $W^{\mathcal{O}(k)}$ vectors for each item, where $k$ is the total number of boxes.  
	Therefore, the dynamic program has a running time of at most $n\cdot W^{\mathcal{O}(1/\delta^3\eps^5)}$. 
	
	\item[Step~\ref{algEnum1:PlacingResidualItems}: Placing residual items]
	If it is possible to place the items in $\mathcal{L}\cup\mathcal{V}\cup\mathcal{T}\cup\mathcal{M}_{\mathcal{V}}$, we place the horizontal items with the algorithm described in the proof of Lemma~\ref{lma:horizontalItems} in $(\log(1/\delta)/\eps)^{\Oh(1/\eps\delta^3)}$. 
	After that we check of the total area of the boxes generated for the small items is large enough, i.e., larger than the total area of small items.
	This step fails inf that is not the case.
	
	The small $\itemsS$ and medium items $\itemsM$ are placed in the final step after the binary search for the right $T'$. 
	The small items are placed into their corresponding boxes with the NFDH algorithm in $\Oh(n \log (n))$ and the medium items on top of the packing using the NFDH-Algorithm from~\cite{CoffmanGJT80} 
	These items add at most $2\eps \OPT$ to the packing height since each item in $\itemsM$ has a height of at most $\eps \OPT$ and the total area of these items is bounded by $(\eps^{13}/k)TW$. 
	
	\item[Break condition]
	If one of the steps to place the items, i.e., Step~\ref{algEnum1:DynamicProgram} or Step~\ref{algEnum1:PlacingResidualItems}, fails for a guessed partition, the guess must have been wrong and we try the next partition.
	If we ca not find a packing for any of the partitions, the value $T'$ was to low and we try the next value for $T'$. 
	
	When we scale back the heights of the items to multiples of $\eps T /\eps$ the final packing has a height of at most $(5/4+5\eps)\OPT_{rounded} + \eps^9\OPT_{rounded} + 2\eps^6\OPT_{rounded} + 2\eps\OPT_{rounded} \leq (5/4+8\eps)\OPT_{rounded}$.
	Since $\OPT_{rounded} \leq \OPT + \eps T \leq (1+\eps)\OPT$ it holds that the schedule we can find has a height of at most $(5/4 + 8\eps)(1+\eps)\OPT \leq (5/4 + 10\eps)\OPT$.
	Hence, if we use the value $\eps'$ instead of $\eps$ for these steps, the generated packing has a height of at most $(5/4 + 10\eps')\OPT \leq (5/4 +\eps)\OPT$.
	
	Using $\delta \geq \eps^{2^{\Oh(1/\eps^{13})}}$, the running time of the algorithm can be bounded by 
	\begin{align*}
	\mathcal{O}(n\log(n))&+\mathcal{O}(\log(n/\eps))\cdot (W/\delta\eps)^{\mathcal{O}(1/\delta^3\eps^5)} \cdot (n\cdot W^{\mathcal{O}(1/\delta^3\eps^5)}+(\log(1/\delta)/\eps)^{\Oh(1/\eps\delta^3)})\\
	&\leq \mathcal{O}(n\log(n))\cdot W^{1/\eps^{2^{\Oh(1/\eps^{13})}}}\\
	&\leq \Oh(n\log(n)) \cdot W^{\Oh_{\eps}(1)}.
	\end{align*}
	
\end{description}
\subsection{Strip Packing with Rotations}
\label{sec:SPwithRotations}

In this scenario each item either can be positioned non rotated (rotation = 0 degrees) or it can be placed rotated by 90 degrees (rotation = 90 degrees).
However, the items in an optimal solution can be rounded shifted and reordered as the items in an optimal packing for Strip Packing without rotations and hence the structural Lemma~\ref{lma:structureLemma} holds for optimal solutions to this problem as well.
Nevertheless, we run into several problems when we try to use the same algorithm for these instances as for the instances for Strip Packing without rotations.
Since we do not know which side of the items will be its width and which side will be its height in the optimal solution, we can no longer round the height of the item as we did before. 
Instead for each item, we will save several rounded values, i.e., both rounded heights and both rounded widths depending on the rotation of the item. 
Furthermore, the partition of the item set is no longer this simple, because an item can belong to one set in one rotation and to another in the other rotation. 
We will handle this issue by leaving this decision to an underlying dynamic program.  

For simplicity of notation we will assume that $\iH{i} \leq W$ and $\iW{i} \leq W$ as well. 
Note that in theory one of these values could be larger than $W$ but in that case we cannot rotate the item.
However this would only simplify the matter since, as described before, we then can decide in which set this item is contained and and round the height of this item etc.

\begin{description}
	\item[Dual approximation] 
	Similar as in the algorithm without rotations, we use a binary search framework to find the packing. 
	We know that $$T:= \max\{\areaI{\items}, h_{\max} := \max\sset{\min\{\iH{i},\iW{i}\}} {i \in \items}\} $$
	is a lower bound on the packing height, while $2T$ is an upper bound because of Lemma~\ref{lma:Steinberg}.
	(Note that this bound has to be slightly adapted, if there are non rotatable items. When considering a non rotatable item, the definition of $h_{\max}$ needs to take into account the maximum of height and width.)
	Therefore, given an optimal solution we would like to round the heights of the items to multiples of $\eps T/n$ as before in Lemma~\ref{lma:arithmeticRoundingOfHeights}.
	Since we do not know which side of the item defines the height, we calculate the rounded (and scaled) values for both sides, but remember the original one as well.
	Using these rounded values for the heights of the items lengthens the schedule by at most $\eps T$ and again we can assume that the optimal height of the schedule is one of the values $\{n/\eps, n/\eps +1, ..., 2n/\eps + n\}$.

	\item[Defining $\delta$ and $\mu$]
	In the next step, we guess the values of $\delta$ and $\mu$ since, as discussed, we cannot determinate them because we do not know which item will be in which set in the optimal solution. 
	There are at most $1/f(\eps) = \Oh(1/\eps^{13})$ possibilities for these values. 
	Knowing these values, we can decide for a given item and its rotation (0 degrees or 90 degrees) in which set $\mathcal{L}, \mathcal{V}, \mathcal{H},\mathcal{T}, \mathcal{M},\mathcal{M}_{V}$ or $\mathcal{S}$ this item is contained in the optimal solution. 
	
	\item[Rounding the horizontal items]
	In the next step, we consider the horizontal items.  
	We want to round the horizontal items similar as in Lemma~\ref{lma:horizontalItems}. 
	However since we do not now in advance which items the set $\mathcal{H}$ contains, we have to guess the rounded widths:
	First, we guess the total height $h(\mathcal{H})$. 
	It holds that $h(\mathcal{H}) \in \{0,1,\dots, 2n/\eps + n\}$ since by the rounding and scaling each horizontal item has an integral height and the total packing height is bounded by $2n/\eps + n$.
	Therefore, we need at most $\Oh(n/\eps)$ guessing steps to determine the correct height $h_{\itemsH}$ of the stack of horizontal items.
	
	After this guess of the height of the stack, we guess the at most $1/\eps\delta^2$ items and their rotation ($0$ degrees or $90$ degrees) that define the rounded widths of the groups.
	There are at most $(2n)^{1/\eps\delta^2}$ possibilities to guess these items, due to the rotation. 
	We determine the height of a group as $h_{G} := \lfloor \eps\delta^2 h_{\itemsH} \rfloor$.
	Using this height there is at most one item per stack, that cannot be placed inside the group since all the items have an integral height.
	Therefore, we place the items which we have guessed to define the widths of the groups on top of the packing in the very last step of the algorithm. 
	Since these items have a height of at most $\mu \OPT$ and since $\mu \leq \delta^2\eps^{13}$, these items add at most $\mu\OPT/\eps\delta^2 \leq \eps^{12}\OPT$ to the packing height.
	Now we can assume in the dynamic program, that each rounded width occurs with a total height of at most $h_{G} \in \Oh(\delta^2 n)$.

	
	\item[Rounding the height of items with height larger than $\delta \OPT$]
	In the following let $\tilde{h}(i)$ as well as $\tilde{w}(i)$ be the rounded height and rounded width for an item for a given rotation ($0$ degrees or $90$ degrees). 
	If an item is contained in the set $\itemsL\cup \itemsV\cup \itemsT\cup \itemsMV$, its height is a multiple of $i\eps^x\OPT$ for $i \in \{1/\eps, \dots,1/\eps^2\}$ and some $x \in \NN$, while its width remains original. 
	If the item is in $\mathcal{H} \cup \mathcal{M}\cup \mathcal{S}$, its height is an integer from $\{1, \dots, n/\eps\}$ and its width is either its original or one of the at most $1/\eps\delta^2$ guessed width values.  
	Note that if we scale all the heights with $1/\eps\delta$ all the considered heights are integral.
	
	\item[Guessing the partition from Lemma~\ref{lma:structureLemma}]
	In the next step, we guess the structure described in Lemma~\ref{lma:structureLemma} using a height of at most $(\frac{5}{4} +5\eps)\OPT_{scaled}$. 
	First, we guess which items are the at most $\mathcal{O}(1/\delta^2\eps)$ large $\mathcal{L}$ and medium vertical $\mathcal{M}_{\mathcal{V}}$ items and their rotations in the optimal packing. 
	This can be done in $n^{\Oh(1/\delta^2\eps)}$.
	Afterward, we guess their positions, i.e., the position of the lower left corner of these items. 
	For the $x$ position there are at most $W$ possibilities, while for the $y$-position there are at most $\Oh(1/\eps\delta)$ possibilities. 
	Hence we can guess the positions of all the large and medium vertical items in $(W/\eps\delta)^{\Oh(1/\eps\delta^2)}$.
	Next, we guess the positions of the at most $\Oh(1/\eps^5\delta^3)$ other boxes using $(W/\eps\delta)^{\Oh(1/\eps^5\delta^3)}$ possibilities since they start and end at multiples of $1/\eps\delta$
	
	\item[The modified dynamic program]
	Again we check with a dynamic program if the guess is feasible. 
	We introduce the vectors for the boxes for vertical and tall items $w_{h_i} = (w_{h_i,1},\dots, w_{h_i,k_{h_i}})$ for each rounded height $h_i$ as before and introduce a new vector $h = (h_{w_1},\dots, h_{w_{1/\eps\delta^2}})$ for each rounded width $w_j$. 
	Furthermore, we introduce two values $a_s$ and $a_m$. 
	These values represent the total area of the small items $\mathcal{S}$ and medium sized items $\mathcal{M}$ receptively.
	In the dynamic program, we consider a sequence of sets $D_i$, $i \in \mathcal{I}$, containing vectors of the form $(h_{w_1},\dots, h_{w_{1/\eps\delta^2}},w_{h_1},\dots w_{h_{1/\delta\eps}},a_s,a_m)$. 
	$D_0$ contains just the vector filled with zeros.
	Iterating over the set of items for each item $i \in I$, we determine for both possible rotations the set it would be contained in. 
	For both rotations, we do the following steps to the set $D_{i-1}$.
	\begin{itemize}[noitemsep]
		\item If the $i$th item is in $\mathcal{V} \cup \mathcal{T}$, we make $k_{\tilde{h}(i)}$ copies of each vector in $D_{i-1}$. 
		In each copy, we add the items width $w(i)$ to another entry of the vector $w_{\tilde{h}(i)}$ and add it to the set $D_i$. 
		\item If the $i$th item is in $\mathcal{H}$, we make a copy of each vector in $D_{i-1}$. 
		In each copy, we add its height $\tilde{h}(i)$ to a the entry $h_{\tilde{w}(i)}$ and add it to the set $D_i$. 
		\item If the $i$th item is in $\mathcal{S}$, it holds that $\tilde{h}(i) = l_i \eps(1+\eps)T/n$ for some $l_i \in \{1,\dots, n/\eps\}$. 
		We make a copy of each vector in the set $D_{i-1}$ and add $l_i \cdot w(i)$ to the value $a_s$ in each vector and add it to the set $D_i$. 
		\item If the $i$th item is in $\mathcal{M}$, we do the same as in the previous case with the difference that we add the value $l_i \cdot w(i)$ to $a_m$.
	\end{itemize}
	If a value in the vector exceeds its boundary, we discard this vector. 
	Furthermore, if a specific vector is created a second time, we save it just once and discard the newly generated vector.

	The values $h_{i,j}$ are bounded by the height of the rounded group, i.e., they are bounded by $\Oh(\delta^2n)$, while the values $w_{i,j}$ are bounded by the guessed width of the corresponding box, i.e., ultimately they are bounded by $W$. 
	Let $A_s$ be the total area of the boxes for small items. 
	Note that since each small item has an integer width and an integer height and the total area of items is bounded by $W\cdot n/\eps$ the total area of small items $a_s$ can have any integer size up to $W n/\eps$.
	Furthermore, the total area of the medium sized items is bounded by $\eps^{13}W\OPT$, where $\OPT \in \{n/\eps, \dots, 2n/\eps+n\}$.
	Thus, we bound the size of $a_m$ by this value.
	
	Let us analyze the running time of the linear program. 
	The entries $w_{i,j}$ can take at most $W$ different values, the entries $h_{i,j}$ can take at most $\delta^2n$ different values and the entries $a_s$ and $a_n$ can take at most $\mathcal{O}(nW/\eps)$ different values. 
	Since each vector in $D$ has at most $\mathcal{O}(1/\eps^5\delta^3)$ entries, the running time of the dynamic program is bounded by $(W)^{\mathcal{O}(1/\eps^5\delta^3)}\cdot (\delta^2n)^{\mathcal{O}(1/\eps\delta^2)}\cdot \mathcal{O}((nW/\eps)^2)$. 
	
	\item[Placing the horizontal items]
	After the processing of the dynamic program, we try for each feasible solution of this program  whether it is possible to place the horizontal and small items.
	The horizontal items are placed with the algorithm from Lemma~\ref{lma:horizontalItems}.
	If the resulting set of boxes for small items is large enough, we save this solution and proceed with trying the next smaller value for $T'$.
	
	When we have found the correct value for $T'$ we place the small and medium sized items using the NFDH-algorithm \cite{CoffmanGJT80}. 
	The overall running time of the algorithm dominated by the running time of the dynamic program and hence is bounded by $(Wn)^{1/\eps^{2^{\mathcal{O}(1/\eps^{13})}}}$.
	
\end{description}

\subsection{Scheduling contiguous moldable tasks}
Again we start with a binary search framework. 
We use the results from Ludwig and Tiwari~\cite{LudwigT94} to find an estimate $U$ for the makespan of the optimal schedule and define $T := U/2$, i.e., we know that $\OPT \in [T,2T]$.
Afterward, we use the same rounding and scaling as before, i.e., we consider only processing times in $\{1, \dots, n/\eps\}$.
Then we use the binary search framework to find the correct value for $T'$. 
The guessing steps work the same as in Section~\ref{sec:SPwithRotations} including the guessing of rounded width of horizontal jobs. 
In the following, we describe how to adjust the dynamic program for this scheduling version. 

Let $\psi_j(p) \in M_j$ be the minimal number of processors needed for job $j \in \jobs$ to have a processing time of at most $p$. 
We iterate over the non-placed jobs in arbitrary order. 
Let $j \in \jobs$. 
To generate the set $D_j$ in the dynamic program we do the following steps:
\begin{itemize}[noitemsep]
	\item First, we determine for each of the at most $1/\eps^2$ large processing times $p > \nicefrac{1}{4}\OPT$ the number of needed processors $\psi_j(p)$. 
	If this number is smaller than $\delta W$, it is feasible to schedule this job as tall job and we try each possible box for this job and processing time in the dynamic program. 
	Otherwise, the job cannot have this processing time for the current choice of large and medium sized items. 
	\item Afterward, we determine for each of the at most $1/\eps\delta$ large processing times $p$ with $\nicefrac{1}{4}\OPT \geq p > \delta\OPT$ the number of needed processors $\psi_j(p)$.
	If this number is smaller than $\mu W$, it is feasible to schedule this job as vertical job and we try each possible box for this job and processing time in the dynamic program.
	Otherwise, the job ca not have this processing time for the current choice of large and medium sized items.
	\item Next, we consider the guessed number of processors for horizontal jobs. 
	For two consecutive rounded numbers of machines $m_i$, $m_{i+1}$, we determine for which number of machines in $M_j \cap \{m_i, \dots, m_{i+1}\}$ the job has the smallest processing time $p_j(m_i)$. 
	If it is smaller than $\mu \OPT$ the job qualifies to be scheduled as a horizontal job.
	Hence, we make a copy of each vector in $D_{j-1}$ and add the value $\lceil n p_j(m_i)/\eps\OPT \rceil$ to the value $h_{m_{i+1}}$.
	\item After that, we try to schedule the job as a small job. 
	We determine the smallest processing time if the job uses less than $\mu W$ machines. 
	If this processing time is smaller than $\mu (1+\eps)^2T$ the job can be scheduled as a small job and we try this possibility too, by adding its work to $a_s$ to each vector from the set $D_{j-1}$. 
	\item Last, we test if the job can be scheduled as a medium job. We determine the smallest processing time if the job uses between $\mu W$ and $\delta W$ processors. 
	If this processing time is smaller than $\eps T$ the job can be scheduled as medium job and we add its work to a copy of each vector from the set $D_{j-1}$.
	On the other hand, the job can be scheduled as medium job, if its processing time is between $\mu (1+\eps)^2T$ and $\delta (1+\eps)^2 T$. 
	hence, we determine the minimal number of processors, such that the job has a processing time between $\mu (1+\eps)^2T$ and $\delta (1+\eps)^2 T$ and we add its work to a copy of each vector from the set $D_{j-1}$.
\end{itemize}

Each vector in the dynamic program has at most $\mathcal{O}(1/\eps^5\delta^3)$ entries. 
The values of the entries are bounded by $W$, $\delta^2 n$ and $nW/\eps$. 
For each job we try at most $\mathcal{O}(1/\eps^5\delta^3)$ entries. 
Therefore, the running time of the dynamic program and the entire algorithm can be bounded by $(Wn)^{1/\eps^{2^{\mathcal{O}(1/\eps^{13})}}}$. 

\section{Conclusion}
In this paper, we have nearly closed the gap between the lower bound of the approximation ratio and best approximation ratio for the problems pseudo-polynomial Strip Packing with and without rotations and the Contiguous Moldable Task Scheduling. 

Still open remains the question whether we actually can find algorithms with approximation ratio exactly $5/4$.
Concerning polynomial algorithms, there is still a large gap between the lower bound for an absolute approximation ratio of $3/2$ unless $\Poly = \NP$ and $5/3+\eps$ which is the best absolute approximation ratio achieved so far~\cite{HarrenJPS14}.
Furthermore, an interesting question is whether we can find better approximations for the case of monotonic moldable jobs.  
While the lower bound of $5/4$ holds for the general case of scheduling contiguous moldable jobs in pseudo-polynomial time, a PTAS could be possible if we consider monotonic jobs.

\bibliography{lowerBound}

\end{document}